\documentclass[%
 nofootinbib,
superscriptaddress,
 amsmath,amssymb,
 aps,
]{revtex4-2}

\usepackage{graphicx}
\usepackage{dcolumn}
\usepackage{bm}

\usepackage{xcolor} 
\usepackage{amsthm} 
\usepackage{tensor} 
\usepackage{amsmath}
\usepackage{amssymb}
\usepackage[colorlinks=true, allcolors=blue]{hyperref} 
\usepackage[bottom,stable]{footmisc} 
\usepackage{verbatim}
\usepackage{mathtools}

\usepackage[mathscr]{euscript}

\newcommand{\horeq}{~\hat{=}~} 

\newtheorem{theorem}{Theorem}

\newtheorem{definition}{Definition}

\newtheorem{lemma}{Lemma}

\newcommand{\be}{\begin{equation}}
\newcommand{\ee}{\end{equation}}

\theoremstyle{remark}
\newtheorem*{remark}{Remark}

\allowdisplaybreaks

\begin{document}

\preprint{APS/123-QED}

\title{The Entropy of Dynamical Black Holes}

\author{Stefan Hollands}
\email{stefan.hollands@uni-leipzig.de}
 
\affiliation{
Institute for Theoretical Physics\\
Leipzig University \\
Br\" uderstrasse 16, D-04103 Leipzig
}%
\affiliation{
Max-Planck-Institute MiS,\\
Inselstrasse 22, D-04103 Leipzig
}%

\author{Robert M. Wald}
\email{rmwa@uchicago.edu}
\author{Victor G. Zhang}%
 \email{victorzhang@uchicago.edu}
\affiliation{%
 Enrico Fermi Institute and Department of Physics\\
 University of Chicago \\
 933 E. 56th St., Chicago, IL 60637
}%

\date{\today}

\begin{abstract}
We propose a new formula for the entropy of a dynamical black hole---valid to leading order for perturbations off of a stationary black hole background---in an arbitrary classical diffeomorphism covariant Lagrangian theory of gravity in $n$ dimensions. In stationary eras, this formula agrees with the usual Noether charge formula, but in nonstationary eras, we obtain a nontrivial correction term. In particular, in general relativity, our formula for the entropy of a dynamical black hole differs from the standard Bekenstein--Hawking formula $A/4$ by a term involving the integral of the expansion of the null generators of the horizon. We show that, to leading perturbative order, our dynamical entropy in general relativity is equal to $1/4$ of the area of the apparent horizon. Our formula for entropy in a general theory of gravity is obtained from the requirement that a ``local physical process version'' of the first law of black hole thermodynamics hold for perturbations of a stationary black hole. It follows immediately that for first order perturbations sourced by external matter that satisfies the null energy condition, our entropy obeys the second law of black hole thermodynamics. For vacuum perturbations, the leading order change in entropy occurs at second order in perturbation theory, and the second law is obeyed at leading order if and only if the ``modified canonical energy flux'' is positive (as is the case in general relativity but presumably would not hold in more general theories of gravity). Our formula for the entropy of a dynamical black hole differs from a formula proposed independently by Dong and by Wall. We obtain the general relationship between their formula and ours. We then consider the generalized second law in semiclassical gravity for first order perturbations of a stationary black hole. We show that the validity of the quantum null energy condition (QNEC) on a Killing horizon is equivalent to the generalized second law using our notion of black hole entropy but using a modified notion of von Neumann entropy for matter. On the other hand, the generalized second law for the Dong--Wall entropy is equivalent to an integrated version of QNEC, using the unmodified von Neumann entropy for the entropy of matter.

\end{abstract}

\maketitle


\section{\label{sec:intro}Introduction}

The discoveries that black holes obey laws of thermodynamics \cite{Bardeen_1973}, radiate thermally \cite{Hawking_1975}, and should be assigned an entropy \cite{Bekenstein_1973} provide some of the deepest insights we presently have on the fundamental nature of black holes within a quantum theory of gravity. A formula for the entropy of a stationary black hole in an arbitrary theory of gravity obtained from a diffeomorphism covariant Lagrangian was obtained in \cite{IyerWald1994}. However, the derivation of \cite{IyerWald1994} required evaluation of the entropy on the bifurcation surface, $\mathcal B$, of the black hole, thus restricting the validity of the formula for entropy to stationary black holes and their linear perturbations, evaluated at the ``time'' represented by $\mathcal B$. It is of considerable interest to obtain an expression for the entropy of a non-stationary black hole in a general theory of gravity at a ``time'' represented by an arbitrary cross-section $\mathcal C$, since this would allow one to investigate whether, classically, black hole entropy satisfies a second law (i.e., whether it is non-decreasing with time) and whether, semiclassically, black hole entropy satisfies a generalized second law (i.e., whether the sum of black hole entropy and a matter contribution to entropy is non-decreasing).
A formula for dynamical black hole entropy in a general theory of gravity on an arbitrary cross-section $\mathcal C$ was proposed in \cite{IyerWald1994}. However, that proposed dynamical entropy formula is not field redefinition invariant, and the authors retracted their proposal in a ``note in proof'' in the published version of \cite{IyerWald1994}. Subsequently, proposed prescriptions for dynamical black hole entropy were given independently by Dong \cite{Dong2013} and Wall \cite{Wall2015}. The Dong and Wall entropy formulas agree in the cases where both have been evaluated and are expected to correspond generally,
although this has not been proven. We will refer to their proposal for dynamical black hole entropy as the Dong--Wall entropy. 

The main aim of this paper is to apply a new strategy to the definition of dynamical black hole entropy, based upon the validity of a local, ``physical process version'' of the first law of black hole mechanics. Since this law is valid only at leading nontrivial order in perturbation theory about a stationary black hole, our resulting definition of dynamical black hole entropy, $S[\mathcal C]$, on an arbitrary cross-section, $\mathcal C$, of the black hole horizon is intended to be applied only at leading nontrivial order in perturbation theory about a stationary black hole. Our formula does not agree with the Dong--Wall entropy, but, as we will elucidate in section \ref{dwent}, there is a close relationship between our formula and theirs. 
Furthermore, in the case of general relativity, we obtain a nontrivial dynamical correction to the Bekenstein--Hawking entropy formula; namely we obtain 
\be
S[\mathcal C] = \frac{A[\mathcal C]}{4} - \frac{1}{4} \int_{\mathcal C} V \vartheta
\label{degr}
\ee
where $A[\mathcal C]$ is the area of the cross-section $\mathcal C$ (so $A[\mathcal C]/4$ is the usual Bekenstein--Hawking entropy), $V$ is an affine parameter of the null generators of the horizon (with $V=0$ corresponding to the bifurcation surface $\mathcal B$), and $\vartheta$ is the expansion of these generators with respect to this affine parametrization.

In the remainder of this introductory section, we shall give an outline of our strategy that leads to our new dynamical entropy formula and explain some of the key features of our formula. To begin, we recall that the derivation of \cite{IyerWald1994} of the entropy of a stationary black hole follows from a ``fundamental identity'' (see \eqref{eq:fund_id} below) that holds in an arbitrary theory of gravity obtained from a diffeomorphism covariant Lagrangian. For source-free perturbations of a stationary black hole whose event horizon is a Killing horizon, the fundamental identity reduces to
\be
d[\delta \textbf{Q}[\xi] - \xi \cdot \boldsymbol{\theta}(\phi,\delta\phi)] = 0.
\label{fundidbh}
\ee
Here $\xi^a$ denotes the horizon Killing field, $\textbf{Q}[\xi]$ is the Noether charge associated with $\xi$, $\boldsymbol{\theta}$ denotes the symplectic potential, and $\phi$ collectively denotes all of the dynamical fields (i.e., the metric together with any matter fields appearing in the Lagrangian). Integration of this equation over a spacelike hypersurface $\Sigma$ that extends from the bifurcation surface $\mathcal B$ of the black hole to infinity yields
\be
\int_{\mathcal B} \delta \textbf{Q}[\xi]  = \delta M - \Omega_{\mathcal H} \delta J.
\ee
Here the right side is the boundary term from infinity resulting from integration of (\ref{fundidbh}) and it involves the mass, $M$, and angular momentum, $J$, at spatial infinity (with $\Omega_{\mathcal H}$ being angular velocity of the horizon, which enters the expression when one expresses $\xi^a$ in terms of asymptotic time translations and rotations). The left side is the boundary term at $\mathcal B$; very importantly, there is no contribution to this term from $\xi \cdot \boldsymbol{\theta}$ because $\xi^a = 0$ on $\mathcal B$. We define
\be
S \equiv \frac{2 \pi}{\kappa} \int_{\mathcal B} \textbf{Q}[\xi]
\label{bhent}
\ee
where $\kappa$ denotes the surface gravity of the black hole (so that $\kappa/2 \pi$ is the Hawking temperature of the black hole). It can be shown that $S$ is a local, geometrical quantity \cite{IyerWald1994}. We thereby obtain the first law of black hole mechanics for arbitrary source-free perturbations of a stationary black hole
\be
\frac{\kappa}{2 \pi} \delta S = \delta M - \Omega_{\mathcal H} \delta J
\label{1stlaw}.
\ee
This formula gives $S$ the interpretation of representing the entropy of the black hole, valid to first order for perturbations of a stationary black hole.

However, as already mentioned above, \eqref{bhent} has the deficiency that the horizon boundary term is evaluated at the bifurcation surface, $\mathcal B$, and, thus, $S$ has the interpretation of representing the entropy, $S[\mathcal B]$, of the black hole only at the ``time'' corresponding to $\mathcal B$. This deficiency is not of importance for a stationary black hole, since the entropy should be ``time independent,'' i.e., independent of the cross-section of the horizon of the black hole at which it is evaluated. However, for a non-stationary black hole, the entropy should evolve with time, and one would like to have an expression for the entropy, $S[\mathcal C]$, on an arbitrary cross-section, $\mathcal C$.

An expression for $S[\mathcal C]$ that would satisfy the first law of black hole mechanics would be obtained if we could define $S[\mathcal C]$ so that
\be
\delta S[\mathcal C] = \frac{2 \pi}{\kappa} \int_{\mathcal C} \delta \textbf{Q}[\xi] - \xi \cdot \boldsymbol{\theta}(\phi,\delta\phi). 
\ee 
If, on the event horizon $\mathcal H$, the pullback, $\underline{\boldsymbol{\theta}}$, of the symplectic potential $\boldsymbol{\theta}$ were of the form $\underline{\boldsymbol{\theta}} = \delta \textbf{B}_{\mathcal H}$
for some quantity $\textbf{B}_{\mathcal H}$ defined on $\mathcal H$, then the desired quantity $S[\mathcal C]$ could be defined by
\be
S[\mathcal C] = \frac{2 \pi}{\kappa} \int_{\mathcal C} \left(\textbf{Q}[\xi] - \xi \cdot \textbf{B}_{\mathcal H} \right) .
\label{dynent}
\ee
However, it is not possible for $\underline{\boldsymbol{\theta}}$ to be of the form of a total variation $ \delta \textbf{B}_{\mathcal H}$ on $\mathcal H$ in general because such a form would imply that the symplectic current flux through the horizon would vanish identically (since the symplectic current is the antisymmetrized second variation of $\boldsymbol{\theta}$), which is not the case. Indeed, it was for this reason that \cite{IyerWald1994} defined the entropy via \eqref{bhent} only on the bifurcation surface $\mathcal B$.

A similar issue arises for the definition of ADM mass, $M$, at spatial infinity with respect to an asymptotic time translation $t^a$. One wishes to define $M$ so that it is the Hamiltonian conjugate to $t^a$, i.e., so that it generates asymptotic time translations on phase space. This leads to the requirement that under asymptotically flat perturbations of any asymptotically flat background solution, we have \cite{IyerWald1994}
\be
\delta M = \int_{\infty} \delta \textbf{Q}[t] - t \cdot \boldsymbol{\theta}(\phi,\delta\phi) 
\label{deltaM}
\ee
where ``$\infty$'' denotes that the integral is taken over a sphere that approaches spatial infinity. However, in this case, as one approaches spatial infinity, the symplectic current goes to zero sufficiently rapidly that the symplectic current flux vanishes. Thus, there is no obstruction to finding a quantity $\textbf{B}_\infty$ such that to leading asymptotic order
\be
\boldsymbol{\theta} = \delta \textbf{B}_\infty
\label{Bspinf}
\ee
and such a $\textbf{B}_\infty$ can be explicitly found \cite{IyerWald1994}.
One may then define the ADM mass as
\be
M = \int_{\infty} \left(\textbf{Q}[t] - t \cdot \textbf{B}_\infty \right) 
\ee
and, thus, there is no difficulty in defining a quantity $M$ that satisfies the desired relation \eqref{deltaM}.

The problem of defining dynamical black hole entropy is much more analogous to the problem of defining the Bondi mass, $M_{\rm B}$, at null infinity $\mathcal I^+$. Here, the symplectic current flux does not, in general, vanish at null infinity, so an analog of (\ref{Bspinf}) cannot hold in general. Nevertheless, it was found in \cite{Wald2000} that one could find another symplectic potential 
$\boldsymbol{\theta}'(\phi,\delta\phi)$ for the pullback of the symplectic current to $\mathcal I^+$ with the property that $\boldsymbol{\theta}' = 0$ in a stationary background $\phi$. The authors of \cite{Wald2000} then defined $\textbf{B}_{\mathcal I^+}$ by 
\be
\underline{\boldsymbol{\theta}} - \boldsymbol{\theta}' = \delta \textbf{B}_{\mathcal I^+}
\label{wzreq}
\ee
where $\underline{\boldsymbol{\theta}}$ denotes the pullback of $\boldsymbol{\theta}$ to $\mathcal I^+$.
They then defined the Bondi mass relative to an asymptotic time translation $t^a$ on a cross-section, $\mathcal C$, of ${\mathcal I^+}$ by
\be
M_{\rm B} [\mathcal C] = \int_{\mathcal C} \left(\textbf{Q}[t] - t \cdot \textbf{B}_{\mathcal I^+} \right). 
\ee
In this case, $M_{\rm B}$ does not satisfy (\ref{deltaM}) in general---since, in general, there is no solution to this equation---but it does satisfy (\ref{deltaM}) for perturbations off of a stationary background spacetime.

In this paper, we will apply the strategy of \cite{Wald2000} to the definition of dynamical black hole entropy. We will prove that for first order perturbations of a stationary black hole, there exists a quantity $\textbf{B}_{\mathcal H}$ defined on the black hole horizon that satisfies\footnote{Note that this is equivalent to the Wald--Zoupas requirement for the existence of a $\boldsymbol{\theta}'$ that satisfies (\ref{wzreq}) and vanishes for a stationary background, since for a non-stationary background, one could define $\boldsymbol{\theta}' = \underline{\boldsymbol{\theta}} - \delta \textbf{B}_{\mathcal H}$.}
\be
\underline{\boldsymbol{\theta}} \horeq \delta \textbf{B}_{\mathcal H}
\label{thetaB}
\ee
where $\underline{\boldsymbol{\theta}}$ denotes the pullback of $\boldsymbol{\theta}$ to ${\mathcal H}$ and ``$\horeq$'' denotes that the equality holds only on the horizon $\mathcal H$.
This enables us to define dynamical black hole entropy via (\ref{dynent}). For first order vacuum perturbations of stationary black hole, it follows immediately from integration of (\ref{fundidbh}) over a hypersurface that extends from an arbitrary cross-section $\mathcal C$ of the horizon to spatial infinity that
\be
\frac{\kappa}{2 \pi} \delta S[\mathcal C] = \delta M - \Omega_{{\mathcal H}} \delta J, 
\label{1stlaw2}
\ee
i.e., the first law of black hole mechanics holds for an arbitrary cross-section. However, it also follows that $\delta S[\mathcal C] = \delta S[\mathcal B]$, i.e., for vacuum perturbations of a stationary black hole, the entropy is ``time independent'' to first order and is equal to the entropy of \cite{IyerWald1994}. In order to obtain a nontrivial time variation of black hole entropy, we must either work to second order in perturbation theory in the vacuum case or allow an external stress-energy, $\delta T_{ab}$, to be present in the first order perturbation. As will be discussed in subsection \ref{vacpert}, the results in the vacuum case at second order are closely analogous to the first order results with an external stress-energy, with the perturbed matter stress-energy $\delta T_{ab}$ replaced by the second order ``modified canonical energy'' \cite{Hollands_Wald_2013} of the gravitational field. For simplicity, we describe only the first order results with an external stress-energy in the discussion below.

When an external stress-energy $\delta T_{ab}$ is present, the fundamental identity (\ref{fundidbh}) on the black hole horizon becomes
\be
d[\delta \textbf{Q}[\xi] - \xi \cdot \boldsymbol{\theta}(\phi,\delta\phi)] = -\xi^a \delta \boldsymbol{\textbf{C}}_a
\label{fundidbh2}
\ee
where, restoring the form indices, $\delta \boldsymbol{\textbf{C}}_a$ is given by
\be
\delta {{C}}_{aa_1 \cdots a_{n-1}} = \delta T_{ae} {{\epsilon}^e}_{a_1 \cdots a_{n-1}}
\ee
where ${\epsilon}_{a_1 \cdots a_{n}}$ is the spacetime volume $n$-form. The change in the dynamical entropy with ``time'' (i.e., its dependence on the cross-section $\mathcal C$) can be seen by integrating (\ref{fundidbh2}) over the portion of the horizon, $\mathcal{H}_{12}$, bounded by cross-sections $ \mathcal C_1$ and $ \mathcal C_2$, where we take $\mathcal C_2$ to lie to the future of $\mathcal C_1$. Using $\underline{\boldsymbol{\theta}} = \delta \textbf{B}_{\mathcal H}$, we obtain
\be
\frac{\kappa}{2 \pi} \left(\delta S[\mathcal C_2] - \delta S[\mathcal C_1] \right)= \int_{\mathcal{H}_{12}} - \xi^a \delta \boldsymbol{\textbf{C}}_a = \int_{\mathcal{H}_{12}} \delta T_{ab} \xi^a k^b \sqrt{h} dV d^{n-2} x
\ee
where $k^a$ is the tangent to the affinely parametrized generators of the horizon. Thus, we obtain
\be
 \frac{\kappa}{2\pi} \Delta \delta S = \Delta \delta {E} 
\label{pp1st}
\ee
where $\Delta \delta S$ is the first order entropy difference between ``times'' $\mathcal C_2$ and $\mathcal C_1$, and $\Delta \delta {E}$ is the first order matter energy flux (relative to the horizon Killing field $\xi^a$) through the horizon between these ``times.'' Since $\kappa/2 \pi$ is the temperature of the black hole, this has exactly the form of the ``physical process'' version of the first law of black hole mechanics \cite{Gao_Wald_2001}. Equation \eqref{pp1st} is valid for first order perturbations with matter sources off of a stationary black hole in an arbitrary theory of gravity obtained from a diffeomorphism covariant Lagrangian. 

If the matter stress-energy satisfies the null energy condition, then $\delta T_{ab} \xi^a k^b \geq 0$. It then follows immediately from (\ref{pp1st}) that our definition of entropy satisfies the classical second law $\Delta \delta S \geq 0$ for first order perturbations. For source free first order perturbations, we have $\Delta \delta S = 0$, so the leading order change in entropy occurs at second order. As previously mentioned, we will show in subsection \ref{vacpert} that the entropy is non-decreasing with time at second order if and only if the modified canonical energy is non-negative at second order. This is the case in general relativity but would not be expected to hold in more general theories of gravity.

For the case of general relativity, the term $- \xi \cdot \textbf{B}_{\mathcal H}$ appearing in (\ref{dynent}) gives rise to the ``dynamical correction term'' to the Bekenstein--Hawking entropy given in (\ref{degr}). This correction term makes $S[\mathcal C]$ smaller than the Bekenstein--Hawking entropy. It follows immediately from (\ref{pp1st}) that, to first order, our entropy increases only when matter crosses the horizon. This contrasts with the Bekenstein--Hawking entropy, which increases before matter is thrown into the black hole since the event horizon is teleologically defined and moves outward in anticipation of matter being thrown in at a later time. We show in appendix \ref{sec:area_app_horizon} that, to first order, our entropy (\ref{degr}) is, in fact, the area of the apparent horizon corresponding to ``time'' $\mathcal C$, and thus can be determined locally, without knowledge of the future behavior of the spacetime.

As previously mentioned, our formula for dynamical black hole entropy differs from the Dong--Wall formula. However, we show in section \ref{dwent} that there is a close relationship between our notion of entropy and the Dong--Wall entropy $S_{\rm DW}$. If, for simplicity, we consider the entropy $S[V]$ evaluated on a cross-section of constant affine time $V$, then we have
\be
S[V] = S_{\rm DW}[V] - V \frac{d}{d V} S_{\rm DW}[V] .
\ee
In particular, we have
\be
\frac{d}{d V} S = - V\frac{d^2}{d V^2} S_{\rm DW}.
\ee
As we shall see in section \ref{gslq}, the fact that the first derivative of our entropy is related to the second derivative of the Dong--Wall entropy will play an important role in the analysis of the generalized second law and its relationship to the quantum null energy condition (QNEC). We will show that with our notion of entropy, the generalized second law is equivalent to QNEC, but with a dynamical correction to the von Neumann entropy. On the other hand, with the Dong--Wall notion of entropy, the generalized second law is equivalent to an integrated version of QNEC, using the unmodified notion of von Neumann entropy of matter.

In section \ref{lagkh}, we give the necessary background material on the Lagrangian formalism that we will be using throughout the paper and we also define our tetrad choice and specify our gauge conditions on the horizon. In section \ref{Bconstruction}, we prove that for perturbations of a stationary black hole whose event horizon is a bifurcate Killing horizon, there exists a quantity $\textbf{B}_{\mathcal H}$ such that $\underline{\boldsymbol{\theta}} \horeq \delta \textbf{B}_{\mathcal H}$. This enables us to give a general definition of dynamical black hole entropy and elucidate its properties in section \ref{dbesec}. We derive the physical process version of the first law of black hole mechanics in section \ref{pp12} for both external matter perturbations and vacuum perturbations and discuss its implications for the second law. The relationship of our notion of entropy with the Dong--Wall entropy is analyzed in section \ref{dwent}. The generalized second law and its relationship to QNEC is analyzed in section \ref{gslq} for both our notion of entropy and the Dong--Wall notion.

We will generally follow the conventions of \cite{Wald1984}. As we have already done above, we will use boldface to denote differential forms when the form indices are not explicitly written. Pullbacks of differential forms to the event horizon $\mathcal{H}$ will be denoted by underlining them as in \eqref{thetaB} above. The covariant spacetime volume element is written as $\epsilon_{a_1 \dots a_n}$, and our orientation conventions for the event horizon and horizon cross-sections are given by \eqref{eq:or_conv_pt_1} and \eqref{eq:or_conv_pt_2} below. We will use the notation ``$\horeq$'' to denote that an equality that holds only when both sides are restricted (but not necessarily pulled back) to $\mathcal{H}$, although we will often simply use ``$=$'' for equations involving quantities that are only defined on $\mathcal{H}$, where no confusion can arise. We also set $\hbar = 1 = G$.

\section{Lagrangian Formalism, Killing Horizons, Tetrad Choice, and Gauge Conditions}
\label{lagkh}

In this section, we review some basic definitions and constructions in the Lagrangian formalism, we specify a tetrad choice on the horizon, and we impose some gauge conditions on the horizon.

\subsection{\label{sec:iyer_wald_review} Lagrangians, Noether Charge, and the Fundamental Identity}

We consider an arbitrary diffeomorphism covariant theory of gravity in $n$-dimensions derived from a Lagrangian $n$-form $\textbf{L}$ of the form\footnote{
Due to the Bianchi-identities, not all of the derivatives
$\nabla_{(a_1} \cdots \nabla_{a_m)} R_{bcde}$ are algebraically independent at each point. An independent set 
is given e.g., by $\nabla_{(a_1} \cdots \nabla_{a_m} R^b{}_{cd)}{}^e$ \cite{muller1997closed}.}
\begin{align}
    \textbf{L}={}\textbf{L}(g_{ab}, R_{abcd}, \nabla_{a_1} R_{bcde}, \dots, \nabla_{(a_1} \cdots \nabla_{a_m)} R_{bcde}, \psi, \dots, \nabla_{(b_1} \cdots \nabla_{a_p)} \psi) \label{eq:form_of_Lag}
\end{align} where $g_{ab}$ is the spacetime metric, $\nabla$ denotes the derivative operator associated with $g_{ab}$, $R_{abcd}$ denotes the curvature of $g_{ab}$, and $\psi$ denotes any matter fields included in the Lagrangian. We shall assume that $\psi$ is a tensor field or fields; spinor fields can be treated via the approach developed in \cite{Prabhu_2017}. We collectively refer to the dynamical fields as $\phi = (g_{ab},\psi)$. It was shown in \cite{IyerWald1994} that any diffeomorphism covariant Lagrangian can be put in the form \eqref{eq:form_of_Lag}. Beginning in section \ref{Bconstruction}, we will restrict consideration to the case where matter fields, $\psi$, are not present in the Lagrangian, but we allow their presence for now.

The first variation of the Lagrangian can always be expressed in the form 
\begin{align}
    \delta \textbf{L} ={}& \textbf{E} \delta \phi + d \boldsymbol{\theta} \label{eq:varied_Lagrangian_form}
\end{align} where $\textbf{E}$ is locally constructed out of $\phi$ and its derivatives and $\boldsymbol{\theta}(\phi,\delta\phi)$ is locally constructed out of $\phi$, $\delta \phi$, and their derivatives and is linear in $\delta \phi$. In \eqref{eq:varied_Lagrangian_form}, contraction of the tensor indices of $\phi$ with the corresponding dual tensor indices of $\textbf{E}$ is understood. The Euler--Lagrange equations of motion obtained from $\textbf{L}$ are $\textbf{E} = 0$. The $(n-1)$-form $\boldsymbol{\theta}$ plays the role of a \textit{symplectic potential} \cite{IyerWald1994} in that the symplectic current $(n-1)$-form is obtained from it via\footnote{Here we assume that the field variations $\delta_1\phi$ and $\delta_2\phi$ arise from a two-parameter variation $\phi(\lambda_1, \lambda_2)$ and thus commute. Note that \eqref{symcur} uses the opposite sign convention of \cite{IyerWald1994}, but the same sign convention as \cite{Hollands_Wald_2013}.} 
\begin{align}
    \boldsymbol{\omega}(\phi,\delta_1 \phi,\delta_2 \phi) ={}& \delta_1 \boldsymbol{\theta}(\phi,\delta_2\phi) - \delta_2 \boldsymbol{\theta}(\phi,\delta_1 \phi).
    \label{symcur}
\end{align}  
The symplectic form $W$ is obtained by integrating the symplectic current over a Cauchy surface, $\Sigma$, of the spacetime
\begin{align}
    W_\Sigma(\phi, \delta_1 \phi, \delta_2\phi) ={}& \int_\Sigma \boldsymbol{\omega}(\phi,\delta_1 \phi, \delta_2 \phi) \, .
\end{align}

Let $\chi^a$ be an arbitrary vector field on the spacetime $\mathcal M$. Since $\chi^a$ is the infinitesimal generator of a diffeomorphism, i.e., a local symmetry of $\textbf{L}$, there is an associated Noether current $(n-1)$-form $\textbf{J}$ defined by
\begin{align}
    \textbf{J}(\phi)={}& \boldsymbol{\theta}(\phi,\mathscr{L}_\chi \phi) - \chi \cdot \textbf{L}(\phi).
    \label{noecur}
\end{align} 
where $\mathscr{L}_\chi $ denotes the Lie derivative with respect to $\chi^a$ and we use the notation ``$\cdot$'' to denote the contraction of a vector field into the first index of a differential form. This definition of $\textbf{J}$ holds for any field configuration $\phi$, i.e., $\phi$ need not be a solution to the field equations, $\textbf{E} = 0$. It can be shown that the Noether current can be written in the form \cite{IyerWald1995}
\begin{align}
    \textbf{J} ={}& d \textbf{Q}[\chi] + \chi^a \boldsymbol{\textbf{C}}_a \, . \label{eq:def_noe_curr_charge}
\end{align} 
Here, the $(n-2)$-form \textbf{Q} is referred to as the ``Noether charge'' and the dual vector valued $(n-1)$-form $\boldsymbol{\textbf{C}}_a$ vanishes when the equations of motion are satisfied and is referred to as the ``constraints'' of the theory associated with diffeomorphism invariance. The Noether charge takes the general form \cite{IyerWald1994}
\begin{align}
    \textbf{Q}[\chi]={}& \textbf{W}_c(\phi) \chi^c + \textbf{X}^{cd}(\phi) \nabla_{[c} \chi_{d]} + \textbf{Y}(\phi,\mathscr{L}_\chi \phi) + d\textbf{Z}(\phi, \chi) \label{eq:noether_charge_decomp}
\end{align} 
where $\textbf{W}_c$, $\textbf{X}^{cd}$, $\textbf{Y}$, and $\textbf{Z}$ are covariant quantities, with $\textbf{W}_c$ and $\textbf{X}^{cd}$ locally constructed from $\phi$, and $\textbf{Y}$ and $\textbf{Z}$ locally constructed from the indicated fields, with $\textbf{Y}$ linear in $\mathscr{L}_\chi \phi$ and $\textbf{Z}$ linear in $\chi$.
For the case where the matter field $\psi$ is a single tensor field $\tensor{A}{^{a_1 \cdots a_k}_{b_1 \cdots b_l}}$ of type $(k,l)$, the constraints take the explicit form \cite{SeifertWald2007}
\begin{align}
    {}& {{C}}_{a a_1 \cdots a_{n-1}} \nonumber \\
    ={}&  {\epsilon}_{c a_1 \cdots a_{n-1}}  \left[2 \tensor{( {E}_G)}{^c_a} - \sum \tensor{A}{^{{d_1} \cdots {d_k}}_{b_1 \cdots a \cdots b_l}} \tensor{( {E}_M)}{_{d_1 \cdots d_k}^{b_1 \cdots c \cdots b_l}} + \sum \tensor{A}{^{d_1 \cdots c \cdots d_k}_{b_1 \cdots b_l}} \tensor{({E}_M)}{_{d_1 \cdots a \cdots d_s}^{b_1 \cdots b_l}} \right], 
    \label{eq:constraints_eom}
\end{align}
where ${\epsilon}_{a_1 \cdots a_n}$ is the spacetime volume form, $({E}_G)_{ab}=0$ are the equations of motion for $g^{ab}$ (obtained by varying the Lagrangian with respect to $g^{ab}$) and $\tensor{({E}_M)}{_{a_1 \cdots a_k}^{b_1 \cdots b_l}}$ are the equations of motion for the matter field $\tensor{A}{^{a_1 \cdots a_k}_{b_1 \cdots b_l}}$. The summations in \eqref{eq:constraints_eom} run over the possible index positions of $a$ and $c$. The generalization of \eqref{eq:constraints_eom} to more than one matter field is straightforward.

Taking the first variation of \eqref{noecur} (taking the vector field $\chi^a$ to be fixed) and using \eqref{eq:varied_Lagrangian_form} and \eqref{symcur}, we obtain \cite{IyerWald1994} 
\begin{align}
    \delta \textbf{J}(\phi) ={}& -\chi \cdot (\textbf{E}(\phi) \delta \phi) + \boldsymbol{\omega}(\phi; \delta \phi, \mathscr{L}_\chi \phi) + d [\chi \cdot \boldsymbol{\theta}(\phi,\delta \phi)], \label{eq:first_var_noe_curr}
\end{align} 
Equating this to the first variation of \eqref{eq:def_noe_curr_charge}, we obtain
\begin{align}
    \boldsymbol{\omega}(\phi, \delta \phi, \mathscr{L}_\chi \phi) ={}& \chi \cdot (\textbf{E}(\phi)\delta \phi) + \chi^a \delta \boldsymbol{\textbf{C}}_a(\phi) + d(\delta \textbf{Q}[\chi] - \chi \cdot \boldsymbol{\theta}(\phi, \delta \phi)]). 
    \label{eq:fund_id}
\end{align}
 Following Hollands and Wald \cite{Hollands_Wald_2013}, we shall refer to this equation as the \textit{fundamental identity}. The fundamental identity holds for arbitrary $\phi$ and $\delta \phi$, i.e., they need not satisfy the equations of motion or linearized equations of motion. For the case where $\chi^a$ is a Killing field of the background $\phi$ so that $\mathscr{L}_\chi \phi = 0$, a second variation of \eqref{eq:fund_id} yields
\begin{align}
    \begin{split}
    \boldsymbol{\omega}(\phi,\delta\phi, \mathscr{L}_\chi \delta \phi) ={}& \chi \cdot (\delta\textbf{E}(\phi) \delta\phi) + \chi \cdot (\textbf{E}(\phi) \delta^2 \phi) + \chi^a \delta^2 \boldsymbol{\textbf{C}}_a(\phi) + d(\delta^2 \textbf{Q}[\chi] - \chi \cdot \delta \boldsymbol{\theta}(\phi,\delta\phi)). \label{eq:fund_id_var}
    \end{split}
\end{align}

 \subsection{Tetrad Choice on a Bifurcate Killing Horizon}
 \label{sec:aux_vf_n}

We are interested in this paper in first and second order perturbations of spacetimes that contain a black hole whose event horizon is a bifurcate Killing horizon. By definition, a Killing horizon $\mathcal H$ is a null surface to which a Killing field $\xi^a$ is normal. Thus, $\xi^a$ is tangent to the null geodesic generators of $\mathcal H$. The surface gravity, $\kappa$, of $\mathcal H$ is defined by 
\be
\xi^b \nabla_b \xi^a = \kappa \xi^a
\label{surgra}
\ee
so $\kappa$ measures the failure of Killing parametrization of the null generators to be an affine parametrization. It is necessary for $\kappa$ to be constant in order that there not be a parallelly propagated curvature singularity on $\mathcal H$ \cite{racz1992extensions}. We shall consider here only the case where $\kappa$ is constant and $\kappa \neq 0$, in which case $\mathcal H$ corresponds to a bifurcate Killing horizon \cite{Racz_Wald_1996}. 

We wish to use a tetrad basis on $\mathcal H$ to the future of the bifurcation surface $\mathcal B$ that is Lie transported by the Killing field. One way of doing this would be to introduce Gaussian null coordinates, using the Killing parameter $v$ as one of the coordinates (see, e.g., \cite{RahmanWald2020}). However, the introduction of Gaussian null coordinates requires an arbitrary choice of a cross-section, $\mathcal C$, of $\mathcal H$, and this would raise questions in our subsequent analysis as to the extent to which our constructions depend upon our choice of $\mathcal C$. Instead, we will introduce a tetrad that is constructed only from $\xi^a$---since our constructions will, in any case, depend upon $\xi^a$. 

We do so by introducing a vector field $N^a$ on $\mathcal H$ that satisfies
\begin{eqnarray}
\nabla_{a} \xi_{b} &\horeq& 2\kappa N_{[a} \xi_{b]} \label{Nprop1} \\
N^a N_a &\horeq& 0 \label{Nprop2}
\end{eqnarray}
where $\horeq $ denotes equality on $\mathcal H$. That there exists a unique $N^a$ satisfying these conditions can be seen as follows: From the hypersurface orthogonality of $\xi^a$ on $\mathcal H$ and Frobenius's theorem, there exists a vector field $w^a$ on $\mathcal{H}$ such that 
\be
\nabla_{[a} \xi_{b]} \horeq  2 w_{[a} \xi_{b]}
\label{frob}
\ee 
Since $\xi^a$ is a Killing field, we have $\nabla_{[a} \xi_{b]}= \nabla_{a} \xi_{b}$ and it follows from \eqref{surgra} that 
\be
\xi^a w_a \horeq  \kappa
\ee 
so, in particular, $w_a \neq 0$.
We define
\be
N^a \horeq (\xi^b w_b)^{-1} w^a + c \xi^a = \frac{1}{ \kappa}  w^a + c \xi^a
\label{Nform}
\ee
with
\be
c = -\frac{1}{2 (\xi^b w_b)^2} w^a w_a = -\frac{1}{2 \kappa^2} w^a w_a
\label{cform}
\ee
Then it is easily seen that $N^a$ satisfies \eqref{Nprop1} and \eqref{Nprop2}. Uniqueness follows from the fact that \eqref{Nprop1} uniquely determines $N^a$ up to addition of a multiple of $\xi^a$ at each point and \eqref{Nprop2} then fixes that multiple. Note that if we contract \eqref{Nprop1} with $\xi^a$ and use \eqref{surgra}, we find that
$N^a \xi_a \horeq 1$. Thus, $N^a$ is a past-directed null vector. Note also that $N^a$ must be invariant under the action of the isometries generated by $\xi^a$, since this action preserves \eqref{Nprop1} and \eqref{Nprop2} but $N^a$ is unique. Thus, we have $\mathscr{L}_\xi N^a \horeq 0$. 

It is useful to complete $\xi^a$ and $N^a$ to a basis of the tangent space on $\mathcal H$ by introducing vectors $s_i^a$ ($i=1,\cdots, n-2$) on $\mathcal H$ such that $s_i^a \xi_a = 0 = s_i^a N_a$, $s_i^a s_{ja} = \delta_{ij}$ and $\mathscr{L}_\xi s_i^a \Big|_{\mathcal{H}}= 0$. This can be done by choosing $s_i^a$ on a cross-section\footnote{Note that there is no reason why there need exist a cross-section to which $N^a$ is orthogonal, so, unlike a Gaussian null coordinate basis, $s_i^a$ is not assumed to be tangent to the cross-section. It may not be possible to make global choices of $s_i^a$ on a cross-section, in which case our construction would have to be done in patches. Different choices of $s_i^a$ will merely differ by rotations in the $(n-2)$-plane orthogonal to $\xi^a$ and $N^a$, and it will be manifest that our results below will not depend upon the choice of $s_i^a$.} so as to satisfy these properties and then Lie-transporting $s_i^a$ along $\xi^a$.

It is important to understand the behavior of our tetrad $\{\xi^a, N^a, s_i^a \}$ as one approaches the bifurcation surface $\mathcal B$. This is most conveniently analyzed by introducing Gaussian null coordinates based on a choice of affine parameter $V$ of the null generators of $\mathcal H$ (see, e.g., \cite{Hollands_2022})---with $V=0$ corresponding to the bifurcation surface $\mathcal B$---since such coordinates are smooth at $\mathcal B$. Since we require $V=0$ at $\mathcal B$, the choice of $V$ is unique up to $V \to fV$, where $f$ may vary from generator to generator but is constant on each generator. The results below hold for any choice of $V$. By \eqref{surgra}, the relationship between Killing parameter $v$ and affine parameter $V$ is
\be
V = e^{\kappa v}
\label{kilpar}
\ee
where the (generator dependent) scaling freedom of $V$ corresponds to the (generator dependent) freedom in the choice of origin of $v$. Thus, we have
\be
\xi^a = \kappa e^{\kappa v} k^a = \kappa V k^a
\label{afftang}
\ee
where $k^a$ is tangent to the affinely parametrized geodesics and thus is smooth and nonvanishing on $\mathcal B$. Thus, we see that ${\xi^a = O(V) }$ at $V=0$. Similarly, $N^a = O(V^{-1})$ at $V=0$, i.e., the components of $N^a$ in Gaussian null coordinates blow up as $1/V$ as $V \to 0$. Finally, the tetrad vectors $s_i^a$ smoothly extend to $\mathcal B$. To see this, we note that 
\be
\mathscr{L}_k s_i^a = \frac{1}{\kappa} e^{-\kappa v} \mathscr{L}_\xi s_i^a + e^{-\kappa v} \xi^a s_i^b \nabla_b v =  \kappa k^a s_i^b \nabla_b v
\ee
Since ${\mathscr L}_\xi [\kappa s_i^b \nabla_b v] = 0$, the right side is smooth, which implies that $s_i^a$ smoothly extends to $\mathcal B$. 

The following is an important consequence of the results of the previous paragraph: Suppose that in a spacetime with a bifurcate Killing horizon we have a tensor field ${S^{a_1 \cdots a_k}}_{b_1 \cdots b_l}$ that is Lie derived by $\xi^a$ and is smooth on ${\mathcal H} \setminus {\mathcal B}$. We can expand ${S^{a_1 \cdots a_k}}_{b_1 \cdots b_l}$ on ${\mathcal H} \setminus {\mathcal B}$ in the basis $\{\xi^a, N^a, s_i^a \}$. Since both ${S^{a_1 \cdots a_k}}_{b_1 \cdots b_l}$ and the basis elements are Lie derived by $\xi^a$, the coefficients appearing in the basis expansion must also be Lie derived by $\xi^a$. It follows immediately that if any nonvanishing term in the basis expansion has strictly more $N^a$'s than $\xi^a$'s, then ${S^{a_1 \cdots a_k}}_{b_1 \cdots b_l}$ cannot be extended to $\mathcal B$. Indeed, in smooth coordinates covering $\mathcal B$, some components of ${S^{a_1 \cdots a_k}}_{b_1 \cdots b_l}$ will blow up as $1/V^p$, where $p$ is the maximum of the number of $N^a$'s minus the number of $\xi^a$'s appearing in any individual nonvanishing term in the basis expansion. On the other hand, if all of the nonvanishing terms in the basis expansion of ${S^{a_1 \cdots a_k}}_{b_1 \cdots b_l}$ have at least as many $\xi^a$'s as $N^a$'s, then ${S^{a_1 \cdots a_k}}_{b_1 \cdots b_l}$ can be smoothly extended to $\mathcal B$.

\subsection{Gauge Conditions on Perturbations}
\label{sec:gauge_orient}

In the remainder of this paper, we will be concerned with perturbations, $\delta g_{ab}$, of a black hole with a bifurcate Killing horizon. Without loss of generality, we will assume that the true event horizon of the perturbed spacetime coincides with $\mathcal H$. We take $\xi^a$ to be fixed under the variation, i.e., we take $\delta \xi^a = 0$. We will impose the following two gauge conditions on $\delta g_{ab}$ at the event horizon: 
\begin{eqnarray}
\xi^a \delta g_{ab} &\horeq& 0 \label{gaugecon1} \\
\nabla_a (\xi^b \xi^c \delta g_{bc}) &\horeq& 0 .
\label{gaugecon2}
\end{eqnarray}
Both of the above gauge conditions can be imposed without any loss of generality. This is most easily seen from the fact that they automatically hold for any perturbation in Gaussian null coordinates based on the Killing parameter $v$ (see \eqref{gnck} below). However, we will use only the above two gauge conditions in our analysis below, e.g., we need not assume that the full Gaussian null gauge conditions are imposed.

Equation \eqref{gaugecon1} states that $\xi^a$ remains the null normal to $\mathcal H$ under perturbations, since the condition implies that $\xi^a$ remains null and orthogonal to any vector tangent to $\mathcal H$. Equation \eqref{gaugecon2} implies that 
\be
\delta \left(\xi^b \nabla_b \xi^a \right) = \xi^b \delta {\Gamma^a}_{bc} \xi^c \horeq - \frac{1}{2} \xi^b \xi^c \nabla^a \delta g_{bc} \horeq - \frac{1}{2} \nabla^a \left(\xi^b \xi^c \delta g_{bc} \right) \horeq 0
\ee
where \eqref{gaugecon1} was used to get the second and third equalities. Comparing with \eqref{surgra}, we see that the second condition is equivalent to 
\be
\delta \kappa = 0
\ee
i.e., the surface gravity of $\xi^a$ on $\mathcal H$ does not change under the perturbation.

If matter fields $\psi$ are present in the Lagrangian, then the results we will obtain below can be straightforwardly generalized if additional conditions are imposed on 
their perturbations, $\delta \psi$. In particular, if an electromagnetic field is present, our results can be straightforwardly generalized if the condition $\xi^a \delta A_a \horeq 0$ is imposed, which can always be achieved without loss of generality by an electromagnetic gauge transformation. 
However, very recently, Wall and Yan \cite{Wall_Yan_2024} have extended the definition of the Dong--Wall entropy to Proca fields without imposing additional restrictions on the Proca field. Furthermore, even more recently, 
Visser and Yan \cite{Visser:2024pwz} have shown that the results of the next section can be generalized to the case where general matter fields are present, without imposing additional restrictions on the perturbations. Thus, we believe that it should be possible to extend our analysis to the case where additional matter fields are present. Nevertheless, for the remainder of this paper, we shall restrict consideration to the case where the metric is the only dynamical field in the Lagrangian, i.e., no additional dynamical fields $\psi$ are present.

\section{The Pullback of the Symplectic Potential is a Total Variation for Perturbations of a Stationary Black Hole}
\label{Bconstruction}

In this section, we consider perturbations of a stationary black hole whose event horizon is a bifurcate Killing horizon. We assume that the perturbations are everywhere smooth and satisfy the gauge conditions \eqref{gaugecon1} and \eqref{gaugecon2}. However, we need not assume that the background spacetime metric $g_{ab}$, is a solution to the field equations, nor do we need to assume that the perturbed metric, $\delta g_{ab}$ is a solution to the linearized field equations. We will prove that there exists a quantity $\textbf{B}_{\mathcal H}$ defined on the black hole horizon that is locally and covariantly constructed from the metric and\footnote{The vector field $N^a$ on $\mathcal H$ will also appear in our expression for $\textbf{B}_{\mathcal H}$. However, $N^a$ on $\mathcal H$ is locally and covariantly determined from the metric and $\xi^a$ via \eqref{Nprop1} and \eqref{Nprop2}.} $\xi^a$
such that the pullback, $\underline{\boldsymbol{\theta}}$, to $\mathcal H$ of the symplectic potential $\boldsymbol{\theta}(g_{ab}, \delta g_{ab})$ is given by
\be
\underline{\boldsymbol{\theta}} \horeq \delta \textbf{B}_{\mathcal H}.
\label{sympotvar}
\ee
As explained in the Introduction, \eqref{sympotvar} will enable us to define dynamical black hole entropy, and we shall then do so in the next section. As stated at the end of the previous section, we restrict consideration here to the vacuum case, but very recently, Visser and Yan \cite{Visser:2024pwz} have shown that \eqref{sympotvar} holds in the non-vacuum case as well. Their expression for $\textbf{B}_{\mathcal H}$ is not covariant in the metric and $\xi^a$, but a covariant expression can be obtained using methods similar to those used in the proof of Theorem \ref{prop:sym_pot_vanish_stat_bgd_pert} below.

The horizon $\mathcal H$ is an $(n-1)$-dimensional null surface. Since the pullback of the spacetime metric to $\mathcal H$ is degenerate, it does not induce a unique volume form on $\mathcal H$. However, we can define a volume form $\boldsymbol{\epsilon}^{(n-1)}_{a_1 \cdots a_{n-1}}$ on $\mathcal H$ that is Lie derived by $\xi^a$ in the background spacetime by
\be
-n \xi_{[a_1}{\epsilon}^{(n-1)}_{a_2 \cdots a_n]} = {\epsilon}_{a_1 a_2 \cdots a_n}^{}
\label{eq:or_conv_pt_1}
\ee
where ${\epsilon}_{a_1 a_2 \cdots a_n}$ is the spacetime volume form. It is also convenient to define an $(n-2)$-form ${\epsilon}^{(n-2)}_{a_1 \cdots a_{n-2}}$ on $\mathcal H$ by
\be
{\epsilon}^{(n-2)}_{a_1 \cdots a_{n-2}} = \xi^c {\epsilon}^{(n-1)}_{ca_1 \cdots a_{n-2}} 
\label{eq:or_conv_pt_2}
\ee
The pullback of $\boldsymbol{\epsilon}^{(n-2)}$ to any cross-section, $\mathcal C$, of $\mathcal H$ will then agree with the volume element on $\mathcal C$ obtained from the nondegenerate metric on $\mathcal C$ obtained from the pullback of the spacetime metric. 

The symplectic potential $\boldsymbol{\theta}$ is an $(n-1)$-form on spacetime, so its pullback $\underline{\boldsymbol{\theta}}$ to $\mathcal H$ must be proportional to ${\epsilon}^{(n-1)}_{a_1 \cdots a_{n-1}}$. It follows from \eqref{eq:or_conv_pt_1} that we have
\be
\underline{\boldsymbol{\theta}} \horeq  \alpha^a  \xi_a \boldsymbol{\epsilon}^{(n-1)}
\label{sympotpull}
\ee
where we have defined 
\be
\alpha^a \equiv \frac{1}{(n-1)!} {\theta}_{a_1 \cdots a_{n-1}}  {\epsilon}^{a a_1 \cdots a_{n-1}}  .
\label{sympotpull2}
\ee
Our results on the form of $\underline{\boldsymbol{\theta}}$ will be a direct consequence of the following theorem.

\begin{theorem}
\label{prop:sym_pot_vanish_stat_bgd_pert}
Let $(\mathcal M, g_{ab})$ be an n-dimensional spacetime with a bifurcate Killing horizon $\mathcal H$ with horizon Killing field $\xi^a$. Let $\delta g_{ab}$ be an arbitrary smooth perturbation of $g_{ab}$ satisfying the gauge conditions \eqref{gaugecon1} and \eqref{gaugecon2}. Let $\alpha^a$ be a vector field on $\mathcal M$ of the form
\begin{align}
    \alpha^a ={}& \sum_{i=0}^{k} T_{(i)}^{a b_1 \cdots b_i cd}\nabla_{(b_1} \cdots \nabla_{b_i)} \delta g_{cd}, \label{eq:alpha_gen_form}
\end{align}
where the the tensors $T_{(i)}^{a b_1 \cdots b_i cd} = T_{(i)}^{a (b_1 \cdots b_i) (cd)} $ are smooth and are locally and covariantly constructed from the metric, curvature, and covariant derivatives of the curvature. Then on the horizon, the scalar function $\alpha_c \xi^c$ takes the form
\begin{align}
    \alpha_c \xi^c \horeq{}& \sum_{i=0}^{k-1} \tilde{T}_{(i)}^{b_1 \cdots b_i cd} \nabla_{(b_1} \cdots \nabla_{b_i)} \mathscr{L}_\xi \delta g_{cd} \label{eq:alpha_xi_final_form}
\end{align} 
where the tensors $\tilde{T}_{(i)}^{b_1 \cdots b_i cd}$ are smooth on $\mathcal H$ and are locally and covariantly constructed from the metric, curvature, and its derivatives as well as from $\xi^a$ and $N^a$ (see \eqref{Nprop1} and \eqref{Nprop2}), with $\xi^a$ and $N^a$ appearing only algebraically.
\end{theorem}

\begin{proof} 
To begin, we focus on the tensor field $T_{(k)}^{a b_1 \cdots b_i cd}$ on $\mathcal H$ appearing in the highest derivative term in the expansion \eqref{eq:alpha_gen_form} of $\alpha^a$. Our aim is to show that $\xi_a T_{(k)}^{a b_1 \cdots b_i cd}\nabla_{(b_1} \cdots \nabla_{b_k)} \delta g_{cd}$ can be written in the form $\tilde{T}_{(k-1)}^{b_1 \cdots b_{k-1} cd}\nabla_{(b_1} \cdots \nabla_{b_{k-1})} {\mathscr L}_ \xi \delta g_{cd}$ plus lower derivative terms of essentially the same character as appeared in the original expression for  $\xi_a \alpha^a$. This will enable us to make an inductive argument to prove the theorem.

By the hypothesis of the theorem, the tensor field $T_{(k)}^{a b_1 \cdots b_i cd}$ is smooth on $\mathcal H$, including at the bifurcation surface $\mathcal B$. Since $\xi^a$ is a Killing field of the background spacetime, $T_{(k)}^{a b_1 \cdots b_k cd}$ is Lie derived by $\xi^a$. Therefore, as shown in the last paragraph of subsection \ref{sec:aux_vf_n}, if we expand it in the basis $\{\xi^a, N^a, s_i^a \}$, the only nonvanishing terms in the basis expansion will have at least as many factors of $\xi^a$ as $N^a$. If we then contract $\alpha^a$ with $\xi^a$, we will eliminate one factor of $N^a$ from each term in the basis expansion, so each nonvanishing term will have at least one more factor of $\xi^a$ than $N^a$. More explicitly, the basis expansion of $\xi_a T_{(k)}^{a b_1 \cdots b_i cd}$ takes the form
\begin{align}
\begin{split}
    \xi_a T_{(k)}^{ab_1 \cdots b_k cd} \horeq{}& c_1 \xi^{b_1}\xi^{b_2} \cdots \xi^{b_k}\xi^c \xi^d \\
    {}& + c_2 N^{(b_1} \xi^{b_2} \cdots \xi^{b_k)} \xi^c \xi^d + c_3 \xi^{b_1} \xi^{b_2} \cdots \xi^{b_k} N^{(c} \xi^{d)}\\
    {}& + c_4 s_{i_{1}}^{(b_1} \xi^{b_2} \cdots \xi^{b_k)} \xi^c \xi^d + c_5 \xi^{b_1} \xi^{b_2} \cdots \xi^{b_k} s_i^{(c} \xi^{d)}\\
    {}&+ \cdots \\
    {}&+ c_6 s_{i_{1}}^{(b_1} s_{i_{2}}^{b_2} \cdots s_{i_{{k-1}}}^{b_{k-1}} \xi^{b_k)} s_{j_1}^{(c} s_{j_2}^{d)} + c_7 s_{i_{1}}^{b_1} s_{i_{2}}^{b_2} \cdots s_{i_{k}}^{b_k} \xi^{(c} s_{j}^{d)}, \label{eq:T_j_basis_exp}
\end{split}
\end{align}
where $c_1, c_2, \dots, c_7$ are Lie derived by $\xi^a$.
Here, the first term in the basis expansion is the term with all $\xi^a$'s. In the second line, we replace one of the $\xi^a$'s with an $N^a$. Since $T_{(k)}^{a b_1 \cdots b_i cd}$ is symmetric in $b_1 \dots b_k$ as well as in $c d$, there are two distinct terms that we can get via such a replacement, namely, one term where the $N^a$ replaces a $\xi^a$ on a $b_i$ index and one term where it replaces a $\xi^a$ on a $c$ or $d$ index. In the third line, we similarly replace a $\xi^a$ by an $s_i^a$. (Since $i$ ranges from $1$ to $n-2$, there are $n-2$ distinct expressions of the form appearing in the third line.) The ``$\dots$'' in the fourth line represents all of the terms we get by continuing to replace $\xi^a$'s by $N^a$'s and $s_i^a$'s, thereby continually decreasing the number of factors of $\xi^a$. The key point is that each term must have at least one more $\xi^a$ than $N^a$. Thus, the final two terms in the basis expansion---shown on the fifth line---have $(k+1)$ factors of $s_i^a$ and one factor of $\xi^a$. 

We now consider each term in the basis expansion separately. Since at least one $\xi^a$ appears in each term in the basis expansion, there are three possible cases, which we will consider separately:

\begin{itemize}

\item {\it Case (1):} There are no $\xi^{b_i}$ factors, i.e., there are no factors of $\xi^a$ that contract into $\nabla_{(b_1} \cdots \nabla_{b_k)}$ in \eqref{eq:alpha_gen_form}. Furthermore, there is only one $\xi^c$ or $\xi^d$ factor, i.e., only one index of $\delta g_{cd}$ contracts into factor of $\xi^a$.

\item {\it Case (2):} As in case (1), there are no $\xi^{b_i}$ factors. However, there are $\xi^c$ and $\xi^d$ factors, i.e., both indices of $\delta g_{cd}$ contract into factors of $\xi^a$.

\item {\it Case (3):} There is at least one $\xi^{b_i}$ factor, i.e., there is at least one contraction of $\xi^a$ with a derivative index in \eqref{eq:alpha_gen_form}.

\end{itemize}

Case (1) is the most straightforward to analyze. Since only one factor of $\xi^a$ occurs in this case, the only terms in the basis expansion \eqref{eq:T_j_basis_exp} that meet the criterion of case (1) take the form of the very last term in \eqref{eq:T_j_basis_exp}. Thus, to analyze this case, we need to consider the quantity
\be
s_{i_{1}}^{b_1} s_{i_{2}}^{b_2} \cdots s_{i_{k}}^{b_k} \xi^{c} s_{j}^{d} \nabla_{(b_1} \cdots \nabla_{b_k)} \delta g_{cd}
\ee
where we dropped the symmetrization over $cd$ in the basis term since $\delta g_{cd}$ is symmetric in $c$ and $d$. We now move $\xi^c$, through all of the derivatives $\nabla_{(b_1} \cdots \nabla_{b_k)}$ to write it as
\be
s_{i_{1}}^{b_1} s_{i_{2}}^{b_2} \cdots s_{i_{k}}^{b_k} \xi^{c} s_{j}^{d} \nabla_{(b_1} \cdots \nabla_{b_k)} \delta g_{cd} = s_{i_{1}}^{b_1} s_{i_{2}}^{b_2} \cdots s_{i_{k}}^{b_k} s_{j}^{d} \nabla_{(b_1} \cdots \nabla_{b_k)}( \delta g_{cd} \xi^{c}) + {\rm lower \, derivative \, terms}
\ee
where the ``lower derivative terms'' have one or more derivatives $\nabla_{b_i}$ acting only on $\xi^c$ and thus have no more than $k-1$ derivatives remaining to act on $\delta g_{cd}$. However, by our gauge condition \eqref{gaugecon1}, we have $\delta g_{cd} \xi^c \horeq 0$. Since derivatives of this quantity in directions tangential to the horizon must also vanish, we see that the terms occurring in case (1) are equivalent to an expression with strictly fewer than $k$ derivatives acting on $\delta g_{cd}$.

The analysis of case (2) is similar to case (1). We move both $\xi^c$ and  $\xi^d$ through $\nabla_{(b_1} \cdots \nabla_{b_k)}$ to obtain a term where all $k$ derivatives act on $\delta g_{cd} \xi^c \xi^d$ plus ``lower derivative terms,'' i.e., terms where at least one derivative acts on $\xi^c$ or $\xi^d$ and thus no more than $k-1$ derivatives act on $\delta g_{cd}$. However, the term $\nabla_{(b_1} \cdots \nabla_{b_k)}(\delta g_{cd} \xi^c \xi^d)$ can have at most one factor $N^{b_i}$ contracting into it, since there are precisely two factors of $\xi^a$ in the terms that arise in case (2) and we must have strictly more factors of $\xi^a$ than $N^a$. Therefore, in the term $\nabla_{(b_1} \cdots \nabla_{b_k)}(\delta g_{cd} \xi^c \xi^d)$, there is at most one derivative of $\delta g_{cd} \xi^c \xi^d$ that is taken in a direction that is not tangential to the horizon. Consequently, this term vanishes by our gauge condition $\nabla_b (\delta g_{cd} \xi^c \xi^d) \horeq 0$. Therefore, we are left only with the ``lower derivative terms,'' i.e., terms with $k-1$ or fewer derivatives acting on $\delta g_{cd}$.

Finally, to treat case (3), we must analyze terms of the form
\be
\xi^{b_i} \nabla_{b_1} \cdots \nabla_{b_i} \cdots \nabla_{b_k} \delta g_{cd} .
\ee
To do so, we first bring $\xi^{b_i}$ through $\nabla_{b_1} \cdots \nabla_{b_{i-1}}$ to rewrite this term as 
\be
\xi^{b_i} \nabla_{b_1} \cdots \nabla_{b_i} \cdots \nabla_{b_k} \delta g_{cd} =  \nabla_{b_1} \cdots \nabla_{b_{i-1}} \left[ \xi^{b_i} \nabla_{b_i} \left(\nabla_{b_{i +1}} \cdots \nabla_{b_k} \delta g_{cd} \right) \right] + {\rm lower \, derivative \, terms}.
\ee 
We then write
\be
 \xi^{b_i} \nabla_{b_i} \left(\nabla_{b_{i +1}} \cdots \nabla_{b_k} \delta g_{cd} \right) = {\mathscr L}_ \xi \left(\nabla_{b_{i +1}} \cdots \nabla_{b_k} \delta g_{cd} \right) + {\rm lower \, derivative \, terms}.
\ee
Finally, we use the fact that since $\xi^a$ is a Killing field of the background metric, ${\mathscr L}_ \xi$ commutes with the background derivative operator $\nabla_b$, so we have
\be
 {\mathscr L}_ \xi \left(\nabla_{b_{i +1}} \cdots \nabla_{b_k} \delta g_{cd} \right) = \nabla_{b_{i +1}} \cdots \nabla_{b_k}  {\mathscr L}_ \xi \delta g_{cd}.
\ee
Putting this all together, we have shown that
\be
\xi^{b_i} \nabla_{b_1} \cdots \nabla_{b_i} \cdots \nabla_{b_k} \delta g_{cd} = \nabla_{b_1} \cdots \nabla_{b_{i-1}} \nabla_{b_{i +1}} \cdots \nabla_{b_k} {\mathscr L}_ \xi \delta g_{cd} + {\rm lower \, derivative \, terms}.
\ee

Combining all of the results of the analyses of cases (1), (2), and (3) above, we see that the highest derivative term in $\xi_a \alpha^a$ can be written in the form
\be
\xi_a T_{(k)}^{ab_1 \cdots b_k cd}\nabla_{(b_1} \cdots \nabla_{b_k)}\delta g_{cd}  \horeq  \tilde{T}_{(k-1)}^{b_1 \cdots b_{k-1} cd}\nabla_{(b_1} \cdots \nabla_{b_{k-1})} {\mathscr L}_ \xi \delta g_{cd} + \sum_{i=0}^{k-1} {T'}_{(i)}^{b_1 \cdots b_i cd}\nabla_{(b_1} \cdots \nabla_{b_i)} \delta g_{cd}
\label{kderform}
\ee
Here $\tilde{T}_{(k-1)}^{b_1 \cdots b_{k-1} cd}$ arises entirely from case (3), wherein one $\xi^b$ in terms in the basis expansion of $\xi_a T_{(k)}^{ab_1 \cdots b_k cd}$ has been combined with a $\nabla_b$ and converted to a ${\mathscr L}_\xi$. It follows that each term in basis expansion of $\tilde{T}_{(k-1)}^{b_1 \cdots b_{k-1} cd}$ has at least as many $\xi^a$'s as $N^a$'s and thus is smooth on $\mathcal H$. The components of $\tilde{T}_{(k-1)}^{b_1 \cdots b_{k-1} cd}$ are obtained algebraically from $\xi_a T_{(k)}^{ab_1 \cdots b_k cd}$ together with the basis elements $N^a$, $\xi^a$ and $\{s_i^a\}$. However, it is clear that $\tilde{T}_{(k-1)}^{b_1 \cdots b_{k-1} cd}$ does not depend on the choice of $\{s_i^a\}$, as can be seen from the fact that we could have worked with the projector
\be
q^{ab} = g^{ab} - 2 \xi^{(a} N^{b)}
\label{qdef}
\ee
rather than choosing the basis $\{s_i^a\}$. Since $\xi_a T_{(k)}^{ab_1 \cdots b_k cd}$ is locally and covariantly constructed from the metric, curvature, and its derivatives, we see that $\tilde{T}_{(k-1)}^{b_1 \cdots b_{k-1} cd}$ is locally and covariantly constructed from the metric, curvature, and its derivatives together with $\xi^a$ and $N^a$, with only algebraic dependence on $\xi^a$ and $N^a$. 

On the other hand ${T'}_{(i)}^{b_1 \cdots b_i cd}$ is algebraically constructed from not only $\xi_a T_{(k)}^{ab_1 \cdots b_k cd}$, $\xi^a$, and $N^a$ but also from derivatives of $\xi^a$. However, we can use the Killing field identity
\be
\label{KVF}
\nabla_a \nabla_b \xi_c = R_{cbad}\xi^d
\ee
to eliminate second derivatives of $\xi^a$ in terms of $\xi^a$ and curvature. By successively using this identity, we can rewrite all factors involving an arbitrary number of derivatives of $\xi^a$ in terms of curvature, derivatives of curvature, $\xi^a$, and $\nabla_b \xi^a$. We may then use 
\be
\nabla_{b} \xi_{a} \horeq 2\kappa N_{[b} \xi_{a]} 
\ee
to eliminate the first derivative of $\xi^a$ in terms of $\xi^a$ and $N^a$. In this way, we see that ${T'}_{(i)}^{b_1 \cdots b_i cd}$ is locally and covariantly constructed from the metric, curvature, derivatives of curvature, $\xi^a$ and $N^a$, with only algebraic dependence on $\xi^a$ and $N^a$.

Next, we note that since $\xi^a$ is a Killing field of the background spacetime, we have
\begin{align}
\tilde{T}_{(k-1)}^{b_1 \cdots b_{k-1} cd}\nabla_{(b_1} \cdots \nabla_{b_{k-1})} {\mathscr L}_ \xi \delta g_{cd} ={}&  {\mathscr L}_ \xi  \left[ \tilde{T}_{(k-1)}^{b_1 \cdots b_{k-1} cd}\nabla_{(b_1} \cdots \nabla_{b_{k-1})}\delta g_{cd} \right] \\
={}& \xi^a \nabla_a \left[ \tilde{T}_{(k-1)}^{b_1 \cdots b_{k-1} cd}\nabla_{(b_1} \cdots \nabla_{b_{k-1})}\delta g_{cd} \right] .
\end{align}
Thus, we may rewrite \eqref{kderform} as 
\be
\sum_{i=0}^{k-1} {T'}_{(i)}^{b_1 \cdots b_i cd}\nabla_{(b_1} \cdots \nabla_{b_i)} \delta g_{cd} = \xi_a T_{(k)}^{ab_1 \cdots b_k cd}\nabla_{(b_1} \cdots \nabla_{b_k)}\delta g_{cd} - \xi^a \nabla_a \left[ \tilde{T}_{(k-1)}^{b_1 \cdots b_{k-1} cd}\nabla_{(b_1} \cdots \nabla_{b_{k-1})}\delta g_{cd} \right] .
\label{kderform22}
\ee
The right side of \eqref{kderform22} is smooth and vanishes on the bifurcation surface on account of the explicit factor of $\xi^a$ in both terms. Thus, the left side of \eqref{kderform22} has exactly the same properties as $\xi_a \alpha^a$ except that (i) there now is additional algebraic dependence on $\xi^a$ and $N^a$ and (ii) the highest derivative term has only $k-1$ derivatives of $\delta g_{cd}$. The first point is of no consequence for any of the arguments that led from \eqref{eq:T_j_basis_exp} to \eqref{kderform}. The second point enables us to prove the theorem by the following inductive argument.

We take as an inductive hypothesis in $k$ that a slight strengthening of the statement of the theorem holds, namely, we allow the tensors appearing in \eqref{eq:alpha_gen_form} to depend algebraically on $\xi^a$ and $N^a$ as well as on the metric, curvature, and covariant derivatives of the curvature. We assume that the inductive hypothesis holds for $k=p-1$. For $k=p$, the above arguments show that the highest derivative term in $\xi_a \alpha^a$ takes the form \eqref{kderform} and that the lower derivative terms are such that the $k=p-1$ hypothesis can be applied to them. Therefore, the inductive hypothesis for $k= p-1$ implies the inductive hypothesis for $k=p$, and all that remains to show is that the inductive hypothesis holds for $k=0$. However, this is just a simple special case of cases (1) and (2) in the above proof of the proposition, and it follows immediately from our gauge conditions \eqref{gaugecon1} and \eqref{gaugecon2} that $\xi_a \alpha^a = 0$ when $k=0$. Thus, our inductive hypothesis holds for $k=0$ and the proof is completed.

\end{proof}

The main result of this section can now be obtained as follows. For $\alpha^a$ defined by \eqref{sympotpull2}, the general form of $\boldsymbol{\theta}$ given by eq.~(23) of \cite{IyerWald1994} shows that \eqref{eq:alpha_gen_form} holds with $k = m+1$, where $m$ is the maximum number of derivatives of the curvature appearing in the Lagrangian \eqref{eq:form_of_Lag}. Therefore, 
by \eqref{sympotpull} and the above theorem, we have
\be
\underline{\boldsymbol{\theta}} \horeq  \boldsymbol{\epsilon}^{(n-1)}  \sum_{i=0}^{m} \tilde{T}_{(i)}^{b_1 \cdots b_i cd} \nabla_{(b_1} \cdots \nabla_{b_i)} \mathscr{L}_\xi \delta g_{cd}.
\ee 
Define $\textbf{B}_{\mathcal H}$ on $\mathcal H$ by
\be
\textbf{B}_{\mathcal H} = \boldsymbol{\epsilon}^{(n-1)}  \sum_{i=0}^{m} \tilde{T}_{(i)}^{b_1 \cdots b_i cd} \nabla_{(b_1} \cdots \nabla_{b_i)} \mathscr{L}_\xi g_{cd} .
\label{BHform}
\ee
Then we have 
\be
\underline{\boldsymbol{\theta}} \horeq  \delta \textbf{B}_{\mathcal H}.
\label{eq:sym_pot_tot_var}
\ee
Namely, when we take a first variation of $\textbf{B}_{\mathcal H}$, we will get a factor of $\mathscr{L}_\xi g_{cd} = 0$ appearing in all terms except the terms where the variation acts on this factor. However, these terms manifestly yield $\underline{\boldsymbol{\theta}}$. 

The quantities $\tilde{T}_{(i)}^{b_1 \cdots b_i cd}$ needed to obtain $\textbf{B}_{\mathcal H}$ can be computed by following the steps used in the proof of the theorem. We will compute $\textbf{B}_{\mathcal H}$ explicitly for general relativity and for theories with Lagrangians that are quadratic in the curvature tensor in the next section. We will also discuss the ambiguities in the definition of $\textbf{B}_{\mathcal H}$ in the next section.

\section{Dynamical Black Hole Entropy}
\label{dbesec}

\subsection{General Definition of Dynamical Black Hole Entropy}
\label{dbedef}

In the previous section, we have proven that for a stationary black hole whose event horizon, ${\mathcal H}$, is a bifurcate Killing horizon there exists an $(n-1)$-form $\textbf{B}_{\mathcal H}$ on ${\mathcal H}$ with the property that
\eqref{eq:sym_pot_tot_var} holds for arbitrary perturbations. Following the motivation given in the Introduction, we define the {\em entropy $(n-2)$-form} $\textbf{S}$ on ${\mathcal H}$ by 
    \begin{align}
        \textbf{S} \equiv{}& \frac{2\pi}{\kappa} ( \textbf{Q}[\xi] - \xi \cdot \textbf{B}_{\mathcal H} ).
        \label{entform}
    \end{align} 
We define the entropy, $S[\mathcal C]$, of the black hole at a ``time'' represented by an arbitrary cross-section $\mathcal C$ by
    \begin{align}
        S[\mathcal C] \equiv{}& \int_{\mathcal C} \textbf{S} =  \frac{2\pi}{\kappa}  \int_{\mathcal C} ( \textbf{Q}[\xi] - \xi \cdot \textbf{B}_{\mathcal H} ). 
        \label{def:entropy_S}
    \end{align}  
Since the definition of $\textbf{S}$ does not depend on a choice of cross-section and $\textbf{S}$ is smooth on ${\mathcal H}$, it follows immediately that $S[\mathcal C]$ varies continuously with the cross-section $\mathcal C$ in the sense defined in \cite{Chen2022}. We now consider the ambiguities in our definition of $S[\mathcal C]$ for stationary black holes, for first order perturbations of stationary black holes, and for second and higher order deviations from a stationary black hole.

For the unperturbed stationary black hole, we have $\textbf{B}_{\mathcal H} = 0$ by \eqref{BHform}, so we have 
\be
S[\mathcal C] =  \frac{2\pi}{\kappa} \int_{\mathcal C} \textbf{Q}[\xi]. 
\label{Sstat}
\ee
Thus, the only possible ambiguities in $S[\mathcal C]$ for stationary black holes arise from ambiguities in the definition of $\textbf{Q}$, which are given by \cite{IyerWald1994}
\begin{align}
    \textbf{Q}[\xi] \rightarrow{}& \textbf{Q}[\xi] + \xi \cdot \boldsymbol{\mu}(\phi) + \textbf{Y}(\phi,\mathscr{L}_\xi \phi) + d\textbf{Z}(\phi, \xi).
    \label{Qamb}
\end{align}
Here the first two terms result from the ambiguity in $\boldsymbol{\theta}$
\begin{align}
    \boldsymbol{\theta}(\phi,\delta \phi) \rightarrow{}& \boldsymbol{\theta}(\phi,\delta \phi) + \delta \boldsymbol{\mu}(\phi) + d\textbf{Y}(\phi,\delta\phi) 
    \label{thamb}
\end{align}
arising from adding an exact term $d\boldsymbol{\mu}$ to the Lagrangian and an exact term $d\textbf{Y}$ to $\boldsymbol{\theta}$ (see \cite{IyerWald1994}). Since the $(n-1)$-form $\boldsymbol{\mu}$ is locally and covariantly constructed from the metric, we have $\mathscr{L}_\xi \boldsymbol{\mu} = 0$ in the stationary background. Since $\boldsymbol{\mu}$ smoothly extends to $\mathcal B$, it then follows from the type of argument given at the end of subsection \ref{sec:aux_vf_n} that the pullback of $\boldsymbol{\mu}$ to ${\mathcal H}$ vanishes in the stationary background. The $\textbf{Y}$ term in \eqref{Qamb} vanishes in the stationary background since $\mathscr{L}_\xi \phi = 0$. Finally, $d\textbf{Z}$ integrates to zero over any cross-section $\mathcal C$. Thus, $S[\mathcal C]$ is completely unambiguous in the stationary background. Furthermore, if the equations of motion hold in the stationary background, then $d \textbf{Q} = \textbf{J}$ by \eqref{eq:def_noe_curr_charge}. However, \eqref{noecur} implies that the pullback of $\textbf{J}$ to ${\mathcal H} $ vanishes, because $\boldsymbol{\theta}(\phi,\mathscr{L}_\xi \phi) = 0$ (since $\xi^a$ is a Killing field) and the pullback of $\xi \cdot \textbf{L}$ to ${\mathcal H} $ vanishes (since $\xi^a$ is tangent to ${\mathcal H}$). Thus, we find that the pullback of $d \textbf{Q}$ to ${\mathcal H}$ vanishes. This implies that $S[\mathcal C]$ does not depend on $\mathcal C$ for a stationary black hole, i.e., the entropy is ``time independent.'' In particular, we may evaluate $S[\mathcal C]$ on the bifurcation surface $\mathcal B$, in which case it reduces to the definition given in \cite{Wald1993,IyerWald1994}. Thus, our definition \eqref{def:entropy_S} reduces to the standard expression for entropy for stationary black holes on any cross-section $\mathcal C$.

For first order perturbations of a stationary black hole, we have
\be
\label{mainSformula}
\delta S[\mathcal C] =  \frac{2\pi}{\kappa} \int_{\mathcal C} (\delta \textbf{Q}[\xi] - \xi \cdot \boldsymbol{\theta})
\ee
where \eqref{eq:sym_pot_tot_var} was used to replace $\delta \textbf{B}_{\mathcal H}$ by $\boldsymbol{\theta}$. The only possible ambiguities in this expression result from the ambiguities \eqref{Qamb} and \eqref{thamb}.
Thus, the ambiguity in $\delta \textbf{S}$ is 
\begin{align}
    \delta\textbf{S} \rightarrow{}& \delta\textbf{S} +\frac{2 \pi}{\kappa} \left[ \xi \cdot \delta \boldsymbol{\mu} + \textbf{Y}(\phi,\mathscr{L}_\xi\delta\phi) + d\delta \textbf{Z}(\phi,\xi) -  \xi \cdot \delta \boldsymbol{\mu}  - \xi \cdot d\textbf{Y}(\phi,\delta\phi) \right] \\
    ={}& \delta\textbf{S} + \frac{2 \pi}{\kappa} \left[\mathscr{L}_\xi \textbf{Y}(\phi,\delta\phi) + d\delta \textbf{Z}(\phi,\xi) - \xi \cdot d\textbf{Y}(\phi,\delta\phi) \right] \\
   ={}& \delta\textbf{S} + \frac{2 \pi}{\kappa} \left[ d (\xi \cdot \textbf{Y})(\phi,\delta\phi) + d\delta \textbf{Z}(\phi,\xi) \right] 
\label{dsamb}
\end{align}
where the general identity
\begin{align}
    \mathscr{L}_\xi \boldsymbol{\alpha} = \xi \cdot d\boldsymbol{\alpha} + d(\xi \cdot \boldsymbol{\alpha}) \label{eq:Lie_d_id}
\end{align} 
for any form $\boldsymbol{\alpha}$ was used in the last line. The exact terms in \eqref{dsamb} integrate to zero on any cross-section. Thus, for arbitrary first order perturbations, $\delta S[\mathcal C]$ is entirely unambiguous. As we will see explicitly below, in the case of general relativity our formula for $\delta S[\mathcal C]$ does not agree with the Bekenstein--Hawking entropy formula in non-stationary eras.

For second order perturbations, there will, in general, be substantial additional ambiguities resulting from ambiguities in defining $\textbf{B}_{\mathcal H}$ at second order. As we shall now explain, there are two sources of these ambiguities. 

The first source has to do with the essential presence of $\xi^a$ in the formula for $\textbf{B}_{\mathcal H}$. At first order, the perturbed metric enters the formula for $\textbf{B}_{\mathcal H}$ only via the factor $\mathscr{L}_\xi \delta g_{ab}$, so all of the other quantities in that formula take their background values. Since we require $\delta \xi^a = 0$, $\xi^a$ is rigidly fixed on spacetime to be the horizon Killing field of the background metric $g_{ab}$. Similarly, $N^a$ on $\mathcal H$ is rigidly fixed by \eqref{Nprop1} and \eqref{Nprop2}. However, the varied metric will not, in general, admit any Killing fields at all, and, at second and higher orders in perturbation theory, the formula for $\textbf{B}_{\mathcal H}$ will depend upon how $\xi^a$ is chosen. We can (and will) require as a gauge choice that $\xi^a$ is a fixed vector field on spacetime, but unless $g_{ab}$ is rigidly tied to $\xi^a$, the formula for $\textbf{B}_{\mathcal H}$ will have a gauge dependence with regard to the metric. In order to avoid this, we must provide a prescription to rigidly fix $\xi^a$ in terms of the metric in a situation where the metric no longer possesses a horizon Killing field.
Without loss of generality, we may assume that the true event horizon of the perturbed spacetime coincides with $\mathcal H$, that $\xi^a$ remains normal to ${\mathcal H}$, and that $\xi^b \nabla_b \xi^a = \kappa \xi^a$ on ${\mathcal H}$. This rigidly fixes $\xi^a$ on $\mathcal H$ up to a choice of cross-section $\mathcal B$ of $\mathcal H$ in the varied spacetime corresponding to the bifurcation surface. However, since as many as $m+1$ derivatives of $\xi^a$ occur in \eqref{BHform}, we must also give a construction of $\xi^a$ off of $\mathcal H$ at least up to this order for the case of a metric that does not possess a horizon Killing field. Once we have determined $\xi^a$, we may define $N^a$ on $\mathcal H$ by \eqref{frob} and the first equalities of \eqref{Nform} and \eqref{cform}.

A simple and natural way to rigidly tie $g_{ab}$ to $\xi^a$ in a neighborhood of $\mathcal H$ would be to introduce Gaussian null coordinates (GNCs) on $\mathcal H$ 
adapted to a choice of affine parameter $V$ on the null geodesics of $\mathcal H$, with $V=0$ corresponding to $\mathcal B$. In these coordinates, the metric takes the form
\begin{equation}
    ds^2 = 2dV(d\rho - \rho\beta_A dx^A - \tfrac{1}{2}\rho^2 \alpha dV) + \gamma_{AB}dx^A dx^B,
    \label{affgnc}
\end{equation}
where $x^A$ are coordinates on a cut of constant $V$ and $\rho$ is an affine parameter along the null geodesics transverse to $\mathcal H$ that are orthogonal to these cuts. We may then define $\xi^a$ in a neighborhood of $\mathcal H$ by
\be
\xi^a = \kappa\left[V \left(\frac{\partial}{\partial V}\right)^a - \rho \left(\frac{\partial}{\partial \rho}\right)^a \right],
\label{xidef}
\ee
and we may then apply a suitable diffeomorphism to $g_{ab}, \xi^a$ such that $\xi^a$ is the same in each spacetime\footnote{In other words, we identify different spacetimes by identifying points near $\mathcal H$ with the same coordinates $(V,\rho,x^A)$.}.

For the case of a metric with a horizon Killing field with surface gravity $\kappa$ and bifurcation surface $\mathcal B$, this definition reproduces the horizon Killing field \cite{Hollands_Wald_2013}, but this formula is well defined in a neighborhood of $\mathcal H$ for an arbitrary metric. In addition, the metric in this gauge satisfies conditions that yield \eqref{gaugecon1} and \eqref{gaugecon2} for perturbations about a metric with a horizon Killing field.
However, there is a rescaling freedom $V \to f(x^A)V$ in the choice of $V$. Although this rescaling freedom does not affect the definition \eqref{xidef} of $\xi^a$ in the case where the metric possesses a horizon Killing field with surface gravity $\kappa$ and bifurcation surface $\mathcal B$, it will affect the definition of $\xi^a$ in the general case. Thus, even with this choice of prescription, we have ambiguities in the definition of $\xi^a$ and corresponding ambiguities in the definition of $\textbf{B}_{\mathcal H}$. Of course, other prescriptions for defining $\xi^a$ off of $\mathcal H$ would also be possible. In appendix \ref{app:C} we will give an alternative construction of $\xi^a$.

It is not necessarily a bad thing that $\textbf{B}_{\mathcal H}$---and thereby our definition of entropy---depends on a choice of ``time translation'' $\xi^a$. It is generally recognized that the notion of ``energy'' should be defined as conjugate to a notion of time translations. Since ``energy'' and ``entropy'' appear together in the first law of thermodynamics for near-equilibrium systems, it seems reasonable that the notion of ``entropy'' should similarly depend on a notion of time translations. Nevertheless, the ambiguities in defining a notion of ``time translation'' gives rise to ambiguities in our definition of entropy beginning at second order in perturbation theory about a black hole whose event horizon is a Killing horizon.

The second source of ambiguities in the formula for entropy at second order arises from the fact that $\textbf{B}_{\mathcal H}$ was constructed so as to satisfy \eqref{eq:sym_pot_tot_var}, which is a condition that applies only to first order perturbations. Therefore, we have freedom to add to $\textbf{B}_{\mathcal H}$ terms that vanish identically for first order perturbations about a background spacetime with a Killing horizon but do not vanish at second order. Nevertheless, based upon the construction of $\textbf{B}_{\mathcal H}$ given in the previous section, it seems natural to impose the following conditions on $\textbf{B}_{\mathcal H}$:
\begin{enumerate}
    \item We require $\textbf{B}_{\mathcal H}$ to be locally and covariantly constructed out of the metric, the curvature tensor, derivatives of the curvature tensor, and the vector fields $\xi^a$, $N^a$ and their covariant derivatives\footnote{Note that $N^a$ on ${\mathcal H}$ itself is locally and covariantly constructed from the metric, $\xi^a$, and the first derivative of $\xi^a$, since it is uniquely determined in a local manner by \eqref{frob} and the first equalities of \eqref{Nform} and \eqref{cform}. It would be consistent with our construction of $\textbf{B}_{\mathcal H}$ in section \ref{Bconstruction} at first order to require that no derivatives of $N^a$ appear in the expression for $\textbf{B}_{\mathcal H}$.}.
\item We require that each term in the expression for $\textbf{B}_{\mathcal H}$ contains a factor\footnote{When matter fields are present, there may be terms proportional to the matter field contracted into $\xi^a$ that do not contain a factor involving $\mathscr{L}_\xi$.} of $\mathscr{L}_\xi g_{ab}$ (or covariant derivatives of this quantity). In particular, we have $\textbf{B}_{\mathcal H} = 0$ in the stationary background.\footnote{It would be consistent with our construction of $\textbf{B}_{\mathcal H}$ in section \ref{Bconstruction} at first order to require that there be no more than $m$ covariant derivatives of $\mathscr{L}_\xi g_{ab}$.}
    \item We require that \eqref{eq:sym_pot_tot_var} holds for arbitrary perturbations of a stationary background.
    \item We require that the total number of derivatives appearing in $\textbf{B}_{\mathcal H}$ cannot exceed the total number of derivatives in $\boldsymbol{\theta}$, which is one fewer than the total number of derivatives appearing in the Lagrangian. (Here, a factor of the curvature tensor counts as ``two derivatives'' in the Lagrangian and in $\textbf{B}_{\mathcal H}$. The factor of $\mathscr{L}_\xi g_{ab}$ in $\textbf{B}_{\mathcal H}$ counts as ``one derivative.'' However, the appearance of $N^a$ does not count towards the derivative limitation in $\textbf{B}_{\mathcal H}$---despite its being determined by the derivative of $\xi^a$---since its presence can also arise from the purely algebraic basis expansion.)
    \item We require that $\textbf{B}_{\mathcal H}$ is invariant under a rescaling of $\xi^a$ by a constant, $\xi^a \to c \xi^a$ and the corresponding rescaling $\kappa \to c \kappa$.
\end{enumerate}

These conditions are quite restrictive for theories with a relatively low total number of derivatives in the Lagrangian. In particular, they uniquely determine $\textbf{B}_{\mathcal H}$ in general relativity.
However, by the time one reaches curvature cubed theories (i.e., a total of 6 derivatives in the Lagrangian), it is not difficult to construct terms---such as ${\epsilon}_{ab_1 \cdots b_{n-1}} R_{pq} N^p N^q {R^a}_{c} \xi^c g^{mn} \mathscr{L}_\xi g_{mn}$---that could be added to $\textbf{B}_{\mathcal H}$ so as satisfy all of the above conditions and are nonvanishing at second order. It is possible that one could impose additional conditions on $\textbf{B}_{\mathcal H}$ so as to further restrict its ambiguities and/or obtain other desired properties of $\textbf{S}$. One attractive possibility would be to impose conditions that imply that $S[\mathcal C]$ satisfies a second law (i.e., that it is non-decreasing with time) at second order (see subsection \ref{vacpert}). However, we do not believe that it would be possible to do this within our framework.

It is instructive to compare the nature of our framework and its definition freedom arising at second order with the framework and definition freedom in the approach of Hollands, Kovacs, and Reall (HKR) \cite{Hollands_2022}. As will be discussed further in appendix \ref{hkr}, HKR take an effective field theory approach, wherein the Lagrangian is expanded about the general relativity Lagrangian in powers of a small parameter $\ell$ and one works only to some finite order in $\ell$. The HKR entropy-current $(n-2)$-form $\textbf{S}_{\rm HKR}$ is constructed so as to agree with the Dong--Wall entropy (see section \ref{dwent}) rather than our entropy current to first order in perturbation theory about a black hole with a bifurcate Killing horizon. The HKR construction is aimed at
providing an entropy satisfying the 2nd law to second order in perturbation theory \cite{Hollands_2022}, or even non-perturbatively in a recent extension \cite{Davies:2023qaa}. But it need not be covariant in the metric \cite{Davies:2022xdq}. As discussed further in appendix \ref{hkr}, the lack of a covariant HKR entropy current $(n-2)$-form means that there is no guarantee that the HKR entropy will be ``cross-section continuous,'' \cite{Chen2022} i.e., if we ``wiggle'' the cut $\mathcal C$ into another cut $\mathcal C'$ that is arbitrarily close to $\mathcal C$, there is no guarantee that $S_{\rm HKR}[\mathcal C']$ will be close to $S_{\rm HKR}[\mathcal C]$, as illustrated in Fig.~\ref{fig:1}. Although we will see in section \ref{dwent} that (our and) the Dong--Wall entropy is cross-section continuous, it appears that the HKR entropy will not, in general, satisfy this property.  A detailed comparison of the definition freedom in entropy at second order in our approach and the HKR approach is given in appendix \ref{hkr}.

\begin{figure}[h!]
\begin{center}
  \includegraphics[trim={0 5cm 0 4cm},clip, width = 0.5\textwidth]{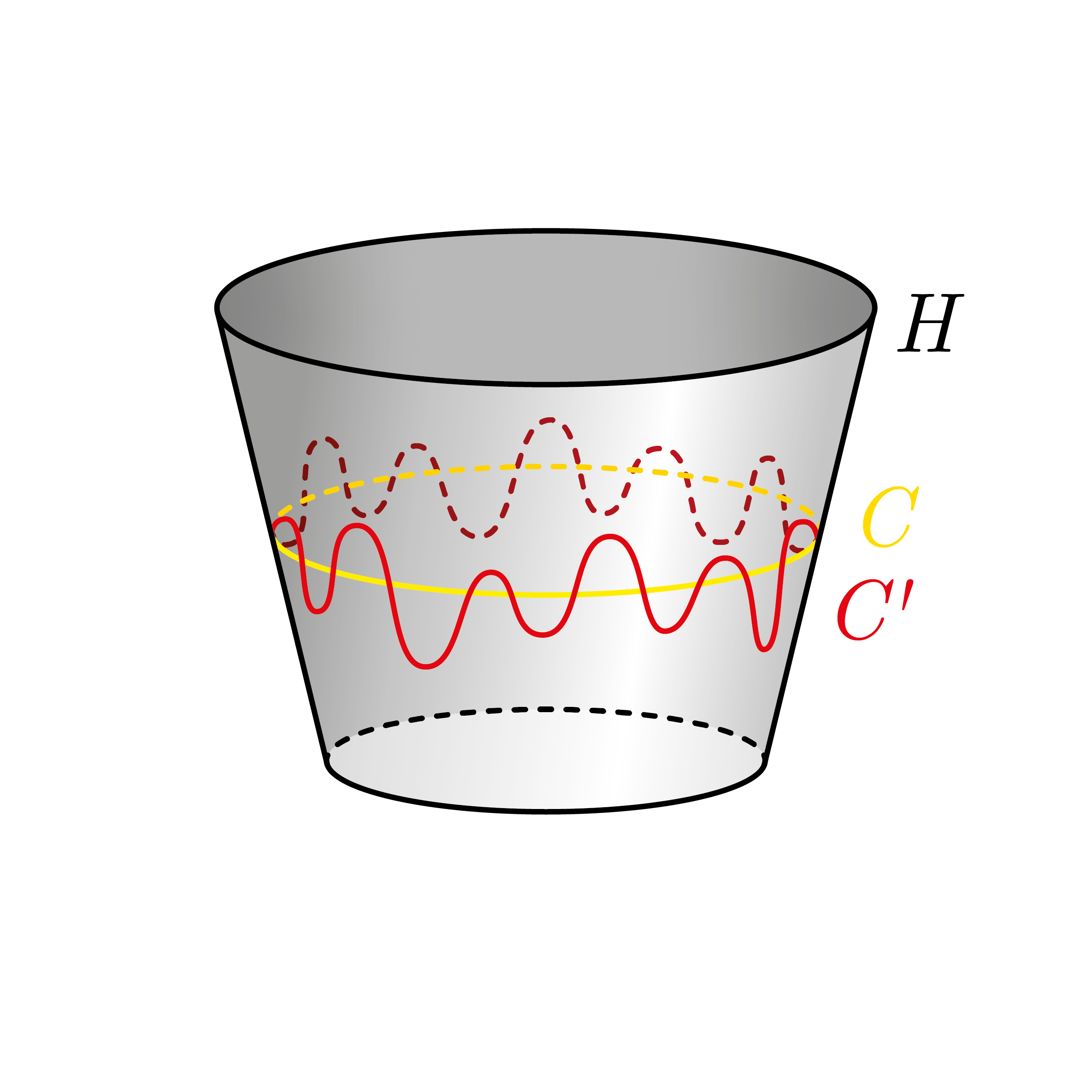}
  \end{center}
  \caption{Cross section continuity \cite{Chen2022} of $S[\mathcal{C}]$ means that $S[{\mathcal C}']$ approaches $S[{\mathcal C}]$ for a cut ${\mathcal C}'$ tending to a cut $\mathcal C$, even if ${\mathcal C}'$ is very wiggly as illustrated. 
  }
  \label{fig:1}
\end{figure}

Finally, we comment on the definition of $\textbf{S}$ and $S[\mathcal C]$ outside of the perturbative context, i.e., for an arbitrary dynamical black hole. We have already noted above that our approach to defining entropy is heavily reliant on the presence of a time translation vector field $\xi^a$. In the case of a stationary black hole whose event horizon is a Killing horizon, the horizon Killing field $\xi^a$ provides a natural notion of ``time translations'' upon which to base notions of energy and entropy. For a black hole that is very nearly stationary---in particular, in first and second order perturbation theory about a stationary black hole---it also appears to be natural to use the definition of $\xi^a$ given above to play the role of ``time translations.'' 
However, for a black hole that is far from ``equilibrium,'' it seems quite artificial to attempt to assign a surface gravity $\kappa$ and bifurcation surface $\mathcal B$ to the event horizon $\mathcal H$ so as to construct a Killing-like vector field $\xi^a$. Furthermore, even in ordinary physics where a notion of time translation is present, the definition of the entropy of a system that is far from equilibrium is not well defined without making a choice of coarse-grained observable \cite{Wald1979}. Thus, we do not advocate applying our formulas for $\textbf{S}$ and $S[\mathcal C]$ outside of the perturbative context around a stationary (i.e., ``equilibrium'') black hole.

We now evaluate our formula for $\textbf{S}$ and $S[\mathcal C]$ for the case of general relativity and, more generally, for theories whose Lagrangians do not depend on derivatives of the curvature tensor.

\subsection{Dynamical Black Hole Entropy in General Relativity}

The Lagrangian for general relativity is 
\begin{align}
   {L}_{a_1 \cdots a_n} = {}& \frac{1}{16\pi} {\epsilon}_{a_1 \cdots a_n} R. \label{eq:gr_Lagrangian}
\end{align}
The symplectic potential $\boldsymbol{\theta}$ is \cite{IyerWald1994} 
\begin{align}
    {\theta}_{a_1 \cdots a_{n-1}} ={}&\frac{1}{16\pi} {{\epsilon}^e}_{a_1 \cdots a_{n-1}} g^{fh}(\nabla_f \delta g_{eh} - \nabla_e \delta g_{fh})
\end{align}
and the Noether charge is \cite{IyerWald1994} 
\begin{align}
    {Q}_{a_1 \cdots a_{n-2}} ={}& -\frac{1}{16\pi} {\epsilon}_{a_1 \cdots a_{n-2}cd} \nabla^c \xi^d .
\end{align} 
Using the methods employed in section \ref{Bconstruction} to obtain $\textbf{B}_{\mathcal H}$ from $\boldsymbol{\theta}$, we obtain 
\begin{align}
{{B}_{\mathcal H}}_{a_1 \cdots a_{n-1}}={}& -\frac{1}{16\pi} N^b {\epsilon}_{ba_1 \cdots a_{n-1}}
 g^{ef}\mathscr{L}_\xi g_{ef} \, .
 \label{eq:b_gr}
\end{align} 
It is easily seen that ${\textbf{B}_{\mathcal H}}$ satisfies all of the requirements enumerated above in subsection \ref{dbedef}. Furthermore, ${\textbf{B}_{\mathcal H}}$ is unique, since it is allowed to contain only one derivative---which is used up by the required factor of $\mathscr{L}_\xi g_{ef}$---and no terms of the required form can be constructed that vanish for arbitrary first order perturbations.
In terms of the affinely parametrized tangent $k^a = \xi^a/(\kappa V)$ (see \eqref{afftang}), we can rewrite ${\textbf{B}_{\mathcal H}}$ as
\begin{align}
{{B}_{\mathcal H}}_{a_1 \cdots a_{n-1}}={}& -\frac{1}{16\pi} n^b {\epsilon}_{ba_1 \cdots a_{n-1}}
 g^{ef}\mathscr{L}_k g_{ef} = - \frac{1}{8\pi} n^b  {\epsilon}_{ba_1 \cdots a_{n-1}}  \vartheta
 \label{eq:b_gr2}
\end{align} 
where $\vartheta$ is the expansion of the null generators of the horizon in the affine parameterization $V$ and $n^a = \kappa V N^a$ (so $n^a k_a = 1$). This agrees with eq.~(6.24) in \cite{ChandrasekaranFP_2018} (see also \cite{shi2021covariant}) for the corresponding term in the symplectic potential on an arbitrary null surface needed to implement the Wald--Zoupas prescription for defining charges.

The entropy $(n-2)$-form \eqref{entform} in general relativity is 
\begin{align}
    {S}_{a_1 \cdots a_{n-2}}={}& -\frac{1}{8\kappa}{\epsilon}_{a_1 \cdots a_{n-2}cd}\nabla^c \xi^d - {\epsilon}^{(n-2)}_{a_1 \cdots a_{n-2}} \frac{1}{4} V \vartheta
\end{align} 
where ${\epsilon}^{(n-2)}_{a_1 \cdots a_{n-2}} = \xi^b N^c {\epsilon}_{bca_1 \cdots a_{n-2}}$ was defined by \eqref{eq:or_conv_pt_1} and \eqref{eq:or_conv_pt_2} and is such that its pullback to any cross-section $\mathcal C$ on the horizon yields the volume element on $\mathcal C$ obtained from the pullback of the metric to $\mathcal C$. 
For a given cross section $\mathcal C$ of the horizon, the entropy is 
\begin{align}
    S[\mathcal C] ={}& \frac{A[\mathcal C]}{4} - \frac{1}{4} \int_{\mathcal C} V \vartheta \, \boldsymbol{\epsilon}^{(n-2)}
    \label{eq:entropy_gr}
\end{align} 
where $A[\mathcal C]$ is the ($(n-2)$-dimensional) area of $\mathcal C$ and the integral over $\mathcal C$ in the second term is taken with respect to the natural volume element ${\epsilon}^{(n-2)}_{a_1 \cdots a_{n-2}}$ on $\mathcal C$. 
The first term in \eqref{eq:entropy_gr} is the usual Bekenstein--Hawking entropy. For a stationary black hole, we have $\vartheta = 0$. However, for a dynamical black hole, we have $\vartheta >0$, so $S[\mathcal C]$ is \textit{smaller} than the Bekenstein--Hawking entropy. We will show in appendix \ref{sec:area_app_horizon} that for first-order perturbations, $S[\mathcal C]$ in general relativity is, in fact, equal to the area of the apparent horizon at the ``time'' corresponding to the cross-section, $\mathcal C$, of the horizon. 

Rignon-Bret \cite{Rignon-Bret_2023} (see also \cite{Ciambelli:2023mir}) has given an analysis of aspects of our entropy expression \eqref{eq:entropy_gr} as well as an alternative entropy expression based on the York symplectic potential that is constructed to vanish on any cross section of a light cone in Minkowski spacetime.

\subsection{Entropy Form for $L(R_{abcd})$ Theories}

Consider a theory of gravity obtained from a Lagrangian that depends on the metric and curvature but does not depend on derivatives of the curvature
\begin{align}
    {L}_{a_1 \cdots a_n} ={}& {\epsilon}_{a_1 \cdots a_n} L(g_{ab}, R_{abcd}).
\end{align} 

The symplectic potential $(n-1)$-form can be calculated from the algorithm given in Lemma 3.1 of \cite{IyerWald1994} or directly from\footnote{The sign difference between our formula and eq.~(4.5) of \cite{Compere_2009} appears to be a result of differences in conventions.} eq.~(4.5) of \cite{Compere_2009}. We obtain,
\begin{align}
    {\theta}_{a_1 \cdots a_{n-1}} ={}& 2{\epsilon}_{ma_1 \cdots a_{n-1}} \left[ \frac{\partial L}{\partial R_{mbcd}} \nabla_d \delta g_{bc} - \nabla_d \left( \frac{\partial L}{\partial R_{mbcd}} \right) \delta g_{bc} \right].
\end{align}
The Noether charge $(n-2)$-form can be obtained either from the Noether current $\textbf{J}$ via the algorithm given in Lemma 1 of \cite{Wald1990} or directly from eq.~(4.6) of \cite{Compere_2009},
\begin{align}
    {Q}_{a_1 \cdots a_{n-2}} ={}&  {\epsilon}_{cda_1 \cdots a_{n-2}}  \Big(-\frac{\partial L}{\partial R_{cdef}} \nabla_{[e} \xi_{f]} - 2 \nabla_f \Big(\frac{\partial L}{\partial R_{cdef}}\Big) \xi_e \Big).
\end{align} From the construction in section \ref{Bconstruction}, we obtain
\begin{align}
     {{B}_{\mathcal H}}_{a_1 \cdots a_{n-1}} ={}&2 N^f {\epsilon}_{fa_1 \cdots a_{n-1}}  \xi_m \frac{\partial L}{\partial R_{mbcd}} N_d \mathscr{L}_\xi g_{bc}.
\end{align} The entropy $(n-2)$-form is then
\begin{align}
    {S}_{a_1 \cdots a_{n-2}} ={}& \frac{2\pi}{\kappa}\Big[ {\epsilon}_{cda_1 \cdots a_{n-2}}  \Big(-\frac{\partial L}{\partial R_{cdef}} \nabla_{[e} \xi_{f]} - 2 \nabla_f \Big(\frac{\partial L}{\partial R_{cdef}}\Big) \xi_e \Big) + 2 {\epsilon}^{(n-2)}_{a_1 \cdots a_{n-2}}  \xi_m \frac{\partial L}{\partial R_{mbcd}} N_d \mathscr{L}_\xi g_{bc} \Big].
\end{align}

We now explicitly evaluate the above quantities for the case of a general curvature squared Lagrangian, i.e.,
\begin{align}
    {L}_{a_1 \cdots a_n} ={}& {\epsilon}_{a_1 \cdots a_n} \left( \alpha_1 R_{abcd}R^{abcd} + \alpha_2 R_{ab}R^{ab} + \alpha_3 R^2\right),
\end{align} 
where $\alpha_1, \alpha_2, \alpha_3$ are arbitrary constants.
The symplectic potential $(n-1)$-form is\footnote{The choice of $\alpha_1 = 1$, $\alpha_2 = -4$, $\alpha_3 = 1$ corresponds to Gauss--Bonnet gravity. There is a typo---previously noted in \cite{Bueno_2017}---in the sign of the second-to-last term in the formula for $\boldsymbol{\theta}$ for Gauss--Bonnet gravity given in eq.~(70) of \cite{IyerWald1994}.} 
\begin{align}
 \begin{split}
	{\theta}_{a_1 \cdots a_{n-1}} ={}& {\epsilon}_{ma_1 \cdots a_{n-1}} \Big\{  4\alpha_1 \Big[R^{mbcd}\nabla_d \delta g_{bc} - \nabla_d R^{mbcd} \delta g_{bc} \Big]\\
 {}&+\alpha_2 \Big[2g^{c[m}R^{b]d}\nabla_d \delta g_{bc} + 2 R^{c[m} \nabla^{b]} \delta g_{bc} + g^{c[b} \nabla^{m]} R \delta g_{bc}  + 2\nabla^{[m} R^{b]c}\delta g_{bc} \Big] \\
 {}&+2 \alpha_3 \Big[R g^{mc}g^{bd}(\nabla_d \delta g_{bc} - \nabla_c \delta g_{bd}) - 2g^{c[m}\nabla^{b]} R \delta g_{bc}  \Big] \Big\}.
	\end{split}
\end{align}
The Noether charge $(n-2)$-form is
\begin{align}
\begin{split}
{Q}_{a_1 \cdots a_{n-2}} ={}& -{\epsilon}_{cda_1 \cdots a_{n-2}} \Big[(2\alpha_1 R^{cdef} + 2\alpha_2 g^{ce}R^{fd}  + 2\alpha_3 R g^{ce}g^{df} ) \nabla_{[e} \xi_{f]} \\
{}&+   (4 \alpha_1 \nabla_f R^{cdmf} + \alpha_2 g^{cm}\nabla^d R  - 2\alpha_2 \nabla^c R^{md} + 4\alpha_3 \nabla^d R g^{cm}) \xi_m \Big] .
\end{split}
\end{align} 
The $\textbf{B}_{\mathcal H}$ $(n-1)$-form is
\begin{align}
{B}_{\mathcal H a_1 \cdots a_{n-1}}={}& N^f{\epsilon}_{fa_1 \cdots a_{n-1}} \left(4 \alpha_1 R^{abcd} \xi_a N_d - \alpha_2 R^{bc} - \alpha_2 R_{ad} \xi^a N^d g^{bc} - 2 \alpha_3 R g^{bc} \right) \mathscr{L}_\xi g_{bc}.
\label{BHR2}
\end{align}
The entropy $(n-2)$-form is 
\begin{align}
\begin{split}
    {S}_{a_1 \cdots a_{n-2}} ={}& -\frac{2\pi}{\kappa} \Big({\epsilon}_{cda_1 \cdots a_{n-2}} \Big[(2\alpha_1 R^{cdef} + 2\alpha_2 g^{ce}R^{fd}  + 2\alpha_3 R g^{ce}g^{df} ) \nabla_{[e} \xi_{f]} \\
    {}&+   (4 \alpha_1 \nabla_f R^{cdmf} + \alpha_2 g^{cm}\nabla^d R  - 2\alpha_2 \nabla^c R^{md} + 4\alpha_3 \nabla^d R g^{cm}) \xi_m \Big] \\
    {}& - {\epsilon}^{(n-2)}_{a_1 \cdots a_{n-2}}(4 \alpha_1 R^{abcd} \xi_a N_d - \alpha_2 R^{bc} - \alpha_2 R_{ad} \xi^a N^d g^{bc} - 2 \alpha_3 R g^{bc} ) \mathscr{L}_\xi g_{bc} \Big).
\end{split} \label{eq:S_form_curv2}
\end{align}
As will be discussed in section \ref{dwent}, this yields a formula for $S[\mathcal C]$ that differs from the Dong--Wall entropy in a manner analogous to how our formula for $S[\mathcal C]$ in general relativity differs from the Bekenstein--Hawking entropy.

\section{The Physical Process First Law and the Second Law}
\label{pp12}

As explained in the Introduction, the main motivation for our new definition of dynamical black hole entropy is to obtain a local in time version of the physical process version of the first law of black hole mechanics. In this section, we derive the physical process version of the first law and discuss the second law of black hole mechanics.
As stated at the end of section \ref{lagkh}, we restrict consideration to Lagrangians for which the only dynamical field is the metric. However, in this context, there are two distinct cases that we consider. 

The first case is where an external stress-energy is present in the first order perturbation. Here, by an ``external stress-energy,'' we mean a source term in the first order gravitational equations arising from matter that is not represented in the Lagrangian. In this case, as we shall see, there can be nontrivial time dependence of the dynamical black hole entropy $S$ at first order when the external matter crosses the horizon. We show that the entropy change is related to the energy flux by the first law of black hole mechanics. As an immediate corollary, we see that if the external matter satisfies the null energy condition, then the second law of black hole mechanics holds, i.e., the entropy is non-decreasing with time in first order perturbation theory. These results hold for an arbitrary theory of gravity.

The second case is that of pure vacuum perturbations. In this case, $S$ does not change with time at first order, so we have to go to second order in perturbation theory to investigate the time dependence of $S$. We find that at second order, results similar to the external matter case hold, with the first order energy flux of external matter replaced by the second order flux of ``modified canonical energy'' \cite{Hollands_Wald_2013}. The flux of modified canonical energy is positive in general relativity---it is proportional to the square of the first order perturbed shear of the horizon---so the second law of black hole mechanics holds at second order in general relativity. However, the flux of modified canonical energy would not be expected to be positive in general theories of gravity, so we would not expect a second law to hold for vacuum perturbations in general theories.

\subsection{Perturbations with External Matter Sources}
\label{extmat}

For an arbitrary Lagrangian theory of gravity where the only dynamical field is the metric, we consider a stationary black hole with horizon Killing field $\xi^a$ that satisfies the vacuum field equations $({E}_G)_{ab} = 0$. We consider perturbations, $\delta g_{ab}$, of this spacetime that are sourced by a matter stress-energy tensor, $\delta T_{ab}$, that is conserved, $\nabla^a \delta T_{ab} = 0$. Thus, $\delta g_{ab}$ satisfies
\be
2 \delta ({ E}_G)_{ab} = \delta T_{ab}.
\label{lineom}
\ee 
The perturbed matter energy flux $(n-1)$-form $\textbf{e}$ on the horizon is given by 
\be
\delta e_{a_1 \cdots a_{n-1}} \equiv -\xi^b \delta {T_b}^c {\epsilon}_{ca_1 \cdots a_{n-1}} = \delta T_{bc} \xi^b \xi^c {\epsilon}^{(n-1)}_{a_1 \cdots a_{n-1}}
\ee
where ${\epsilon}^{(n-1)}_{a_1 \cdots a_{n-1}}$ was defined by \eqref{eq:or_conv_pt_1}.
We have the following theorem:

\begin{theorem}[Physical process version of the first law \cite{Gao_Wald_2001}]
\label{thm:phys_pro_ver_1st_law}
Let $\mathcal C_1$ and $\mathcal C_2$ be arbitrary cross-sections of the horizon $\mathcal H$ with $\mathcal C_2$ lying to the future of $\mathcal C_1$. Let $\Delta \delta S$ denote the perturbed entropy difference between the cross-sections
\be
\Delta \delta S = \delta S[\mathcal C_2] - \delta S[\mathcal C_1]
\ee
and let $\Delta \delta E$ denote the perturbed total energy flux between the cross-sections 
\be
\Delta \delta E = \int_{\mathcal{H}_{12}} \delta \textbf{\textup{e}} = -\int_{\mathcal{H}_{12}} \xi^b \delta {T_b}^c {\epsilon}_{ca_1 \cdots a_{n-1}} = \int_{\mathcal{H}_{12}} \delta T_{bc} \xi^b \xi^c {\epsilon}^{(n-1)}_{a_1 \cdots a_{n-1}}
\label{enflux}
\ee
where ${\mathcal{H}_{12}}$ denotes the portion of $\mathcal{H}$ bounded by $\mathcal C_1$ and $\mathcal C_2$. Then we have
\begin{align}
    \frac{\kappa}{2\pi} \Delta \delta S ={}& \Delta \delta E .
\label{eq:phys_pro_ver_local_eqn}
\end{align} 
\end{theorem}
\begin{proof}
The pullback to the horizon of the fundamental identity \eqref{eq:fund_id} with $\chi^a$ taken to be the horizon Killing vector field $\xi^a$ yields
    \begin{align}
       d[\delta \textbf{Q}[\xi] - \xi \cdot \boldsymbol{\theta}] ={}& -\xi^a \delta {\textbf{C}}_a .
       \label{fundid3}
    \end{align} 
Using the definition \eqref{entform} of $\textbf{S}$ together with \eqref{eq:sym_pot_tot_var}, we can rewrite the left side of \eqref{fundid3} as
    \begin{align}
      d[\delta \textbf{Q}[\xi] - \xi \cdot \boldsymbol{\theta}]  ={}& \frac{\kappa}{2\pi} d \delta \textbf{S} .
    \end{align} 
Using the general formula \eqref{eq:constraints_eom} with ${ E}_M = 0$ for $ \boldsymbol{\textbf{C}}_a$ together with \eqref{lineom}, we can rewrite the right side of \eqref{fundid3} as 
\be    
    -\xi^a \delta {\textbf{C}}_a = \delta \textbf{e}
\ee
so we have 
\be
\frac{\kappa}{2\pi} d \delta \textbf{S} = \delta \textbf{e}.
\ee
Integration of this equation over the region $\mathcal{H}_{12}$ of the horizon bounded by $\mathcal C_1$ and $\mathcal C_2$ then immediately yields \eqref{eq:phys_pro_ver_local_eqn}.

\end{proof}

\noindent
{\em Remark:} In the case of general relativity, \eqref{eq:phys_pro_ver_local_eqn} yields
\be
\frac{\kappa}{8 \pi} \Delta \left[  A 
- \int_{\mathcal C} V \vartheta \, \boldsymbol{\epsilon}^{(n-2)}\right] = \Delta E.
\ee
This result can also be derived by multiplying the Raychaudhuri equation by $\kappa V$ and integrating the resulting expression between the two cross-sections \cite{Mishra:2017sqs,Visser:2024pwz}.

\bigskip

There are two important immediate consequences of Theorem \ref{thm:phys_pro_ver_1st_law}. First, if $\delta T_{bc}$ satisfies the null energy condition (i.e., if $\delta T_{bc} n^b n^c \geq 0$ for all null $n^a$), then it follows immediately from \eqref{enflux} that $\Delta \delta E \geq 0$. It then follows immediately from Theorem \ref{thm:phys_pro_ver_1st_law} that $\Delta \delta S \geq 0$, i.e., $ \delta S[\mathcal C_2] \geq \delta S[\mathcal C_1]$ whenever $\mathcal C_2$ lies to the future of $\mathcal C_1$. Thus, in an arbitrary theory of gravity, the second law of black hole mechanics holds to first order for perturbations of a black hole with external matter satisfying the null energy condition. 

Second, Theorem \ref{thm:phys_pro_ver_1st_law} shows that, to first order, the entropy of a black hole changes between cross-sections $\mathcal C_1$ and $\mathcal C_2$ if and only if there is a net energy flux of matter between $\mathcal C_1$ and $\mathcal C_2$. Thus, the increase in entropy of the black hole has an entirely local cause. In the case of general relativity, this contrasts sharply with the properties of the Bekenstein--Hawking entropy, $A/4$. The event horizon of a black hole is ``teleological'' in nature, i.e., its location is determined by properties of the spacetime in the asymptotic future. If one is going to throw matter into a black hole at a later time, then at early times the event horizon area will have already slightly increased in anticipation of the later arrival of the matter. In other words, if one makes a firm decision to throw a baseball into a black hole, the Bekenstein--Hawking entropy of the black hole will have increased while the baseball is still in one's hand. Our definition of the entropy of a dynamical black hole does not have this property. Although the $A/4$ term in \eqref{eq:entropy_gr} will have increased while the baseball is still in one's hand, the horizon will also be expanding at that time, so the second term in \eqref{eq:entropy_gr} will contribute negatively. Theorem \ref{thm:phys_pro_ver_1st_law} shows that these two terms cancel at first order, so there is no change in the black hole entropy while the baseball is still in one's hand. The entropy of the black hole changes only when the baseball crosses the horizon.

\subsection{Vacuum Perturbations}
\label{vacpert}

If we have no external matter sources, $\delta T_{ab} = 0$, then Theorem \ref{thm:phys_pro_ver_1st_law} states that $\Delta \delta S = 0$, i.e., there is no change of black hole entropy with time at first order. Thus, we must go to second order in perturbation theory to obtain the leading order dynamical behavior of black hole entropy. To do so, we use the varied fundamental identity \eqref{eq:fund_id_var} with $\chi^a = \xi^a$. Since we assume that the vacuum equations of motion hold, the first three terms on the right side of \eqref{eq:fund_id_var} vanish\footnote{The first two terms would vanish in any case under pullback to $\mathcal H$.}. We obtain\footnote{Here $\delta^2$, the second variation, refers to the second derivative with respect to a parameter $s$ of a family 
of metrics $g^{(s)}_{ab} = g_{ab} + s\delta g_{ab} + \tfrac{1}{2} s^2 \delta^2 g_{ab} + O(s^3)$.}
\be
    \boldsymbol{\omega}(g, \delta g,\mathscr{L}_\xi \delta g)=  d [\delta^2 \textbf{Q}[g;\xi] - \xi \cdot \delta \boldsymbol{\theta}(g,\delta g)] \, .
\ee
We can rewrite this equation as 
\begin{align}
    \boldsymbol{\omega}(g,\delta g,\mathscr{L}_\xi \delta g) + d [\xi \cdot \delta \boldsymbol{\theta} (g, \delta g) - \xi \cdot \delta^2 \textbf{B}_{\mathcal H}(g))]={}&  d [\delta^2\textbf{Q}[g;\xi] - \xi \cdot \delta^2 \textbf{B}_{\mathcal H}(g)] =  \frac{\kappa}{2\pi} d \delta^2 \textbf{S}(g).
    \label{modce}
\end{align}
The term $\boldsymbol{\omega}(g,\delta g,\mathscr{L}_\xi \delta g)$ is known as the {\em canonical energy} $(n-1)$-form \cite{Hollands_Wald_2013} of the metric perturbation $\delta g$. We refer to entire left side of \eqref{modce} as the {\em modified canonical energy} $(n-1)$-form $\textbf{e}_G$,
\be
\textbf{e}_G(g;\delta g, \delta g) \equiv \boldsymbol{\omega}(g,\delta g,\mathscr{L}_\xi \delta g) + d [\xi \cdot \delta \boldsymbol{\theta} (g, \delta g) - \xi \cdot \delta^2 \textbf{B}_{\mathcal H}(g))].
\label{modce1}
\ee
Note that 
$\textbf{e}_G$ is quadratic in $\delta g_{ab}$ and does not depend on $\delta^2 g_{ab}$.
The modified canonical energy was introduced in \cite{Hollands_Wald_2013} for the case of general relativity as a means of obtaining a quantity with a locally positive flux whose total flux through the black hole horizon agrees with the total canonical energy flux. For the case of general relativity, the modified canonical energy flux is given by 
\be
\label{modceGR}
\textbf{e}_G = \frac{1}{4\pi}  (\kappa V)^2 \delta \sigma_{cd} \delta \sigma^{cd} \, \boldsymbol{\epsilon}^{(n-1)} ,
\ee
(see eq.~(105) of \cite{Hollands_Wald_2013}), where $V$ is an affine parameter, $\sigma_{cd}$ is the shear of the generators of the horizon with respect to $V$, and we have used the fact \cite{Hollands_Wald_2013} that $\delta \vartheta|_{\mathcal H}=0$ for vacuum perturbations. 
Although the modified canonical energy was introduced in an entirely ad hoc way in \cite{Hollands_Wald_2013}, we now see that it naturally enters the formula for the local change in dynamical black hole entropy at second order.

From \eqref{modce} we see that the situation for vacuum perturbations at second order corresponds to the situation for external matter perturbations at first order, with the modified canonical energy flux $\textbf{e}_G$ replacing the matter stress-energy flux $\textbf{e}$. In particular, for vacuum perturbations we have 
\be
\label{entropyflux}
\frac{\kappa}{2\pi} \Delta \delta^2 S = \int_{\mathcal{H}_{12}} \textbf{e}_G(g; \delta g, \delta g).
\ee
It follows immediately that, in an arbitrary theory of gravity, the second law of black hole mechanics will hold at second order for vacuum perturbations if and only if the modified canonical energy flux through the horizon is positive. This is the case in general relativity. However, it seems unlikely that this will be the case in more general theories of gravity even if we allow ourselves to modify the entropy making use of the permissible ambiguities (see appendix \ref{hkr} for further discussion).

\section{Relationship with the Dong--Wall Entropy}
\label{dwent}

In 2013, Dong \cite{Dong2013} proposed a formula for dynamical black hole entropy for theories obtained from a Lagrangian that depends upon the curvature but not derivatives of the curvature. In 2015, Wall \cite{Wall2015} proposed a formula for dynamical black hole entropy for general Lagrangian theories, valid to linear order for perturbations of a stationary black hole. Wall's approach for obtaining dynamical black hole entropy does not have any obvious relationship with Dong's. Nevertheless, for the case of a Lagrangian that depends upon the curvature but not its derivatives, Wall's formula was found to agree \cite{Wall2015} with the linearization of Dong's formula. In Appendix \ref{dongen}, we show that when we apply the ideas of Wall's method as elucidated in this section to calculate entropy in dilaton gravity models, it agrees with the linearization of the entropy obtained by Dong and Lewkowycz \cite{Dong2018} using Dong's approach. Thus, we believe that Wall's entropy agrees with the linearization of Dong's entropy in general. We will therefore refer to this entropy as the Dong--Wall entropy and will denote it by $S_{\rm DW}$. Nevertheless, it should be kept in mind that general agreement has not been proven. In the event that they do not agree, the entropy defined in this section would correspond to that proposed by Wall \cite{Wall2015}.

In this section, we will give a derivation of the Dong--Wall entropy based on the fundamental identity \eqref{eq:fund_id}. We will then 
obtain a simple, general relationship between the Dong--Wall entropy and our entropy, valid up to linear order in perturbations about a stationary black hole. As we shall see, the relationship between the Dong--Wall entropy and our entropy in general theories of gravity is very similar to the relationship \eqref{eq:entropy_gr} between the Bekenstein--Hawking entropy and our entropy in general relativity.

The starting point of our derivation of the Dong--Wall entropy is the formula \eqref{eq:constraints_eom} for the pullback of the linearized constraints to the horizon, which take the form 
\begin{align}
    \xi^a \delta {\textbf{C}}_a ={}& - 2 (\delta {E}_G)_{ab} \xi^a \xi^b \boldsymbol{\epsilon}^{(n-1)},
    \label{consform}
\end{align} 
where $\boldsymbol{\epsilon}^{(n-1)}$ was defined by \eqref{eq:or_conv_pt_1}. The linearized equations of motion $(\delta {E}_G)_{ab}$ take the general form
\begin{align}
    \delta {E}_G^{ab} ={}& \sum_{i=0}^{k} U_{(i)}^{ab c_1 \cdots c_i de}\nabla_{(c_1} \cdots \nabla_{c_i)} \delta g_{de}, 
    \label{EGform}
\end{align}
where the tensors $U_{(i)}^{ab c_1 \cdots c_i de} = U_{(i)}^{(ab) (c_1 \cdots c_i) (de)} $ are smooth and are locally and covariantly constructed from the metric, curvature, and covariant derivatives of the curvature. This is closely analogous to the quantity $\alpha^a$ in \eqref{eq:alpha_gen_form} except that there are two free indices on $\delta {E}_G^{ab}$ instead of the single free index on $\alpha^a$. Both of these free indices on $\delta {E}_G^{ab}$ are contracted into $\xi^a$ in \eqref{consform}, as opposed to the single contraction that occurs in \eqref{eq:alpha_xi_final_form}. Thus, one might expect that a strengthened form of Theorem \ref{prop:sym_pot_vanish_stat_bgd_pert} should hold for 
$(\delta {E}_G)_{ab} \xi^a \xi^b$. We show in Appendix \ref{sec:appendix_a} that this is the case. Specifically, we show that $(\delta {E}_G)_{ab} \xi^a \xi^b$ can be written in the form 
\begin{align}
 ( \delta {E}_G)_{ab} \xi^a \xi^b    \horeq{}& \sum_{i=0}^{k-2} \tilde{U}_{(i)}^{c_1 \cdots c_i de} \nabla_{(c_1} \cdots \nabla_{c_i)} (\mathscr{L}^2_\xi - \kappa \mathscr{L}_\xi) \delta g_{de} 
 \label{EGform2}
\end{align} 
where the tensors $\tilde{U}_{(i)}^{c_1 \cdots c_i de}$ are smooth on $\mathcal H$ and are locally and covariantly constructed from the metric, curvature, and its derivatives as well as from $\xi^a$, and $N^a$ (see \eqref{Nprop1} and \eqref{Nprop2}), with $\xi^a$ and $N^a$ appearing only algebraically. Thus, we see that $( \delta {E}_G)_{ab} \xi^a \xi^b$ has the form
\be
( \delta {E}_G)_{ab} \xi^a \xi^b \horeq \kappa\mathscr{L}_\xi P(g,\xi; \delta g) - \mathscr{L}_\xi^2 P(g,\xi; \delta g) = (\kappa - \mathscr{L}_\xi) \delta P(g,\xi; \mathscr{L}_\xi  g)
\label{delY}
\ee
with
\be
P(g,\xi; \delta g) = -  \sum_{i=0}^{k-2} \tilde{U}_{(i)}^{c_1 \cdots c_i de} \nabla_{(c_1} \cdots \nabla_{c_i)} \delta g_{de}.
\ee

We define the Dong--Wall entropy $(n-2)$-form by\footnote{Re-analyzing Wall's derivation \cite{Wall2015}, the authors of \cite{Bhattacharyya:2021jhr} have emphasized the fact that the entropy integrand has a current part, i.e. it is an $(n-2)$-form. However, we have not been able to find a  relationship between their proposed formula and our formula \eqref{eq:S_DW_form}.}
\begin{align}
    \textbf{S}_{\rm DW} ={}& \textbf{S} + \frac{4\pi}{\kappa} P(g,\xi; \mathscr{L}_\xi  g) \boldsymbol{\epsilon}^{(n-2)} 
    \label{eq:S_DW_form}
\end{align}
where our entropy $(n-2)$-form $\textbf{S}$ was defined by \eqref{entform} and $\boldsymbol{\epsilon}^{(n-2)}$ was defined in \eqref{eq:or_conv_pt_2}.
It is clear that $\textbf{S}_{\rm DW} = \textbf{S}$ in the stationary background, wherein both are given by \eqref{Sstat}.
Since $\delta \textbf{S}$ is uniquely defined for arbitrary first order perturbations and $\delta P$ is uniquely determined by \eqref{delY}, we see that $\delta \textbf{S}_{\rm DW}$ is uniquely defined for arbitrary first order perturbations. However, 
there will be ambiguities in $\textbf{S}_{\rm DW}$ at second and higher orders resulting from the ambiguities in $\textbf{S}$ discussed at the end of subsection \ref{dbedef} (see also appendix \ref{app:C}) as well as from the additional ambiguities in defining $P$ at second and higher orders. The Dong--Wall entropy of a cross-section $\mathcal C$ is defined by
\be
S_{\rm DW} [\mathcal C] = \int_{\mathcal C} \textbf{S}_{\rm DW}.
\ee
Since $\textbf{S}_{\rm DW}$ does not depend on a choice of cross-section and $\textbf{S}_{\rm DW}$ is smooth on ${\mathcal H}$, it follows immediately that $S_{\rm DW}$ is cross-section continuous, as defined in \cite{Chen2022}.

A simple relationship between $\textbf{S}_{\rm DW}$ and $\textbf{S}$ for first order perturbations can now be obtained as follows. As in the proof of Theorem \ref{thm:phys_pro_ver_1st_law}, we have
\begin{align}
\label{SxiC}
\frac{\kappa}{2\pi}d \delta \textbf{S} = -  \xi^a \delta {\textbf{C}}_a.
\end{align}
Using \eqref{consform} and \eqref{delY}, we obtain
\begin{align}
    \frac{\kappa}{2\pi}d\delta \textbf{S} ={}& 2 (\kappa - \mathscr{L}_\xi) \left[ \delta P(g,\xi; \mathscr{L}_\xi  g) \boldsymbol{\epsilon}^{(n-1)} \right] \, .
    \label{dSY}
\end{align}
Applying the general identity \eqref{eq:Lie_d_id} to $\mathscr{L}_\xi  [\delta P(g,\xi; \mathscr{L}_\xi  g) \boldsymbol{\epsilon}^{(n-1)}]$, we obtain
\be
\mathscr{L}_\xi  [\delta P \boldsymbol{\epsilon}^{(n-1)}] = d [\delta P \xi \cdot \boldsymbol{\epsilon}^{(n-1)}] + \xi \cdot d [\delta P  \boldsymbol{\epsilon}^{(n-1)}] = d [\delta P \boldsymbol{\epsilon}^{(n-2)}] 
\label{lieYid} 
\ee
where the second term in the first equality vanishes because $\boldsymbol{\epsilon}^{(n-1)}$ is an $(n-1)$-form on the $(n-1)$-dimensional manifold $\mathcal H$, so its ``$d$'' vanishes. Thus, bringing the $\mathscr{L}_\xi  \delta P$ term on the right side of \eqref{dSY} to the left side and using \eqref{lieYid} and the definition \eqref{eq:S_DW_form} of $\textbf{S}_{\rm DW}$, we obtain
\be
\frac{\kappa}{2\pi}d\delta \textbf{S}_{\rm DW} = 2\kappa  \delta P(g,\xi; \mathscr{L}_\xi  g) \boldsymbol{\epsilon}^{(n-1)} .
\label{sdwy}
\ee
Canceling the factors of $\kappa$ and contracting both sides with $\xi^a$, we obtain
\be
\frac{1}{2\pi} \xi \cdot d\delta \textbf{S}_{\rm DW} = 2 \xi \cdot [\delta P(g,\xi; \mathscr{L}_\xi  g) \boldsymbol{\epsilon}^{(n-1)}] = 2 \delta P(g,\xi; \mathscr{L}_\xi  g) \boldsymbol{\epsilon}^{(n-2)} .
\ee
By the general identity \eqref{eq:Lie_d_id}, we have
\be
\xi \cdot d\delta \textbf{S}_{\rm DW} = \mathscr{L}_\xi \delta \textbf{S}_{\rm DW} - d[\xi \cdot \delta \textbf{S}_{\rm DW}] .
\ee 
Thus, we obtain
\be
 4\pi \delta P(g,\xi; \mathscr{L}_\xi  g) \boldsymbol{\epsilon}^{(n-2)} = \mathscr{L}_\xi \delta \textbf{S}_{\rm DW} - d[\xi \cdot \delta \textbf{S}_{\rm DW}] .
\label{Ydwe}
\ee
By the definition \eqref{eq:S_DW_form} of $\textbf{S}_{\rm DW}$, we have
\be
\delta \textbf{S} = \delta \textbf{S}_{\rm DW}  - \frac{4\pi}{\kappa} \delta P(g,\xi; \mathscr{L}_\xi  g) \boldsymbol{\epsilon}^{(n-2)} 
\ee
and substituting in \eqref{Ydwe} we have\footnote{Note that we can combine 
\eqref{eq:S_DW_S_Lie_xi_rel} and \eqref{SxiC} into the equation 
\be
\label{DW:cond}
2\pi k^a \delta \boldsymbol{\textbf C}_a = d({\mathscr L}_k \delta \textbf{S}_{\rm DW})
\ee
to first order, 
where $\xi^a = \kappa V k^a$ and $V$ is an affine parameter vanishing on $\mathcal B$
such that $k^a \nabla_a V = 1$. This is identical in form to the defining equation (to first order) of the Dong--Wall entropy as given in \cite[eq. 120]{Hollands_2022}
if we replace $l^a$ there with $k^a$, and the entropy-current $(n-2)$-form $\boldsymbol{s}$ defined by \cite[eq. 119]{Hollands_2022} with 
$\boldsymbol{S}_{\rm DW}$, noting that $\boldsymbol{F}_l$ in \cite[eq. 120]{Hollands_2022} vanishes up to and including first order variations. Also note that $\boldsymbol{s}$ is defined 
only to first order by \cite[eq. 120]{Hollands_2022}, and it is shown in \cite[Prop. 1]{Hollands_2022} that it can be modified at second order if necessary so as to be covariant. By contrast, our derivation automatically delivers a covariant form of $\boldsymbol{S}_{\rm DW}$.}
\be
\delta \textbf{S} = \delta \textbf{S}_{\rm DW} - \frac{1}{\kappa}\mathscr{L}_\xi \delta \textbf{S}_{\rm DW} + \frac{1}{\kappa} d [\xi \cdot \delta \textbf{S}_{\rm DW}] .
\label{eq:S_DW_S_Lie_xi_rel}
\ee
The last term is exact, so it does not contribute when one integrates over a cross-section $\mathcal C$ to get the Dong--Wall entropy, $\delta S_{\rm DW} [\mathcal C]$, of the cross-section. Thus, we obtain the desired general relationship between our entropy and the Dong--Wall entropy for an arbitrary cross-section $\mathcal C$
\be
\delta S[\mathcal C] = \delta S_{\rm DW} [\mathcal C] - \frac{1}{\kappa} \int_{\mathcal C}  \mathscr{L}_\xi \delta \textbf{S}_{\rm DW}.
\label{ssdwrel}
\ee

It is worth noting that if we multiply \eqref{delY} by $\boldsymbol{\epsilon}^{(n-2)}$ and use \eqref{Ydwe} we obtain
\be
( \delta {E}_G)_{ab} \xi^a \xi^b \boldsymbol{\epsilon}^{(n-2)} \horeq  \frac{1}{4\pi} \left(\kappa - \mathscr{L}_\xi \right) \left[\mathscr{L}_\xi \delta \textbf{S}_{\rm DW} - d(\xi \cdot \delta \textbf{S}_{\rm DW}) \right].
\ee
Integrating over an arbitrary cross-section $\mathcal C$, we obtain
\be
\int_{\mathcal C} ( \delta {E}_G)_{ab} \xi^a \xi^b \boldsymbol{\epsilon}^{(n-2)} =   \frac{1}{4\pi} \int_{\mathcal C}  \left(\kappa - \mathscr{L}_\xi \right) \mathscr{L}_\xi \delta \textbf{S}_{\rm DW}.
\ee
Let $V$ be an affine parameter for null geodesic generators of $\mathcal H$ with $V=0$ corresponding to the bifurcation surface. Let $k^a$ be the corresponding affinely parametrized tangent to the null generators. Then $\xi^a = \kappa V k^a$ on $\mathcal H$ and on differential forms we have 
$\kappa \mathscr{L}_\xi  - \mathscr{L}_\xi^2  = - \kappa^2 V^2 \mathscr{L}^2_k $ modulo exact terms. We therefore obtain
\be
\int_{\mathcal C} V^2 ( \delta {E}_G)_{ab} k^a k^b \boldsymbol{\epsilon}^{(n-2)}  = - \frac{1}{4\pi}  \int_{\mathcal C}  V^2 \mathscr{L}^2_k \delta \textbf{S}_{\rm DW}.
\ee
If $\mathcal C$ is a cross-section of constant $V$ (which always can be achieved for any given cross-section using the rescaling freedom of $V$), then we obtain
\be
\int_{\mathcal C}  ( \delta {E}_G)_{ab} k^a k^b \boldsymbol{\epsilon}^{(n-2)} =  - \frac{1}{4\pi} \frac{d^2}{dV^2}  \delta S_{\rm DW} (V) .
\label{dwent2}
\ee
This equation is the defining property\footnote{Note that there is a missing factor of $1/2\pi$ on the left side of eq.~(7) (the first law) in \cite{Wall2015}. Putting this factor back in and taking two derivatives in affine parameter of that equation, one obtains eq.~(8) of \cite{Wall2015} with a factor of $2\pi$ on the left side instead of on the right side. To match with \eqref{dwent2}, note that the Wall's equations of motion tensor $H_{ab}$ (defined in eq.~(2) in \cite{Wall2015}) is given by $H_{ab} = 2 ({E}_G)_{ab}$. } of the entropy given by Wall (see eqs.~(7) and (8) of \cite{Wall2015}). 

For the case of general relativity (i.e.~using the Lagrangian \eqref{eq:gr_Lagrangian}),  we have
\be
( \delta {E}_G)_{ab} k^a k^b = \frac{1}{16\pi}\delta R_{ab} k^a k^b = -\frac{1}{16\pi} \mathscr{L}_k \delta \vartheta  
\ee
where the linearized Raychaudhuri equation was used in the second equality. Thus, we have
\be
( \delta {E}_G)_{ab} \xi^a \xi^b = -\frac{1}{16\pi} \kappa^2 V^2 \mathscr{L}_k \delta \vartheta = 
\frac{1}{16\pi} (\kappa - \mathscr{L}_\xi)(\kappa V \delta \vartheta).
\ee
Comparing with \eqref{delY}, we obtain
\be
\delta P(g,\xi; \mathscr{L}_\xi  g) = \frac{1}{16\pi} \kappa V \delta \vartheta.
\ee
Substituting this result in \eqref{eq:S_DW_form}, we find that for general relativity, we have
\be
\delta \textbf{S}_{\rm DW} = \delta \textbf{S} + \frac{1}{4} V \delta \vartheta \boldsymbol{\epsilon}^{(n-2)} .
\ee
Thus, in general relativity the term arising from $P$ in $\textbf{S}_{\rm DW}$ cancels the term arising from $\textbf{B}_{\mathcal H}$ in our formula \eqref{entform} for $\textbf{S}$, and we are left with only the contribution from $\textbf{Q}$, which yields the Bekenstein--Hawking entropy. Thus, the Dong--Wall entropy equals the Bekenstein--Hawking entropy in general relativity.

The type of cancellation between the $\textbf{B}_{\mathcal H}$ and $P$ terms that occurs in the calculation of the Dong--Wall entropy in general relativity does not occur for higher derivative theories. 
In particular, for the Lagrangian
\be
\textbf{L} = R_{ab}R^{ab}\boldsymbol{\epsilon},
\label{ric2}
\ee
the equations of motion are
\begin{align}
    \begin{split}
	({E}_G)^{ab}= {}& \tensor{R}{^{(a}_{c d e}} \Big(g^{b)d}R^{e c}  - R^{ b) e}g^{c d} \Big) \\
	{}&- \nabla_d \nabla_e \Big(g^{d (a}R^{b) e} - g^{de} R^{ab} -g^{ab} R^{de} + g^{e(a}R^{b)d} \Big) - \frac{1}{2} R_{de}R^{de} g^{ab}.
	\end{split}
\end{align} 
This can be put in the form
\be
    (\delta {E}_G)_{ab} \xi^a \xi^b = \kappa \mathscr{L}_\xi P(g,\xi; \delta g)  - \mathscr{L}_\xi^2 P(g,\xi; \delta g) 
\ee
with
\begin{align}
    \begin{split}
    P(g,\xi; \delta g) ={}& 2 \xi^n N^m \delta R_{mn} - \frac{1}{2} g^{bd}q^{mn}\nabla_m \nabla_n \delta g_{bd} + g^{bd} q^{mn}\nabla_m \nabla_b \delta g_{nd} \\
	{}&- R_{fmqe}N^f \xi^m N^q \xi^e  g^{bd} \delta g_{bd} - g^{bd}q^{mn}R_{fbme} \xi^e N^f  \delta g_{nd} - q^{mn}\tensor{R}{^f_{bme}}N^b \xi^e \delta g_{nf}\\
    {}& - 2N^d \xi^f q^{mn} \nabla_m \nabla_d \delta g_{nf} + N^m \xi^f q^{bd} \nabla_b \nabla_d \delta g_{fm} - q^{pn} R_{mnpe}N^m \xi^e g^{ab}\delta g_{ab} 
    \end{split} \label{eq:Y_Rab2_legal}
\end{align}  
where $q^{ab} = g^{ab} - 2\xi^{(a}N^{b)}$ is the orthogonal projector to $\xi^a$ and $N^a$. In this case, it is easily seen that contribution to $\textbf{S}_{\rm DW}$ from $P$ does not cancel the contribution of $\textbf{B}_{\mathcal H}$ to $\textbf{S}$ (see \eqref{BHR2} with $\alpha_1 = \alpha_3 = 0$ and $\alpha_2 = 1$).

We now evaluate the perturbed Dong--Wall entropy $\delta S_{\rm DW}[\mathcal C]$ for the Lagrangian \eqref{ric2}, evaluated on an arbitrary cross section $\mathcal C$.
Taking the Lie derivative of \eqref{eq:Y_Rab2_legal}, we find, after some algebra, that $\mathscr{L}_\xi P(g,\xi; \delta g)$ can be put in the form
\begin{align}
\begin{split}
\mathscr{L}_\xi P(g,\xi; \delta g) \horeq {}& 2\xi^n N^m \mathscr{L}_\xi \delta R_{mn} + 2 q^{mn}\nabla_m (\delta R_{nf}\xi^f) \\
{}& + \frac{1}{2} q^{mn}\nabla_m \nabla_n (g^{ab}\mathscr{L}_\xi \delta g_{ab}) - q^{pn} R_{mnpe} N^m \xi^e g^{ab}\mathscr{L}_\xi \delta g_{ab} + g^{ab}\mathscr{L}_\xi \delta g_{ab} R_{mn}\xi^m N^n.
\end{split}\label{eq:Y_R_ab_Lie_rewritten}
\end{align} 

We now substitute \eqref{eq:Y_R_ab_Lie_rewritten} and \eqref{eq:S_form_curv2} (with $\alpha_1 = \alpha_3 = 0$ and $\alpha_2 = 1$)
into \eqref{eq:S_DW_form} and integrate the resulting $(n-2)$-form over $\mathcal C$. The computations needed for this are most readily done using Gaussian null coordinates based upon the cross-section $\mathcal C$ and the Killing parameterization. The metric in these coordinates takes the form (see e.g. \cite{Hollands:2006rj})\footnote{ The Killing Gaussian null coordinates $v,r$ in \eqref{gnck} are related to the affine Gaussian null coordinates $V, \rho$ of \eqref{affgnc} by $V = e^{\kappa v}$, $\rho= \kappa^{-1}e^{-\kappa v} r$. The quantities $\beta_A$ and $\gamma_{AB}$ of \eqref{gnck} are the same quantities that appear in \eqref{affgnc}, but $\tilde{\alpha}$ is a different quantity from $\alpha$.}
\be
ds^2 = 2dvdr - 2r \tilde{\alpha} dv^2 - 2r\beta_A dvdx^A + \gamma_{AB} dx^A dx^B
\label{gnck}
\ee
with $\xi^a = (\partial/\partial v)^a$ and $\tilde{\alpha} \horeq \kappa$. Here $\mathcal H$ corresponds to $r=0$ and the cross-section $\mathcal C$ of $\mathcal H$ corresponds to $v=0$. 
We can express $N^a$ and $q^{ab}$ on $\mathcal H$ in terms of $\xi^a$, $l^a \equiv (\partial/\partial r)^a$, and $\beta_a = \beta_A (dx^A)_a$ as follows: 
\begin{align}
    N^a \horeq {}& l^a + \frac{1}{2\kappa}\beta^a - \frac{1}{8\kappa^2} \beta^p \beta_p \xi^a\\
    q^{ab} \horeq {}& \gamma^{ab} - \frac{1}{\kappa}\xi^{(a}\beta^{b)} + \frac{1}{4\kappa^2}\beta_p \beta^p \xi^a \xi^b \, .
\end{align} 
Using these relations for $N^a$ and $q^{ab}$ to rewrite \eqref{eq:Y_R_ab_Lie_rewritten}, we find 
\begin{align}
\begin{split}
\frac{4\pi}{\kappa} \mathscr{L}_\xi P(g,\xi; \delta g) \boldsymbol{\epsilon}^{(n-2)} ={}& 2\pi \mathscr{L}_\xi  \delta \Big[ \boldsymbol{\epsilon}^{(n-2)} \Big(4 \xi^m l^c  R_{mc}  - g^{ab}\mathscr{L}_\xi g_{ab} \gamma^{pq}\nabla_p l_q \Big) \Big]\\
{}&+ \frac{2\pi}{\kappa} {D}_q \Big(\boldsymbol{\epsilon}^{(n-2)} \Big[ 2  \delta R_{nf}\xi^f \gamma^{qn} - \frac{1}{\kappa}\beta^q \xi^m \xi^f \delta R_{mf} \\
{}& + \frac{1}{2}  D^q (g^{ab}\mathscr{L}_\xi \delta g_{ab}) -  \frac{1}{2}\beta^q g^{ab}\mathscr{L}_\xi \delta g_{ab}\Big] \Big) \, .
\end{split}  \label{eq:Y_Rab2_full_simp}
\end{align}
The last two lines are a total spatial divergence, which vanish when integrated over $\mathcal C$. 
The entropy $(n-2)$-form for the Lagrangian \eqref{ric2} is given by \eqref{eq:S_form_curv2} with $\alpha_1 = \alpha_3 = 0$ and $\alpha_2 = 1$, i.e.,
\begin{align}
\begin{split}
{S}_{a_1 \cdots a_{n-2}} ={}& -\frac{2\pi}{\kappa} \Big( {\epsilon}_{cda_1 \cdots a_{n-2}} \Big[ 2g^{ce} R^{df} \nabla_{[e}\xi_{f]} + (g^{cm} \nabla^d R - 2 \nabla^c R^{md})\xi_m\Big] \\
{}&+ {\epsilon}^{(n-2)}_{a_1 \cdots a_{n-2}} (R_{ad}\xi^a N^d g^{bc} + R^{bc}) \mathscr{L}_\xi g_{bc} \Big)\, .
\end{split}
\end{align} Taking a variation of this quantity, pulling it back to $\mathcal C$, and evaluating it in Gaussian null coordinates in a manner similar to what was done above for $P$, we find, after some algebra, that 
\begin{align}
\begin{split}
\delta \textbf{S} \horeq{}& 2\pi \Big( \delta \Big[ \boldsymbol{\epsilon}^{(n-2)} \Big( 4 R_{mn} \xi^m l^n - g^{ab} \mathscr{L}_\xi g_{ab} \gamma^{pq} \nabla_p l_q\Big) \Big] \\
{}&- \frac{1}{\kappa}\mathscr{L}_\xi \delta \Big[  \boldsymbol{\epsilon}^{(n-2)} \Big( 4 R_{mn}\xi^m l^n  - g^{ab} \mathscr{L}_\xi g_{ab} \gamma^{pq} \nabla_p l_q\Big) \Big] \\
{}& - \frac{2}{\kappa} {D}_m \Big[\boldsymbol{\epsilon}^{(n-2)}  \delta R_{dn}\xi^d \gamma^{mn} \Big] \Big)\, .
\end{split} \label{eq:S_Rab2_cs_simp}
\end{align}
Note that the form of this equation---where the second term is minus the 
Lie derivative of the first term and the third term is exact---is to be expected in view of  \eqref{eq:S_DW_S_Lie_xi_rel}. The last line of \eqref{eq:S_Rab2_cs_simp} vanishes when integrated over $\mathcal C$. Adding 
\eqref{eq:Y_Rab2_full_simp} and \eqref{eq:S_Rab2_cs_simp} and integrating over $\mathcal C$, we find that the Dong--Wall entropy is
\begin{align}
\label{Ricci2DW}
    \delta S_{\rm DW}[\mathcal C] ={}& 2\pi \delta \int_{\mathcal C} \boldsymbol{\epsilon}^{(n-2)} (4 R_{mc}\xi^m l^c - 2 K \bar K)
\end{align} where 
\begin{align}
    K ={}& \frac{1}{2}g^{ab}\mathscr{L}_\xi g_{ab} \\
    \bar K ={}& \gamma^{pq} \nabla_p l_q
\end{align} are the traces of the extrinsic curvatures in the $\xi$ and $l$ directions, respectively. This agrees with Wall's formula \cite{Wall2015} and with the linearization of Dong's formula \cite{Dong2013}.

Finally, we note that if we integrate \eqref{sdwy} between a cross-section $\mathcal C_1$ and another cross-section $\mathcal C_2$ lying to the future of $\mathcal C_1$, we obtain
\be
\delta S_{\rm DW}[\mathcal C_2] - \delta S_{\rm DW}[\mathcal C_1] = 4 \pi \int_{\mathcal H_{12}} \delta P \boldsymbol{\epsilon}^{(n-1)} 
\label{sdwy2}
\ee where $\mathcal H_{12}$ denotes the portion of $\mathcal H$ bounded by $\mathcal C_1$ and $\mathcal C_2$ and we have written $\delta P$ instead of $\delta P(g,\xi; \mathscr{L}_\xi  g)$ to make the notation less cumbersome. Thus, the Dong--Wall entropy will satisfy the second law of black hole mechanics at first order if and only if $\delta P (x) \geq 0$ for all $x \in {\mathcal H}$. However, from \eqref{delY}, we see that if the matter stress-energy tensor satisfies the null energy condition, then we have
\be
- e^{\kappa v} \mathscr{L}_\xi\left(e^{-\kappa v} \delta P \right) = (\kappa - \mathscr{L}_\xi) \delta P = ( \delta {E}_G)_{ab} \xi^a \xi^b = \frac{1}{2} ( \delta T_{ab}) \xi^a \xi^b\geq 0 \,
\label{delY2}
\ee where $v$ is the Killing parameter along the null geodesics (see \eqref{kilpar}).
Thus, we have
\be
\mathscr{L}_\xi\left(e^{-\kappa v} \delta P \right) \leq 0.
\ee
For any $x \in {\mathcal H}$, let $\gamma$ denote the null geodesic of $\mathcal H$ passing through $x$. Integrating this equation along $\gamma$ with respect to $v$ from $x$ to infinity and making the additional assumption that $\lim_{v \to \infty} e^{-\kappa v} \delta P = 0$, we obtain
\be
\delta P (x) \geq 0 \, .
\ee
Thus the Dong--Wall entropy also satisfies the second law of black hole mechanics. However, as with the Bekenstein--Hawking entropy as discussed at the end of subsection \ref{extmat}, the Dong--Wall entropy may increase in anticipation of matter that will be thrown into the black hole at a later time.

\section{The Generalized second law and QNEC}
\label{gslq}

The \textit{generalized second law} (GSL) of thermodynamics states that the sum of the entropy of a black hole and the entropy of matter outside the black hole never decreases with time. The GSL was originally proposed by Bekenstein \cite{Bekenstein_1973} in order to rescue the second law of thermodynamics in view of the fact that the total entropy of matter in the universe will decrease if matter falls into a black hole; by assigning an entropy to a black hole, Bekenstein argued that the total (``generalized'') entropy would be non-decreasing even if the entropy of matter decreases. Bekenstein's proposal was originally made within the context of classical black hole physics, where it is not consistent because, classically, black holes have vanishing thermodynamic temperature, but they would have to be assigned a non-zero temperature in order to satisfy the first law of thermodynamics with a finite entropy. However, shortly thereafter, Hawking \cite{Hawking_1975} discovered that within the context of semiclassical physics, black holes radiate thermally at a temperature $T = \kappa/2 \pi$. This not only removed a major inconsistency of Bekenstein's proposal but it made the GSL the only viable candidate for a thermodynamic second law, since quantum fields can violate all local energy conditions and the second law of black hole mechanics will be violated if there is a flux of negative energy into the black hole---as occurs during the Hawking radiation process. Thus, matter entropy and black hole entropy can individually decrease in various circumstances; only their sum has a possibility of being non-decreasing in all circumstances.

There is considerable circumstantial evidence in support of the GSL. In particular, gedankenexperiments attempting to violate the GSL fail \cite{Unruh_Wald_1982} in a manner reminiscent of the failure of gedankenexperiments designed to violate the ordinary second law.
Of course, one would like to obtain a precise formulation and proof of the GSL rather than just having circumstantial evidence in support of it. However, even in the context of non-gravitational physics, there is no precise formulation and proof of the ordinary second law of thermodynamics. Rather, there is a basic argument for the ordinary second law, which traces back to Boltzmann. This argument involves the introduction of a coarse-grained observable to define the entropy of matter. In the context of classical physics, the entropy of a state is defined to be the logarithm of the volume of phase space corresponding to the value of the coarse-grained observable in that state. If the entropy of a state is not at its maximum value, dynamical evolution is very likely to evolve that state to a state of higher entropy. If the entropy of a state is at its maximum, it is very likely to remain in a maximum entropy state for much longer than any realistic observation time. Thus, this argument suggests that if one starts a system in a ``special'' (i.e., a low entropy) state, one will observe an entropy increase until it reaches an ``equilibrium'' (i.e., maximum entropy) state, at which point one will observe no further changes. A similar argument applies in the quantum context, with the phase space volume of the value of the coarse-grained observable replaced by the dimension of the eigensubspace of the eigenvalue of the coarse-grained observable \cite{Wald1979}. It should be emphasized that the von Neumann entropy of the quantum state does not play any role in this argument for justifying the validity of the ordinary second law in the quantum context \cite{Wald1979}. Indeed, for any isolated system in any quantum state, the von Neumann entropy is constant, but one would still expect to observe an entropy increase in such a system if one started it in a ``special'' state. 

It is therefore quite remarkable that a precise formulation and argument/proof of the GSL has been given for semiclassical general relativity, using von Neumann entropy as the notion of entropy of matter outside the black hole. In particular, Wall \cite{Wall_2012} has shown that
\begin{align}
    \frac{d}{dV} \left( \frac{A}{4} + S_{\rm vN}  \right)  \geq {}&0.
\end{align}
Here $A(V)/4$ is the Bekenstein--Hawking entropy of the black hole on the cross-section $\mathcal C$ corresponding to affine time $V$ and $S_{\rm vN}(V)$ is the von Neumann entropy of the quantum field matter outside the black hole at a ``time'' corresponding to a spacelike hypersurface extending from $\mathcal C$ to spatial infinity. In fact, $S_{\rm vN}(V)$ is ill defined on account of the infinite entanglement of the quantum field across the horizon. However, von Neumann entropy differences between states should be well defined, and that is what should be needed to make sense of $d S_{\rm vN}/dV$ in the GSL (see \cite{Kudler-Flam:2023hkl} for a recent proposal on how to rigorously define entropy differences in quantum field theory via a cross product construction).

In this section, we will investigate the validity of the GSL for first order perturbations of a stationary black hole in a general theory of gravity, using both our definition of entropy, $S[\mathcal C]$, and the Dong--Wall definition, $S_{\rm DW} [\mathcal C]$. Our analysis will be based upon the {\em quantum null energy condition} (QNEC) for quantum fields in fixed classical background, first introduced in \cite{Bousso2015}. QNEC asserts that if $\mathcal C$ is a cross-section of a null hypersurface $\mathcal N$ such that the expansion and shear of $\mathcal N$ vanish on $\mathcal C$, then for any point $p$ on $\mathcal C$ we have \cite{Bousso2015}
\begin{align}
   \sqrt{h} \langle T_{ab} \rangle  k^a k^b   \geq{}& \frac{1}{2\pi}\frac{\delta^2}
    {\delta V(p)^2} S_{\rm vN} ,
    \label{eq:qnec}
\end{align} 
where $h$ is the determinant of the induced metric on $\mathcal C$, $V(y)$ is the affine parameter (a function of the coordinates in the transverse direction $y$), and $S_{\rm vN}$ is the von Neumann entropy of the quantum field ``outside'' the null hypersurface (where ``outside'' means on a spacelike hypersurface whose future Cauchy horizon coincides with $\mathcal N$). For simplicity, we will restrict consderation here to the case where the deformation of $\mathcal C$ is a uniform translation in $V$. Integration of \eqref{eq:qnec} over $\mathcal C$ yields
\be
\int_{\mathcal C} \langle T_{ab}\rangle k^a k^b  \geq \frac{1}{2 \pi} \frac{d^2}{d V^2} S_{\rm vN}.
\label{eq:qnec2}
\ee
QNEC has been rigorously proven/formulated in \cite{CeyhanFaulkner_2020} within the general framework of half-sided modular inclusions and relative entropy, and this setting is expected to cover the case of a bifurcate Killing horizon in curved spacetime as considered here. We omit a more precise technical discussion of the setting and, for purposes of readability, we stick to informal notations such as \eqref{eq:qnec2}. In the remainder of this section, we will assume that \eqref{eq:qnec2} holds for an arbitrary cross-section $\mathcal C$ of a bifurcate Killing horizon. 

Equation \eqref{eq:qnec2} holds in a fixed background spacetime. We assume that for an unperturbed stationary black hole with a quantum field in a Hartle--Hawking-like vacuum state, both sides of \eqref{eq:qnec2} vanish. We now consider perturbing the quantum state, taking 
back-reaction into account by treating $\delta \langle T_{ab}\rangle$ as an ``external matter perturbation'' as in subsection \ref{extmat}. Then we have
\be
2\delta ({E}_G)_{ab} k^a k^b = \delta \langle  T_{ab} \rangle k^a k^b.
\label{lineom2}
\ee 
Consequently, when back-reaction is taken into account, QNEC in the form \eqref{eq:qnec2} yields
\be
\int_{\mathcal C} 2\delta ({E}_G)_{ab} k^a k^b  \geq \frac{1}{2 \pi} \frac{d^2}{d V^2} \delta S_{\rm vN}.
\label{eq:qnec3}
\ee

The GSL for the Dong--Wall definition of black hole entropy can now be obtained as follows. Using \eqref{dwent2}, we obtain
\be
- \frac{1}{2 \pi} \frac{d^2}{d V^2} \delta S_{\rm DW} \geq \frac{1}{2 \pi} \frac{d^2}{d V^2} \delta S_{\rm vN} \, ,
\label{gsldw1}
\ee
i.e.,
\be
\frac{d^2}{d V^2} \left[\delta S_{\rm DW} + \delta S_{\rm vN} \right] \leq 0 \, .
\label{gsldw2}
\ee
Integrating this equation from an initial cross-section to $V = \infty$ and making the additional assumption that $\lim_{V \to \infty} d\delta S_{\rm DW}/dV = \lim_{V \to \infty} d\delta S_{\rm vN}/dV = 0$, we obtain the desired GSL
\be
\frac{d}{d V} \left[\delta S_{\rm DW} + \delta S_{\rm vN} \right] \geq 0 \, .
\label{gsldw3}
\ee
This derivation corresponds closely to Wall's derivation of the GSL for general relativity \cite{Wall_2012}, although he based his argument on the monotonicity of the relative entropy rather than QNEC. The above argument shows that this form of the GSL corresponds to an integrated form of QNEC.

The GSL for our definition of black hole entropy can be obtained by evaluating the physical process version of the first law \eqref{eq:phys_pro_ver_local_eqn} on cross-sections $\mathcal C$ of constant $V$ and differentiating with respect to $V$.
We obtain 
\be
\frac{\kappa}{2 \pi} \frac{d}{dV} \delta S = \int_{\mathcal C} \delta \langle T_{ab} \rangle k^a \xi^b = \kappa V \int_{\mathcal C} \delta \langle T_{ab} \rangle k^a k^b \geq \frac{\kappa}{2 \pi} V \frac{d^2}{d V^2} \delta S_{\rm vN} .
\label{gsl1}
\ee
This relation can be put in a GSL form by defining a {\em modified von Neumann entropy} of matter, $\tilde{S}_{\rm vN}$, by
\be
\tilde{S}_{\rm vN} \equiv {S}_{\rm vN} - V \frac{d}{dV} {S}_{\rm vN}.
\label{modvn}
\ee
Thus, $\tilde{S}_{\rm vN}$ agrees with ${S}_{\rm vN}$ in stationary eras but has a dynamical correction when ${S}_{\rm vN}$ changes with time. Since 
\be
\frac{d}{dV} \tilde{S}_{\rm vN} = - V \frac{d^2}{dV^2} {S}_{\rm vN},
\ee
we see that \eqref{gsl1} takes the GSL form\footnote{The idea of defining a modified von Neumann entropy by \eqref{modvn} so as to be able to write \eqref{gsl1} in the form \eqref{gsl2} was suggested to us by Jon Sorce. It would be interesting to see if this definition of modified von Neumann entropy can be motivated by additional, independent arguments. In this regard, Igor Klebanov has pointed out to us that the expression $V \frac{d}{dV} {S}_{\rm vN}$ is formally equal to the 
monotonic $c$-function for $2$-dimensional CFTs defined by \cite{Casini:2006es}. 
The lightray algebras of a Killing horizon have connections with a chiral CFT, so this might be more than a coincidence.}
\be
\frac{d}{dV} \left[ \delta S +  \delta \tilde{S}_{\rm vN} \right] \geq 0 \, .
\label{gsl2} 
\ee
This version of the GSL directly corresponds to QNEC, rather than an integrated form of QNEC.

\newpage

\begin{acknowledgments}
We thank Tom Faulkner, Igor Klebanov, Harvey Reall, and Jon Sorce for discussions.
SH and RMW thank the Erwin Schrödinger Institut in Vienna, where part of this work has been completed, for its hospitality and support. 
We thank Thomas Endler from MPI-MiS for creating Fig.~\ref{fig:1} based on our suggestions and for giving us permission to include it in this publication.
SH is grateful to the Max-Planck Society for supporting the collaboration between MPI-MiS and Leipzig U., grant Proj.~Bez.~M.FE.A.MATN0003. The research of RMW and VGZ was supported in part by  NSF grant PHY-2105878 and Templeton Foundation grant 62845 to the University of Chicago.
\end{acknowledgments}

\appendix

\section{$S$ is the Area of the Apparent Horizon to First Order in General Relativity}
\label{sec:area_app_horizon}

The apparent horizon on an arbitrary Cauchy surface $\Sigma$ is the boundary of the region containing outer trapped surfaces that lie within $\Sigma$. The apparent horizon is a marginally outer trapped surface, i.e., its outgoing null expansion vanishes. If the null energy condition holds and if weak cosmic censorship is valid, then an apparent horizon must lie within a black hole, so apparent horizons are very useful for determining the presence of a black hole given only data on $\Sigma$. For a stationary black hole, the apparent horizon on any Cauchy surface $\Sigma$ is given by the intersection of $\Sigma$ with the event horizon $\mathcal H$. 

For a nonstationary black hole, the location of an apparent horizon will depend on the choice of Cauchy surface, and the area of an apparent horizon at a ``time'' (i.e., choice of Cauchy surface) corresponding to a cross-section $\mathcal C$ of the horizon will not, in general, be a meaningful concept. However, for first order perturbations of a stationary black hole with bifurcate Killing horizon, the notion of the area of an apparent horizon at a time corresponding to $\mathcal C$ is a well defined concept. To see this, we note that,
as proven in section 2.1 of \cite{Hollands_Wald_2013}, for a stationary black hole with bifurcate Killing horizon, any horizon cross-section $\mathcal C$ is a strictly stably outermost marginally trapped surface in the terminology of \cite{Andersson_2005}. This implies that given any function, $f$, on $\mathcal C$, there exists a unique function, $\psi$ on $\mathcal C$ such that if we deform $\mathcal C$ infinitesimally by $\psi l^a$ (with $l^a$ a fixed vector field tranverse to the horizon), then to first order, the outgoing expansion of $\delta \mathcal C$ will be $f$. Now suppose we perturb the black hole so that its perturbed expansion at $\mathcal C$ is positive, $\delta \vartheta > 0$. To find the apparent horizon at first order, we infinitesimally perturb $\mathcal C$ inwards by $\psi l^a$ and denote the first order change of the outgoing null-expansion $\vartheta$ by $\delta_{\psi l} \vartheta$. If we now choose $\psi$ so that $\delta_{\psi l} \vartheta=-\delta \vartheta$, then to first order, this uniquely locates an apparent horizon lying along the flow of $\psi l^a$ from $\mathcal C$. To first order, the area of the apparent horizon will not depend on the choice of $l^a$.

The computation of the location of the apparent horizon is most conveniently done using Gaussian null coordinates based upon the cross-section $\mathcal C$ and the Killing parametrization, as previously introduced in \eqref{gnck}.
We take $l^a = (\partial/\partial r)^a$, so $l^a$ is the (past directed) ingoing null normal to $\mathcal C$, normalized so that $l^a \xi_a = 1$.  The perturbed expansion of $\mathcal C$ along the horizon direction with respect to $\xi^a$ is then $\delta \tilde \vartheta = \kappa V \delta \vartheta$. The equation for $\psi$ that we must solve is
\be
\delta_{\psi l} \tilde \vartheta = - \delta  \tilde{\vartheta} = - \kappa V \delta \vartheta.
\label{thetavar}
\ee
The variation $\delta_{\psi l} \tilde \vartheta =: W\psi$ is given by a linear operator $W$  intrinsic to $\mathcal C$ acting on $\psi$ which has been given in \cite{Andersson_2005} (see also eqs.~(5.12) and (5.34) of \cite{Engelhardt_Wall_2019}) as\footnote{The differences between formula \eqref{Wop} for our operator defined by $\delta_{\psi l} \tilde \vartheta =: W\psi$ and the corresponding eq.~(1) of \cite{Andersson_2005} result from the facts that (i) we are considering a null displacement rather than a spacelike displacement and (ii) they set $\tilde{\vartheta} = 0$ since they are considering a marginally outer trapped surface. Equation \eqref{Wop} is valid for the change in the expansion of the orthogonal null geodesics of any $n-2$ dimensional surface $\mathcal C$ under an infinitesimal displacement of the surface by $\psi l^a$, where $l^a$ is the other null normal field to $\mathcal C$, normalized to have unit inner product with the null normal defining the expansion.} 
\be
W\psi = -D^a D_a \psi + \beta^a D_a \psi + \Big(G_{ab}\xi^a l^b + \frac{1}{2}( R[\gamma]  - \frac{1}{2}\beta^a \beta_a + D_a \beta^a) - \tilde \vartheta \overline{\tilde \vartheta} \Big) \psi
\label{Wop}
\ee
where $R[\gamma]$ denotes the scalar curvature of the induced 
derivative $D_a$ on $\mathcal C$ associated with the induced metric $\gamma_{ab}$, 
$\tilde \vartheta =\frac{1}{2} \gamma^{cd} {\mathscr L}_\xi \gamma_{cd}$ is the outgoing expansion of $\mathcal C$, and $\overline{\tilde \vartheta}=\frac{1}{2} \gamma^{cd} {\mathscr L}_l \gamma_{cd}$ the ingoing expansion of $\mathcal C$. 

We now use Einstein's equation $G_{ab}=0$ and the fact that $\mathcal H$ has vanishing outgoing shear and expansion in the stationary background. Using the fact that $\tilde{\alpha} \horeq \kappa$ and using Einstein's equation again, the scalar curvature of $\mathcal C$ is found to be (see also eq.~(82) of \cite{Hollands:2006rj})
\begin{align}
R[\gamma] ={}& D_a \beta^a + \frac{1}{2}\beta^a \beta_a + 2\kappa 
\overline{\tilde \vartheta}.
\end{align}
Thus, $\psi$ is determined by
\be
-D^a D_a \psi + \beta^a D_a \psi + (D_a  \beta^a) \psi +  \kappa \overline{\tilde \vartheta} \psi = - \kappa V \delta \vartheta.
\ee
The second and third terms on the left side combine to a total divergence $D_a (\psi \beta^a)$ intrinsic to $\mathcal C$. Integrating over $\mathcal C$, we see that $\psi$ satisfies the relation 
\be
\kappa \int_{\mathcal C} \overline{\tilde \vartheta}  \psi \, \boldsymbol{\epsilon}^{(n-2)} = -\kappa \int_{\mathcal C}   V\delta \vartheta \, \boldsymbol{\epsilon}^{(n-2)}.
\ee
On the other hand, the change in area resulting from displacing $\mathcal C$ infinitesimally along $\psi l^a$ is
\be
\delta A = \int_{\mathcal C} \psi \overline{\tilde \vartheta} \, \boldsymbol{\epsilon}^{(n-2)}. 
\ee
Thus, we see that the area of the apparent horizon differs from the area of $\mathcal C$ by
\be
\delta A =  \int_{\mathcal C}    \psi \overline{\tilde \vartheta} \, \boldsymbol{\epsilon}^{(n-2)} = - \int_{\mathcal C}   V \delta \, \vartheta \boldsymbol{\epsilon}^{(n-2)}.
\ee
Comparing with \eqref{eq:entropy_gr}, we see that up to first order, we have
\be
S[\mathcal C] = \frac{A_{\rm app}[\mathcal C]}{4}
\ee
where $A_{\rm app}[\mathcal C]$ is the area of the apparent horizon at the time corresponding to the cross-section $\mathcal C$ of the horizon.

\section{Generalization of Theorem \ref{prop:sym_pot_vanish_stat_bgd_pert} to the Case of a Tensor Dotted with Two $\xi$'s} 
\label{sec:appendix_a}

Consider the hypotheses of Theorem \ref{prop:sym_pot_vanish_stat_bgd_pert}, but with $\alpha^a$ in \eqref{eq:alpha_gen_form} replaced by
\begin{align}
    \beta^{ab} ={}& \sum_{i=0}^{k} T_{(i)}^{ab c_1 \cdots c_i de}\nabla_{(c_1} \cdots \nabla_{c_i)} \delta g_{de}. \label{eq:alpha_gen_form2}
\end{align}
We wish to obtain a strengthening of the conclusion \eqref{eq:alpha_xi_final_form} of Theorem \ref{prop:sym_pot_vanish_stat_bgd_pert} that is valid for the form of $\beta_{ab} \xi^a \xi^b$. Specifically, we will show that $\beta_{ab} \xi^a \xi^b$ can be written in the form
\begin{align}
    \beta_{ab}  \xi^a \xi^b  \horeq{}& \sum_{i=0}^{k-2} \tilde{T}_{(i)}^{c_1 \cdots c_i de} \nabla_{(c_1} \cdots \nabla_{c_i)} (\mathscr{L}_\xi - \kappa)\mathscr{L}_\xi \delta g_{de} \label{eq:alpha_xi_final_form2}
\end{align}
where the tensors $\tilde{T}_{(i)}^{b_1 \cdots b_i cd}$ are smooth on $\mathcal H$ and are locally and covariantly constructed from the metric, curvature, and its derivatives as well as from $\xi^a$, and $N^a$, with $\xi^a$ and $N^a$ appearing only algebraically.

To show this, we start with the highest derivative term in $\beta_{ab} \xi^a \xi^b$ and proceed as in the proof of Theorem \ref{prop:sym_pot_vanish_stat_bgd_pert}. Each term in the basis expansion of $\xi_a \xi_b T_{(k)}^{ab c_1 \cdots c_i de}$ corresponding to \eqref{eq:T_j_basis_exp} will now have at least two ``extra'' $\xi$'s, i.e., there will be at least 2 more $\xi$'s than $N$'s in each term. We can perform the exactly the same manipulations with one of these ``extra'' $\xi$'s as was done to derive \eqref{kderform} in the proof of Theorem \ref{prop:sym_pot_vanish_stat_bgd_pert}. We thereby obtain
\be
\xi_a \xi_b T_{(k)}^{ab c_1 \cdots c_k de}\nabla_{(c_1} \cdots \nabla_{c_k)}\delta g_{de}  \horeq  \tilde{T}_{(k-1)}^{c_1 \cdots c_{k-1} de}\nabla_{(c_1} \cdots \nabla_{c_{k-1})} {\mathscr L}_ \xi \delta g_{de} + \sum_{i=0}^{k-1} {T'}_{(i)}^{c_1 \cdots c_i de}\nabla_{(c_1} \cdots \nabla_{c_i)} \delta g_{de}
\label{kderform2}
\ee
with $\tilde{T}_{(k-1)}^{c_1 \cdots c_{k-1} de}$ and ${T'}_{(i)}^{c_1 \cdots c_i de}$ satisfying the same properties as enumerated in Theorem \ref{prop:sym_pot_vanish_stat_bgd_pert} but with $\tilde{T}_{(k-1)}^{c_1 \cdots c_{k-1} de}$ now satisfying the additional property that there is at least one ``extra'' $\xi^a$ appearing in each term in its basis expansion. We may now perform the same manipulations with this additional $\xi^a$ to obtain
\be
\xi_a \xi_b T_{(k)}^{ab c_1 \cdots c_k de}\nabla_{(c_1} \cdots \nabla_{c_k)}\delta g_{de}  \horeq  \tilde{T}_{(k-2)}^{c_1 \cdots c_{k-2} de}\nabla_{(c_1} \cdots \nabla_{c_{k-2})} {\mathscr L}^2_ \xi \delta g_{de} + \sum_{i=0}^{k-1} {T''}_{(i)}^{c_1 \cdots c_i de}\nabla_{(c_1} \cdots \nabla_{c_i)} \delta g_{de}
\label{kderform3}
\ee
where ${T''}_{(i)}^{c_1 \cdots c_i de}$ incorporates all of the additional lower derivative terms arising from this second set of manipulations. Now, the left side of \eqref{kderform3} is $O(V^2)$ as $V \to 0$, where $V$ denotes an affine parameter on the horizon (with $V=0$ corresponding to the bifurcation surface). However, the first term on the right side is only $O(V)$ since 
\begin{align}
\tilde{T}_{(k-2)}^{c_1 \cdots c_{k-2} de}\nabla_{(c_1} \cdots \nabla_{c_{k-2})} {\mathscr L}^2_ \xi \delta g_{de} ={}& {\mathscr L}^2_ \xi \left[\tilde{T}_{(k-2)}^{c_1 \cdots c_{k-2} de}\nabla_{(c_1} \cdots \nabla_{c_{k-2})} \delta g_{de} \right] \nonumber \\
={}& \kappa V {\mathscr L}_ k \left(\kappa V {\mathscr L}_ k  \left[\tilde{T}_{(k-2)}^{c_1 \cdots c_{k-2} de}\nabla_{(c_1} \cdots \nabla_{c_{k-2})} \delta g_{de} \right] \right) \nonumber \\
={}& \kappa^2( V^2 {\mathscr L}^2_ k + V {\mathscr L}_ k)  \left[\tilde{T}_{(k-2)}^{c_1 \cdots c_{k-2} de}\nabla_{(c_1} \cdots \nabla_{c_{k-2})} \delta g_{de} \right] 
\end{align}
and the factor with $V {\mathscr L}_ k $ is only $O(V)$. We can rectify this by replacing ${\mathscr L}^2_ \xi \delta g_{de}$ in \eqref{kderform3} by $[{\mathscr L}^2_ \xi - \kappa {\mathscr L}_ \xi] \delta g_{de}$ and compensating for this by adding the quantity $\tilde{T}_{(k-2)}^{c_1 \cdots c_{k-2} de}\nabla_{(c_1} \cdots \nabla_{c_{k-2})} \kappa {\mathscr L}_ \xi \delta g_{de}$ to the second term on the right side of \eqref{kderform3}. Since this quantity involves at most $(k-1)$ derivatives of $\delta g_{de}$, it does not change the character of this term. Thus, have shown that
\be
\xi_a \xi_b T_{(k)}^{ab c_1 \cdots c_k de}\nabla_{(c_1} \cdots \nabla_{c_k)}\delta g_{de}  \horeq  \tilde{T}_{(k-2)}^{c_1 \cdots c_{k-2} de}\nabla_{(c_1} \cdots \nabla_{c_{k-2})} ({\mathscr L}_ \xi - \kappa) {\mathscr L}_ \xi \delta g_{de} + \sum_{i=0}^{k-1} {T'''}_{(i)}^{c_1 \cdots c_i de}\nabla_{(c_1} \cdots \nabla_{c_i)} \delta g_{de}
\label{kderform4}
\ee
where $\tilde{T}_{(k-2)}$ and ${T'''}_{(i)}$ are smooth and are locally and covariantly constructed from the metric, curvature, $\xi^a$, and $N^a$, with only algebraic dependence on $\xi^a$ and $N^a$. But now the left side and the first term on the right side are both $O(V^2)$ as $V \to 0$. Therefore, the last term on the right side is also $O(V^2)$ and therefore ${T'''}_{(k-1)}^{c_1 \cdots c_i de}$ must also have two ``extra'' $\xi$'s in each term of its basis expansion. This allows us to make an inductive argument to prove \eqref{eq:alpha_xi_final_form2} in parallel with the proof of Theorem \ref{prop:sym_pot_vanish_stat_bgd_pert}.

\medskip

 \begin{remark}
For the case of a tensor
\begin{align}
    \gamma^{a_1 \cdots a_p} ={}& \sum_{i=0}^{k} T_{(i)}^{a_1 \cdots a_p c_1 \cdots c_i de}\nabla_{(c_1} \cdots \nabla_{c_i)} \delta g_{de}, 
    \label{eq:alpha_gen_form3}
\end{align}
contracted into $p$ factors of $\xi^a$, the corresponding result is
\begin{align}
    \gamma_{a_1 \cdots a_p}  \xi^{a_1} \cdots \xi^{a_p}  \horeq{}& \sum_{i=0}^{k-p} \tilde{T}_{(i)}^{c_1 \cdots c_i de} \nabla_{(c_1} \cdots \nabla_{c_i)} (\mathscr{L}_\xi - \kappa(p-1)) \cdots (\mathscr{L}_\xi - 2\kappa)(\mathscr{L}_\xi - \kappa)\mathscr{L}_\xi \delta g_{de} 
    \label{eq:alpha_xi_final_form3} 
\end{align} where the tensors $\tilde{T}_{(i)}^{c_1 \cdots c_i de} $ are smooth on $\mathcal{H}$ and are locally and covariantly constructed from the metric, curvature, and its derivatives as well as from $\xi^a$, and $N^a$, with $\xi^a$ and $N^a$ appearing only algebraically.

 \end{remark}

\section{Ambiguities of $S$ Beyond First Order}
\label{app:C}
\subsection{Non-dynamical Background Structures}
\label{C1}

As discussed in section \ref{dbedef}, the definition of $S$ suffers from ambiguities at second order and beyond. 
One source of ambiguities arises from the presence of $\xi^a$ in our formula for $\textbf{B}_{\mathcal H}$ and the fact that,
going beyond linear order in perturbation theory, we can no longer treat $\xi^a$ as the exact Killing field on the stationary background. 
The approach taken in section \ref{dbedef} was to tie the metric $g_{ab}$ to $\xi^a$
by postulating in effect that $g_{ab}$ is in a Gaussian null gauge \eqref{affgnc} and is
related to $\xi^a$ by \eqref{xidef}, in such a way that $\mathcal B$ is at $\rho=0=V$.


In this section, we will present another possible approach wherein the class of metrics $g_{ab}$ considered (i.e. effectively the gauge near $\mathcal H$), as well as the form of $\xi^a$ off of $\mathcal H$ is determined by a non-dynamical background structure given only \emph{on $\mathcal H$,} as opposed to \emph{off of $\mathcal H$} as in \eqref{xidef}. In particular, off of $\mathcal H$, though not on $\mathcal H$, the $\xi^a$ considered in this appendix will depend on $g_{ab}$, unlike in section \ref{dbedef}. 

Our main reason for considering a different approach here is that with the approach given below, it will be possible to make a precise comparison between the ambiguities and structure of $S$ in our approach and that of recent proposal by \cite{Hollands_2022} (see section \ref{hkr}). In fact, while \cite{Hollands_2022}
also works in a Gaussian null gauge, just as we did to tie the gauge of $g_{ab}$ to $\xi^a$ in section \ref{dbedef}, the gauge considered by \cite{Hollands_2022} is tied to the cross section $\mathcal C$ on which one wants to define $S[{\mathcal C}]$. By contrast, in the present paper we have $\xi^a$ which is not tied from the outset to a particular $\mathcal C$ on which we anticipate evaluating $S$.

In order to characterize the ambiguities in $S$ precisely, it is necessary to specify the geometric structures underlying our construction.
The rigidly fixed, non-dynamical background structure considered in this appendix consists of (a) an $(n-1)$-dimensional hypersurface $\mathcal H$ of $\mathcal M$ having the topological structure of a trivial fibre bundle with fibres $\mathbb R$ and compact base; (b) a preferred cross section, $\mathcal B$,  on $\mathcal H$; (c) a vector field $\xi^a$ tangential to the fibres which changes its direction at $\mathcal B$ (and vanishes only on $\mathcal B$); (d) a positive real number $\kappa$; (e) a vector field $N^a$ transverse to $\mathcal H$ 
that is defined for each point in ${\mathcal H} \setminus {\mathcal B}$, and is non-zero where defined. 

\begin{definition}
\label{def:P}
Given our rigidly fixed, non-dynamical background structure $({\mathcal H}, {\mathcal B}, \xi^a, N^a, \kappa)$, we consider compatible spacetime metrics $g_{ab}$ on $\mathcal M$ in the following sense. We require that $\mathcal H$ is null with normal $\xi^a$ and, furthermore, is such that $\xi^a \nabla_a \xi^b = \kappa \xi^b,$  
$N^a N_a = 0$, and $g_{ab}\xi^b N^a = 1$ on $\mathcal H$. 
We call this space of metrics $\mathcal P$. $\textbf{B}_{\mathcal H}$, as well as $\textbf{S}$
are functionals on $\mathcal  P$.
\end{definition}

The space $\mathcal  P$ does not imply any physical restrictions onto the metric.
Indeed, let $\tilde{\mathcal H}$ be 
any null surface with compact cross sections ruled by affinely parameterized null geodesics 
with tangent $k^a$. Let $V$ be any affine parameter along the null geodesics vanishing on some cut, $\mathcal B$ and set $\tilde \xi^a = \kappa V k^a$. Then we have $\tilde \xi^a\nabla_a \tilde \xi^b=\kappa \tilde \xi^b$ 
and $\tilde \xi^a$ is obviously null and normal to $\tilde{\mathcal H}$. Furthermore, let $\tilde N^a$ be a second null field defined on $\tilde{\mathcal H}$ such that $\tilde \xi^a \tilde N_a = 1$. By applying a diffeomorphism 
$\phi$, we can clearly achieve that $\phi^* \tilde N^a = N^a, \phi^* \tilde \xi^a = \xi^a$ and $\phi[\tilde {\mathcal H}] = \mathcal H$. Thus, $\phi^* g_{ab}$ is in $\mathcal P$.

Due to the derivatives acting on $\xi^a$ in \eqref{BHform}, 
we first define $N^a$ off of $\mathcal H$ by demanding that, besides 
relation \eqref{naxi} and $N^a \xi_a = 1, N^a N_a = 0$ on $\mathcal H$, we 
have $N^a \nabla_a N^b = 0$.
Since $\xi^a$ is hypersurface orthogonal, we have $\nabla_{[a} \xi_{b]} \horeq 2\kappa w_{[a} \xi_{b]}$ for some $w_a$. In fact, we can and will adjust the extension 
such that $w_a = N_a$. To see this, note that we are allowed to change $\xi^a \to \xi^a + fp^a$ where $f$ is any function on $\mathcal M$ vanishing on $\mathcal H$ and $p^a$ is any vector field.
Since $\xi^a$ is normal to $\mathcal H$, we must have $\nabla_a f = c\xi_a$ on $\mathcal H$, where $c$ is a function that we can choose as we please. This means that $w_a \to w_a - cp_a/\kappa$ which we use to change $w_a$ to $N_a$.

Our requirement $\nabla_{[a} \xi_{b]} \horeq 2\kappa N_{[a} \xi_{b]}$  fixes the first derivative of $\xi^a$ on $\mathcal H$ only partially. Note that we are not free to fix the symmetrized derivative of $\xi^a$ e.g. by \eqref{Nprop1} since that relation used the Killing property of $\xi^a$, and we are no longer requiring at this stage that $\xi^a$ is a Killing vector field with respect to $g_{ab}.$ 
However, we may -- and will -- demand without loss of generality that $\xi^a$ has been extended off of $\mathcal H$  such that 
\begin{equation}
\label{naxi}
  \nabla_{[a} \xi_{b]} \horeq 2\kappa N_{[a} \xi_{b]}, \quad   (\nabla_a \xi_b) N^a N^b \horeq 0, \quad (\nabla_{(a} \xi_{b)}) \xi^b \horeq 0. 
\end{equation}
We are allowed to change $\xi^a \to \xi^a + f\xi^a$, where $f$ is any function on $\mathcal M$ vanishing on $\mathcal H$, without affecting the first relation in \eqref{naxi}.
Since $\nabla_a f = c\xi_a$ on $\mathcal H$, we can thereby change $\nabla_{(a} \xi_{b)} \to \nabla_{(a} \xi_{b)} + c\xi_a \xi_b$. Since 
$c$ is a function that we can choose as we please and since $N^a \xi_a = 1$, we can clearly impose $(\nabla_a \xi_b) N^a N^b \horeq 0$. 
The condition $(\nabla_{(a} \xi_{b)}) \xi^b \horeq 0$ easily follows from $\nabla_{[a} \xi_{b]} \horeq 2\kappa N_{[a} \xi_{b]}$ and 
$\xi^a \nabla_a \xi^b \horeq \kappa \xi^b$. Finally, using $N^a \nabla_a N^b = 0$ and $(\nabla_{(a} \xi_{b)}) \xi^b \horeq 0 \horeq (\nabla_a \xi_b) N^a N^b$, we have 
\begin{equation}
\label{firstorder}
N^a \nabla_a(N^b \xi_b) \horeq 0, \quad 
\nabla_a(\xi^b\xi_b) \horeq -2\kappa \xi_a. 
\end{equation}
Transvecting 
$\nabla_{[a} \xi_{b]} \horeq 2\kappa N_{[a} \xi_{b]}$ with $N^a$, we get ${\mathscr L}_N \xi_a \horeq -2\kappa N_a$, 
which clearly uniquely determines the first order change of $\xi^a$ off of $\mathcal H$.

As we have argued, \eqref{naxi} defines $\xi^a$ to first order off of $\mathcal H$, but we would like to define it to arbitrary orders off of $\mathcal H$. In order to have a definition that gives us the Killing vector field $\xi^a$ in the special case that $\mathcal H$ 
is a Killing horizon for $g_{ab}$, we require $\xi^a$ to satisfy the geodesic deviation equation, 
\begin{equation}
\label{eq:deviation}
    N^a \nabla_a (N^b \nabla_b \xi_c) = R_{cbad} N^b N^a \xi^d,
\end{equation}
which is a viewed as a second order ordinary differential equation for $\xi^a$ whose initial conditions are posed on $\mathcal H$. 
Since the geodesic deviation equation implies that $N^a \nabla_a (N^b \nabla_b(\xi^c N_c)) = 0$, the initial conditions $\xi^a N_a \horeq 1, N^c \nabla_c (N^a \xi_a) \horeq 0$, imply that 
$N^a \xi_a = 1$
in a neighborhood of $\mathcal H$.


Consider now a metric $\tilde g_{ab}$ having a stationary bifurcate Killing horizon $\mathcal H$ with Killing field $\tilde \xi^a$ that vanishes on $\mathcal B$. We have constructed in equations
\eqref{Nprop1} to \eqref{cform} a vector field $\tilde N^a$ at points of $\mathcal H$ from the Killing field $\tilde \xi^a$ and the metric $\tilde g_{ab}$. We can align these with $\xi^a, N^a$ as in definition \ref{def:P} by a diffeomorphism $\phi$, so that $g_{ab}=\phi^*\tilde g_{ab} \in \mathcal{P}$.
Since any Killing field satisfies \eqref{KVF}, we automatically have $N^a \nabla_a (N^b \nabla_b \xi_c) = R_{cbad} N^b N^a \xi^d$, where the derivatives and Riemann tensor refer to $g_{ab}$. Thus, our extension procedure produces precisely the horizon Killing field if there is one.


We likewise require that any variation $\delta g_{ab}$ be compatible 
with our the fixed (not varying) background structure $({\mathcal H}, {\mathcal B}, \xi^a, N^a, \kappa)$ in this appendix, in the sense that it is given as the derivative of a 1-parameter family of metrics in the space $\mathcal  P$ (definition \ref{def:P}), i.e. as an element in the tangent space of $\mathcal  P$. 
Then we have $\delta g_{ab} \xi^a = 0$ on $\mathcal H$. Furthermore, since $\delta \kappa = 0$ as $\kappa$ is part of the rigidly fixed background structure, the last equation in \eqref{firstorder} implies that $0=\delta[\nabla_a (\xi^b \xi_b)]=\nabla_a(\delta g_{bc}\xi^c\xi^b)$. Thus, the previously considered gauge conditions  \eqref{gaugecon1}, \eqref{gaugecon2} hold generally, not just off a stationary (bifurcate Killing-) horizon.

Our main first order variation formula for $S$ was \eqref{mainSformula} for any perturbation $\delta g_{ab}$
satisfying gauge conditions  \eqref{gaugecon1}, \eqref{gaugecon2}. In the construction leading to \eqref{mainSformula}, we previously considered $\xi^a$ as fixed not only on $\mathcal H$, but also off of $\mathcal H$. By contrast, the algorithm we have just given above for extending $\xi^a, N^a$ off of $\mathcal H$ depended on the chosen metric, and therefore while $\delta \xi^a \horeq 0$, the derivatives of $\delta \xi^a$ off of $\mathcal H$ will not in general vanish. Nevertheless, by a lengthy argument, we have shown that \eqref{mainSformula} continues to hold.  As a consequence we still have the formulas expressing the second law obtained within the previous framework of non-dynamical background structure in sections \ref{extmat} and \ref{vacpert}. Similarly, we still have the second order variation formulas \eqref{modce} and \eqref{modce1} in the present setting with a varying $\xi^a$, with only minor modifications to the definition of the boundary terms in the modified canonical energy \eqref{modce1}. However, we shall not present the details of these modifications here.

\subsection{Structure and Ambiguities of $S$}
\label{app:C3}
As in the approach taken in section \ref{dbedef},
there is substantial potential ambiguity in $S$ beyond first order.
The first ambiguity arises because it may happen that two metrics $g_{ab}$ and $\phi^* g_{ab}$ may both be in $\mathcal P$, where $\phi$ is a non-trivial diffeomorphism 
fixing $\mathcal H$ and hence $\xi^a$ on $\mathcal H$. We therefore have 
$(\phi^* g_{ab})\xi^a N^b = 1$ by definition \ref{def:P}, so $g_{ab}$ satisfies the definition \ref{def:P} also for $\tilde N^a = \phi^{-1*} N^a$. Since $g_{ab}$ and $\phi^* g_{ab}$ must be regarded as physically equivalent, this is telling us that 
we have, in effect, changed $N^a \to \tilde N^a$ while keeping $g_{ab}$ fixed, analogous to the freedom of choosing a different affine parameterization in of $\mathcal H$ in section \ref{dbedef}.

The second ambiguity arises from the fact that $\textbf{B}_{\mathcal H}$, which enters $S$, was constructed to satisfy \eqref{eq:sym_pot_tot_var}, which is a condition that applies only to first order perturbations. We may add many terms to $\textbf{B}_{\mathcal H}$ (and hence to $\textbf{S}=\textbf{Q} - \xi \cdot \textbf{B}_{\mathcal H}$) that vanish identically at first order but not second. The following lemma will enable us to give a precise classification of such terms.

\begin{lemma}
\label{lem:2}
    Let $\textbf{B}_{\mathcal H}$ be an $(n-1)$-form valued functional satisfying 1.--5. in section \ref{dbedef}. Then it can be written as a sum of monomials of the following factors i)--iii), contracted with $N^a, \xi^a, s^{a}_i$'s
    and multiplied by powers of $\kappa$. 
    \begin{itemize}
\item[i)] $s_{i_1}^{a_1} \cdots s_{i_j}^{a_j} \nabla_{(a_1} \cdots \nabla_{a_j)} 
{\mathscr L}_\xi g_{ab}$;
\item[ii)] $s_{i_1}^{a_1} \cdots s_{i_j}^{a_j} \nabla_{(a_1} \cdots \nabla_{a_{j-1})} ({\mathscr L}_\xi)^s N_{a_j}$, $s\ge 0$;
\item[iii)] $\nabla_{(a_1} \cdots \nabla_{a_j}R^b{}_{a_{j+1}a_{j+2})}{}^c, j>0$, $R_{abcd}$.
    \end{itemize}
Each contracted monomial in $\textbf{B}_{\mathcal H}$ is invariant under a rescaling $\xi^a \to c \xi^a$, $N^a \to N^a/c$, $\kappa \to c \kappa$. There must be at least one factor of type i), ii) with $s>0$ or a fully contracted factor of type iii) with more $\xi^a$'s than $N^a$'s. The total number of derivatives in each monomial must be less than the number of derivatives in the Lagrangian, with each occurrence of $\kappa$ counting as one derivative. 
\end{lemma}

\begin{proof}
By assumptions 1.--5. $\textbf{B}_{\mathcal H}$ is a sum of monomials 
consisting of the following factors: (a) $\nabla_{a_1} \dots \nabla_{a_j} N_b$, (b)
$\nabla_{a_1} \dots \nabla_{a_j} \xi_b$, (c) $\nabla_{a_1} \dots \nabla_{a_j} R_{abcd}$, which are fully contracted into the legs of the tetrad $N^a, \xi^a, s_i^a$. We will now show by induction in $j \ge 0$ that each such term can be converted into one of the terms i)--iii) listed in lemma 1. For term c)
it follows from the results of \cite{muller1997closed}.
Thus, we only have to deal with the terms (a) and (b). 
For $j=0$, there is nothing to show.

Assuming that the statement is true up to $j-1$ derivatives, let us prove it for $j$
derivatives. We begin by considering all possible ways of dotting the legs of our null tetrad into a term of the form (b), $\nabla_{a_1} \dots \nabla_{a_j} \xi_b$. By commuting derivative operators at the 
expense of Riemann tensors (giving terms of the form (b) with fewer derivatives and terms of the form (c)), and by moving any $N^a$ factors into the derivative operators at the expense of new terms of the form (a) with fewer than $j$ derivatives, we may bring the expression into the form $\nabla_{a_1} \dots \nabla_{a_{j-k}} (N^c \nabla_c)^k \xi_b$, dotted into $\xi^a$'s and $s_i^a$'s. For $k \ge 2$, 
we may use \eqref{eq:deviation} to lower $k$ at the expense of 
terms of the form (c), so only the case $k=0,1$ needs to be considered. As an example of such a term for $k=1$, consider
$s_{i_1}^{a_1} \dots s_{i_s}^{a_j} \xi^{a_{s+1}} \dots \xi^{a_{j-1}} \nabla_{a_1} \dots \nabla_{a_{j-1}} (N^c \nabla_c \xi_b)$. At the expense of terms of the form (a) when a derivative hits $N^c$, we can move all derivatives tangent to $\xi^d$ onto $\xi_b$ and then use $\xi^d \nabla_d \xi_b \horeq \kappa \xi_b$, since all derivatives are at this stage tangent to $\mathcal H$. Thus, we can effectively assume that 
the term we are dealing with is of the form
\begin{equation}
\begin{split}
\text{term (b)} &=s_{i_1}^{a_1} \dots s_{i_{j-1}}^{a_{j-1}} \nabla_{a_1} \dots \nabla_{a_{j-1}} (N^c \nabla_c \xi_b) \\
&=s_{i_1}^{a_1} \dots s_{i_{j-1}}^{a_{j-1}} \nabla_{a_1} \dots \nabla_{a_{j-1}} (N^c( 
\tfrac{1}{2} {\mathscr L}_\xi g_{cb} + 2\kappa N_{[c} \xi_{b]})).
\end{split}
\end{equation}
The first term on the right side is of the form (i) in the lemma, whereas the second has fewer than $j$ derivatives acting onto $\xi_b$ and at most $j$ derivatives acting onto $N^c$, i.e. it is a term of the form (a), after symmetrizing derivatives at the expense of terms of the form (c) with Riemann tensors. The case $k=0$ is treated similarly.

Next, we consider all possible ways of dotting the legs of our null tetrad into a term of the form (a), $\nabla_{a_1} \dots \nabla_{a_j} N_b$. By commuting derivative operators at the 
expense of Riemann tensors (giving terms of the form (a) with fewer derivatives and terms of the form (c)), and by moving any $N^a$ factors into the derivative operators at the expense of new terms of the form (a) with fewer than $j$ derivatives, we may bring the expression into the form $\nabla_{a_1} \dots \nabla_{a_{j-k}} (N^c \nabla_c)^k N_b$, dotted into $\xi^a$'s and $s_i^a$'s. Such a term vanishes for $k>0$ because of 
$N^c \nabla_c N_b=0$. We can move all derivatives tangent to $\xi^d$ onto $N_b$ at the expense of terms of the form (b) with fewer than $j$ derivatives. Thus, we can effectively assume that 
the term we are dealing with is of the form
\begin{equation}
\label{adef}
\text{term (a)} = s_{i_1}^{a_1} \dots s_{i_{j-s}}^{a_{j-s}} \xi^{a_{j-s+1}} 
\dots \xi^{a_j} \nabla_{a_1} \dots \nabla_{a_{j}} N_b .
\end{equation}
We can now gradually express the number of derivatives into the 
$\xi^a$-direction in terms of Lie derivatives ${\mathcal L}_\xi$
and terms such that each factor has fewer than $j$ derivatives. So, we 
have reduced the (a) terms to terms of the form 
\begin{equation}
\label{aadef}
\text{term (a)} = s_{i_1}^{a_1} \dots s_{i_{j-s}}^{a_{j-s}}  \nabla_{a_1} \dots \nabla_{a_{j-s}} ({\mathcal L}_\xi)^s N_b , 
\end{equation}
and we can symmtrize the $a_i$-indices at the expense of terms of the type (c). If we dot $\xi^b$ into this expression, we get terms of type (b), since $\xi^b N_b = 1$. If we dot $N^b$ into this expression, we get terms of the type (a) with fewer derivatives since $N^b N_b = 0$ in a neighborhood of $\mathcal H$. So we have effectively reduced attention to terms of the form claimed in item ii).
\end{proof}

The main consequence of this lemma is as follows: If $\mathcal H$ is a stationary bifurcate Killing horizon with Killing vector field $\xi^a$, then any factor of type i), ii) with $s>0$ or any fully contracted factor of type iii) with more $\xi^a$'s than $N^a$'s vanishes. By contrast, 
a factor of type ii) with $s=0$, or a fully contracted factor of type iii) with no more $\xi^a$'s than $N^a$'s factors will not, in general, vanish. It is thus clear that the ambiguity in $\textbf{B}_{\mathcal H}$ 
consists in adding monomials with at least \emph{two} factors of either type i), ii) with $s>0$, or a fully contracted factor of type iii) with more $\xi^a$'s than $N^a$'s.

\subsection{Comparison with HKR Proposal}
\label{hkr}

In \cite{Hollands_2022} Hollands, Kovacs, and Reall (HKR) constructed an entropy-current $(n-2)$-form $\textbf{S}_{\rm HKR}$
in an ``effective field theory'' (EFT) framework
\begin{equation}
    L=\frac{1}{8\pi}(2\ell^{-2}\Lambda + R + \ell^2 L_4 + \ell^4 L_6 + \dots),
\end{equation}
where each $L_{N}$ is a Lagrangian that is locally and covariantly constructed out of the metric and $N$ derivatives. EFT is understood to mean that one restricts attention to only those solutions in the theory truncated at some given order $N$ such that, roughly speaking, the higher order terms\footnote{The cosmological constant term is also considered ``higher order'', in the sense that $\Lambda \ll 1$, cf. the cosmological constant problem.} in the equations of motion are locally small near $\mathcal H$. $\textbf{S}_{\rm HKR}$ is by construction equal to the Dong--Wall entropy-current $(n-2)$-form \cite{Dong2013, Wall2015} $\textbf{S}_{\rm DW}$ to first order off a solution with a stationary bifurcate Killing horizon. \cite{Hollands_2022,Davies:2023qaa} showed that the ambiguities can also be exploited, for each $N$, so that 
$\textbf{S}_{\rm HKR}$ satisfies the 2nd law in the EFT sense, i.e. up to terms that are of the same order in $\ell$ as the terms neglected when truncating the EFT order $N$.

A comparison between the ambiguities in the HKR approach and our approach may now be given as follows: In
\cite[lem.~2.3]{Hollands_2022} it was shown that the admissible terms in $\textbf{S}_{\rm HKR}$ must be monomials in Gaussian null coordinate (GNC) components (see \eqref{affgnc}) of the following factors of total boost weight zero (see \cite[def. 2.2]{Hollands_2022}):
    \begin{itemize}
\item[i')] $D_{(a_1} \cdots D_{a_j)} K_{ab}$;
\item[ii')] $D_{(a_1} \cdots D_{a_{j-1})} \beta_{a_j}$, $D_{(a_1} \cdots D_{a_j)} \bar K_{ab}$;
\item[iii')] $\nabla_{(a_1} \cdots \nabla_{a_j}R^b{}_{a_{j+1}a_{j+2})}{}^c$, $R_{abcd}$;
    \end{itemize}
where $n^a = (\partial_\rho)^a$.
To compare with lemma \ref{lem:2}, we have
\begin{equation}
\label{NxisK}
\begin{split}
    N^a \horeq{}& \kappa^{-1}[\nabla^a \log V + \tfrac{1}{2} \beta^a - \tfrac{1}{8} V \beta^b \beta_b k^a],\\
    \xi^a \horeq & \kappa Vk^a,\\
    s_i^a \horeq {}& e_i^a - \tfrac{1}{2}Ve_i^b \beta_b k^a,\\
    \tfrac{1}{2} {\mathscr L}_\xi g_{ab} \horeq {}& \kappa VK_{ab} - \tfrac{1}{4}V^3 \kappa K_{cd}\beta^c\beta^d \, k_{a}k_b,
\end{split}      
\end{equation}
where $\sum_i e_i^a e_i^b = \gamma^{ab}$.
The scaling requirement in lemma \ref{lem:2} 
is equivalent to 
the zero total boost weight requirement\footnote{Here it must be understood that the boost weights of $\kappa, V$ be $+1,-1$ respectively. Note that in the HKR scheme, no explicit factors of $\kappa$ nor $V$ appear.} in \cite{Hollands_2022}. Furthermore, 
comparing i')--iii') in the light of \eqref{NxisK}
with the terms i)--iii) in lemma \ref{lem:2}, one can see some broad similarities between iii) and iii') and i) and i'), if we consider $x^A$  GNC components as analogous to $s^a_i$ tetrad components.  Items ii) and ii') are also similar, in the following sense. If in ii), we have $s=0$ (no Lie-derivative in the $\xi^a$-direction), then one can see from equations \eqref{NxisK} that we basically get a specific combination of the {\em two} types of terms in ii'). When $s \ge 1$ in the term ii), we also get certain combinations of the {\em two} types of terms in ii'). To see this, we first consider one Lie-derivative in the $\xi^a$-direction on $N_b$ \eqref{NxisK}, to get
\begin{equation}
    {\mathscr L}_\xi N_a \horeq \tfrac{1}{2} V \mathscr{L}_k \beta_a
    + \tfrac{1}{4} V^2(  \beta^c \mathscr{L}_k \beta_c +  K_{bc} \beta^b \beta^c ) k_a.  
\end{equation}
This already is of the type of terms of the form i'), ii') 
except for the term involving $\mathscr{L}_k \beta_a$. However, 
by equations given in \cite[sec.~A.1]{Hollands_2022}, we 
have 
\begin{equation}
    \mathscr{L}_k \beta_a +  K_{ab}\beta^b \horeq 2R_{bcda}n^b k^c k^d, 
\end{equation}
so $\mathscr{L}_\xi N_a$ is in fact a combination of the terms in i'), ii'), iii') not involving $\bar K_{ab}$. This remains true
by an inductive argument for $(\mathscr{L}_\xi)^s N_a$, and more generally for all our terms of type ii) in lemma \ref{lem:2}. 
Thus, our terms i), ii), iii) in lemma \ref{lem:2} correspond to specific combinations of the terms i'),ii'),iii'). 

Thus, we see that the ambiguities present in our approach are more restrictive than in the HKR approach. Correspondingly, our $S[\mathcal C]$ is more rigidly fixed than $S_{\rm HKR}[\mathcal C]$. We do not believe that we would have sufficient re-definition freedom in our approach to ensure that our entropy flux \eqref{entropyflux} satisfies the 2nd law beyond general relativity. On the other hand, it may still be possible to ensure that the second law holds for specific EFTs, which may thus provide an interesting selection criterion to discriminate unphysical EFTs.

It appears very unlikely that the more general freedom allowed by HKR will, in general, yield an entropy $S_{\rm HKR}[\mathcal C]$ that is cross section continuous \cite{Chen2022}.
To appreciate why, recall that, as shown in \cite[sec. 2.1]{Hollands_2022}, under a change of affine parameter $V=V'\psi(x^{A})$, where $\psi>0$, the quantities associated with a corresponding GNC system (see \eqref{affgnc}) change as
\begin{equation}
\label{GNChange}
\begin{split}
\beta'_a \horeq & \, \beta_a + 2D_a \log \psi -2V K_a{}^b D_b \log \psi, \\
\bar K_{ab}' \horeq & \, \psi^{-1}({\bar K}_{ab}-VD_a D_b \log \psi - V(D_a \log \psi)D_b \log \psi
- V\beta_{(a} D_{b)} \log \psi), \\
K_{ab}' \horeq & \, \psi K_{ab},
\end{split}
\end{equation} where $K_{ab}$ and $\bar K_{ab}$ are the extrinsic curvatures along $V$
respectively $\rho$.
Let $\mathcal C$ be a cut of constant $V$ and $\mathcal C'$ a cut of constant $V'$
and consider a $\psi$ that is extremely close to 1 but with extremely large angular derivatives (see fig. \ref{fig:1}). As an example, consider now the Ricci-squared theory, with Dong--Wall entropy given by \eqref{Ricci2DW}. It is easy to see that the first term in \eqref{Ricci2DW} is cross section covariant because the definition of the Ricci tensor does not involve any notion of foliation. Let us next focus on the $K\bar K$-term in \eqref{Ricci2DW}. Using \eqref{GNChange}, it is seen to change as 
\begin{equation}
    K'\bar K' \horeq K (\bar K  
-V' D_a D^a
 \psi - V' \beta^{a} D_{a} \psi 
 )
\end{equation}
under a change of affine parameter. Therefore, since $\gamma_{ab}' \horeq \gamma_{ab}$, we obtain
\begin{equation}
    \int_{{\mathcal C}'} K'\bar K'  \boldsymbol{\epsilon}'{}^{(n-2)}
= \int_{{\mathcal C}'}  K\left( \bar K  
-V'  D_a D^a \psi - V' \beta^{a} D_{a} \psi 
    \right)  \boldsymbol{\epsilon}^{(n-2)}.
\end{equation}
The terms in this last expression involving derivatives of $\psi$ threaten to blow up. However, these terms crucially are 
{\it linear} in $\psi$, and since $V'$ is constant on $\mathcal C'$, 
we may pull it out of the integral and then integrate the angular derivatives $D_a$ by parts. Hence, it is clear that also the $K\bar K$ term in the Dong--Wall entropy given by \eqref{Ricci2DW} is cross section continuous. On the other hand, it seems highly unlikely that 
$\textbf{S}_{\rm HKR}$ would only produce terms that are linear in $\psi$, since, in particular, it contains non-covariant terms starting at sufficiently high EFT order \cite{Davies:2022xdq}. By contrast, since our entropy is obtained from an entropy $(n-2)$-form that is covariant in the metric and $\xi^a$, it automatically yields a cross-section continuous entropy.

Of course, in an EFT framework, it is not sensible to consider surfaces, like $\mathcal{C}'$, that are wiggly on a scale that is comparable to or smaller than the cutoff scale $\ell$ of the EFT. More precisely, in the EFT setting, we fix a GNC system with respect to which the EFT assumptions \cite{Hollands_2022} are formulated, and we would not be allowed to perform a change of GNCs with a $\psi$ such that e.g., 
$\ell \bar K'$ as in \eqref{GNChange} failed to remain small. By the results of \cite{Hollands_2022} and their generalization \cite{Davies:2023qaa} on the second law for $S_{\rm HKR}$, if we do chose to impose such a restriction, this would preclude the possibility that 
reparameterization non-invariant terms could spoil the second law. However, one can see from \eqref{GNChange} that this would imply a restriction on how large $V$ can be, which is somewhat against the EFT spirit, since that is not a restriction on the UV behavior. If we merely insist that the wiggles in $\mathcal{C'}$ were on a scale considerably larger than $\ell$, i.e., if we merely require $|D_{a_1} \cdots D_{a_k} \psi| \ll O(\ell^{-k})$ but we allow $V$ to be arbitrarily large, then a violation of the second law of $S_{\rm HKR}$ 
\cite{Hollands_2022,Davies:2023qaa} appears to be possible due to reparameterization non-invariance.

\section{Application to Dilaton Gravity Theories with Scalar Fields}
\label{dongen}

Dong and Lewkowycz~\cite{Dong2018} have analyzed a dilaton gravity theory coupled to two scalar fields and have obtained an entropy formula for this theory using Dong's approach. In this appendix we show that the linearization of their entropy expression agrees with what would be calculated from \eqref{eq:S_DW_form}. This lends further support to the belief that Dong's entropy agrees in general with the entropy that would be obtained from Wall's approach.

The Lagrangian of dilaton gravity in two dimensions coupled with two scalar fields, $\sigma$ and $\omega$, considered in \cite{Dong2018} is
\begin{align}
    {L}_{ab} ={}& -\frac{1}{2}{\epsilon}_{ab} \Big[\phi R - (\nabla \sigma)^2 - (\nabla \omega)^2 + \lambda \omega \nabla_c \nabla_d \sigma \nabla^c \nabla^d \sigma   \Big],
\end{align}where $\lambda$ is a constant.
The equations of motion for $g_{ab}$ are
\begin{align}
    \begin{split}
		({E}_G)_{ab}={}& -\frac{1}{2}\Big\{ \nabla_a \sigma \nabla_b \sigma + \nabla_a \omega \nabla_b \omega + \nabla_a \nabla_b \phi - g_{ab}(g^{ed} \nabla_e \nabla_d \phi   +\frac{1}{2}[ (\nabla \sigma)^2  + (\nabla \omega)^2- \lambda \omega  (\nabla \nabla \sigma)^2])\\
		{}& + \lambda [ \nabla^m (\omega \nabla_m \nabla_b \sigma) \nabla_a \sigma + (a \leftrightarrow b) ]  -  \lambda \nabla^c (\omega \nabla_{a} \nabla_{b} \sigma \nabla_c \sigma ) \Big\} 
	\end{split}
\end{align} 
The pullback of the constraints to the horizon can be put in the form
\begin{align}
    \xi^a{{C}}_{ab_2} ={}& 2{\epsilon}^{(1)}_{b_2} [\kappa \mathscr{L}_\xi P - \mathscr{L}_\xi^2 P]
\end{align} where 
\begin{align}
    P(\psi, \xi; \delta \psi) ={}& \frac{1}{2}\kappa(\delta \phi - \lambda \omega \mathscr{L}_\xi \delta \sigma N^a \nabla_a \sigma)
\end{align} where $\psi$ collectively refers to the dynamical fields $g,\phi,\sigma,\omega$.

The Noether charge 0-form is
\begin{align}
    Q ={}& -\frac{1}{2}{\epsilon}_{ae}A^{aebd} \nabla_{[b} \xi_{d]} + \Big[{\epsilon}_{ae} \nabla_f A^{fbae} + {\epsilon}_{af} {B}^{a(fb)} \Big] \xi_b 
\end{align} where
\begin{align}
    A^{abcd}={}& - \frac{1}{2}\phi(g^{ac}g^{bd} - g^{ad}g^{bc}) \\
    \tilde{B}^{abc} ={}& \frac{1}{2} \nabla_d \phi (g^{ac}g^{bd}-g^{ad}g^{bc})  + \frac{1}{2}  \lambda \omega \nabla_{(m} \nabla_{n)} \sigma \nabla_t \sigma    (2 g^{ma}g^{nc} g^{bt}  -  g^{mb}g^{nc} g^{at} ) .
\end{align} 

The symplectic potential $\boldsymbol{\theta}$ is given by
\begin{align}
	\begin{split}
	{\theta}_{b_2} ={}& {\epsilon}_{ab_2}\Big(A^{abcd} \nabla_d \delta g_{bc} + \tilde{B}^{abc} \delta g_{bc} + D^a \delta \sigma + F^{af}\nabla_f \delta \sigma + M^a \delta \omega  \Big)
	\end{split}
\end{align} where
    \begin{align}
    D^a ={}& \nabla^a \sigma   +  \lambda  g^{me}g^{na} \nabla_e (\omega \nabla_{(m} \nabla_{n)} \sigma )   \\
    F^{af} ={}& -\lambda \omega  g^{ma}g^{nf} \nabla_{(m} \nabla_{n)} \sigma \\
    M^a ={}& \nabla^a \omega 
    \end{align}
The pullback of $\boldsymbol{\theta}$ to the horizon takes the form
\begin{align}
    \underline{\boldsymbol{\theta}} \horeq{}& \delta \textbf{B}_{\mathcal{H}}
\end{align}    
where
\be    
  {B}_{\mathcal{H}b} = \frac{1}{2}{\epsilon}^{(1)}_{b} \Big( 2 \lambda \omega N^{a}  \nabla_{a} \sigma \mathscr{L}_\xi  \sigma \Big) 
\ee

Our entropy 0-form ${S}$ (see \eqref{entform}) is given by
\begin{align}
    S ={}& \frac{2\pi}{\kappa}(Q - \xi \cdot \textbf{B}_\mathcal{H} ) \\
    ={}& \frac{2\pi}{\kappa}[ -\kappa \phi  + \kappa\lambda \omega N^a \nabla_a \sigma \mathscr{L}_\xi \sigma -\mathscr{L}_\xi (- \phi + \lambda \omega N^a \nabla_a \sigma \mathscr{L}_\xi^2 \sigma ) ]
\end{align} 
The Dong--Wall entropy (see \eqref{eq:S_DW_form}) is
\begin{align}
    S_{\rm DW} ={}& \frac{2\pi}{\kappa}(Q - \xi \cdot \textbf{B}_\mathcal{H} ) - \frac{4 \pi}{\kappa} P(\psi, \xi; \mathscr{L}_\xi \psi)  \\
    ={}& 2\pi[- \phi + \lambda \omega N^a \nabla_a \sigma \mathscr{L}_\xi \sigma ], 
\end{align} 
This agrees\footnote{To compare, one must take into account---as one also must do in comparing eq.~(14) of \cite{Wall2015} with eq.~(1.3) of \cite{Dong2013}---that Dong works in Euclidean signature and uses a complex $z$ coordinate.} with the linearization of (A.44) of \cite{Dong2018}.


\bibliography{dyn_bh_ent}

\end{document}